\documentclass[12pt,  draftclsnofoot,onecolumn,article,romanappendices]{IEEEtran}
\usepackage{setspace}
\usepackage{cite}
\ifCLASSINFOpdf
   \usepackage[pdftex]{graphicx}
\else
  \usepackage[dvips]{graphicx}
\fi
\usepackage{amsmath,amssymb,amsthm,mathrsfs,amsfonts,dsfont,stmaryrd}
\usepackage{array}
\usepackage{dsfont}
\usepackage{amsbsy}
\usepackage{epstopdf}
\usepackage{graphicx,color}
\usepackage{hyperref}

\newtheorem{lemma}{Lemma}
\newtheorem{proposition}{Proposition}
\newtheorem{remark}{Remark}
\newtheorem{corollary}{Corollary}
\long\def\symbolfootnote[#1]#2{\begingroup%
\def\thefootnote{\fnsymbol{footnote}}\footnote[#1]{#2}\endgroup}
\newtheorem{theorem}{Theorem}
\newtheorem{definition}{Definition}
\newtheorem{example}{Example}

\newcommand{\dv}{\mathbf} % determenistic vector
\newcommand{\mc}{\mathcal} % determenistic vector
\newcommand{\mkv}{-\!\!\!\!\minuso\!\!\!\!-}

 \allowdisplaybreaks

\usepackage{algorithm}
\usepackage{algpseudocode}

\algnewcommand{\Inputs}[1]{%
  \State \textbf{Inputs:}
  \Statex \hspace*{\algorithmicindent}\parbox[t]{.8\linewidth}{\raggedright #1}
}
\algnewcommand{\Initialize}[1]{%
  \State \textbf{initialization}
  \Statex \hspace*{\algorithmicindent}\parbox[t]{.95\linewidth}{\raggedright #1}
}

\begin{document}
%\fontencoding{OT1}\fontsize{10}{11}\selectfont
%\setlength{\columnsep}{0.15 in}

\title{On the Capacity of Cloud Radio Access Networks with Oblivious Relaying }
\author{\small I\~naki Estella Aguerri \qquad \quad Abdellatif Zaidi   \quad \quad Giuseppe Caire  \quad \quad Shlomo Shamai (Shitz)\\
\thanks{The results in this paper have been partially presented in \cite{EZSC:ISIT:2017}. Inaki Estella Aguerri is with the Mathematical and Algorithmic Sciences Lab, Paris Research Center, Huawei Technologies, 92100 Boulogne-Billancourt, France. Abdellatif Zaidi is with Universit\'e Paris-Est, France, and currently on leave at the Mathematical and Algorithmic Sciences Laboratory, Huawei Paris Research Center, 92100 Boulogne-Billancourt, France. Giuseppe Caire is with the Technische Universit\"at Berlin, 10587 Berlin, Germany.  Shlomo Shamai (Shitz) is with the Technion Institute of Technology, Technion City, Haifa 32000, Israel. The work of G. Caire is supported by an Alexander von Humboldt Professorship. The work of S. Shamai has been supported by the European Union's Horizon 2020 Research And Innovation Programme, grant agreement no. 694630, and partly by the US-Israel, Binational Science Foundation (BSF). Emails. \{inaki.estella@huawei.com, abdellatif.zaidi@u-pem.fr, caire@tu-berlin.de, sshlomo@ee.technion.ac.il\}. }}

% make the title area
\maketitle
\begin{abstract}
We study the transmission over a network in which users send information to a remote destination through relay nodes that are connected to the destination via finite-capacity error-free links, i.e., a cloud radio access network. The relays are constrained to operate without knowledge of the users' codebooks, i.e., they perform oblivious processing. The destination, or central processor, however, is informed about the users' codebooks. We establish a single-letter characterization of the capacity region of this model for a class of discrete memoryless channels in which the outputs at the relay nodes are independent given the users' inputs. We show that both relaying \`a-la Cover-El Gamal, i.e.,  compress-and-forward with joint decompression and decoding, and ``noisy network coding'', are optimal. The proof of the converse part establishes, and utilizes,  connections with the Chief Executive Officer (CEO) source coding problem under logarithmic loss distortion measure.
Extensions to general discrete memoryless channels are also investigated. In this case, we establish inner and outer bounds on the capacity region. For memoryless Gaussian channels within the studied class of channels, we characterize the capacity region when the users are constrained to time-share among Gaussian codebooks. Furthermore, we also discuss the suboptimality of separate decompression-decoding and the role of time-sharing.
% Furthermore, we study the related distributed information bottleneck problem and characterize optimal tradeoffs between rates (i.e., complexity) and information (i.e., accuracy) in the vector Gaussian case.
\end{abstract}

% The paper headers
%\markboth{IEEE Transactions on Communications}%
%{Submitted paper}

% keywords
%\sqrt{•}
\IEEEpeerreviewmaketitle

%%%%%%%%%%%%%%%%%

\section{Introduction}

Cloud radio access networks (CRAN) provide a new architecture for next-generation wireless cellular systems in which base stations (BSs) are connected to a cloud-computing central processor (CP) via error-free finite-rate fronthaul links. This architecture is generally seen as an efficient means to increase spectral efficiency in cellular networks by enabling joint processing of the signals received by multiple BSs at the CP and, so, possibly alleviating the effect of interference. Other advantages include low cost deployment and flexible network utilization~\cite{Peng2015FHConstCRAN}.

In a CRAN network, each BS acts essentially as a relay node; and so it can in principle implement any relaying strategy, e.g., decode-and-forward~\cite[Theorem 1]{Cover:1979}, compress-and-forward~\cite[Theorem 6]{Cover:1979} or combinations of them. Relaying strategies in CRANs can be divided roughly into two classes: i) strategies that require the relay nodes to know the users' codebooks (i.e., modulation, coding),  such as decode-and-forward, compute-and-forward~\cite{Nazer:IT:2011, HongCaire:IT:2013,Nazer:ISIT:2009} or variants thereof, and ii) strategies in which the relay nodes operate without knowledge of the users' codebooks, often referred to as \textit{oblivious relay processing} (or \textit{nomadic transmission})~\cite{Sanderovich2008:IT:ComViaDesc, Simeone2011:IT:CodebookInfoOutBand,Sanderovich:2009:IT}. This second class is composed essentially of strategies in which the relays  implement forms of compress-and-forward~\cite{Cover:1979}, such as successive Wyner-Ziv compression~\cite{Hwan:2013:VT,ZhouYu:2013:JSAC, DBLP:journals/corr/ZhouX0C16} and quantize-map-and-forward~\cite{ADT11} or noisy-network coding \cite{Lim:IT:2011:NoisyNetwork}. Schemes that combine the two approaches have been shown to possibly outperform the best of the two~\cite{Estella2016CQCF}, especially in scenarios in which there are more users than relay nodes.

In essence, however, a CRAN architecture is usually envisioned as one in which BSs operate as simple radio units (RUs) that are constrained to implement only radio functionalities such as analog-to-digital conversion and filtering while the baseband functionalities are migrated to the CP. For this reason, while relaying schemes that involve partial or full decoding of the users' codewords can sometimes offer rate gains, they do not seem to be suitable in practice.  In fact, such schemes assume that all or a subset of the relay nodes are fully aware (at all times!) of the codebooks and encoding operations used by the users. For this reason, the signaling required to enable such awareness is generally prohibitive, particularly as the network size gets large. Instead, schemes in which relay nodes perform oblivious processing are preferred in practice. Oblivious processing was first introduced in~\cite{Sanderovich2008:IT:ComViaDesc}. The basic idea is that of using \textit{randomized encoding} to model lack of information about codebooks.  For related works, the reader may refer to~\cite{Simeone2011:IT:CodebookInfoOutBand, Dytso2014} and~\cite{Tian2011}. In particular,~\cite{Simeone2011:IT:CodebookInfoOutBand} extends the original definition of oblivious processing of \cite{Sanderovich2008:IT:ComViaDesc}, which rules out time-sharing, to include settings in which transmitters are allowed to switch among different codebooks, constrained relay nodes are unaware of the codebooks but are given, or can acquire, time- or frequency-schedule information\footnote{Typically, this information is small, e.g., 1 bit that captures on/off activity; and, so, obtaining it is generally much less demanding that obtaining full information about the users' codebooks.}. The framework is referred to therein as ``\textit{oblivious processing with enabled time-sharing}".

%----------------------------------------
\begin{figure}[t!]
\centering
\includegraphics[width=0.85\textwidth]{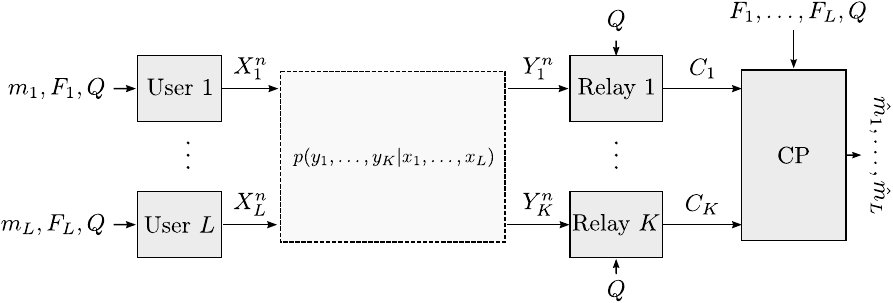}
\vspace{-5mm}
\caption{CRAN model with oblivious relaying and time-sharing.}
\label{fig:Schm}
\end{figure}
%----------------------------------------

In this work, we consider transmission over a CRAN in which the relay nodes are constrained to operate without knowledge of the users' codebooks, i.e., are oblivious, and only know time- or frequency-sharing information. The model is shown in Figure~\ref{fig:Schm}. Focusing on a class of discrete memoryless channels in which the relay outputs are independent conditionally on the users' inputs, we establish a single-letter characterization of the capacity region of this class of channels. We show that both relaying \`a-la Cover-El Gamal, i.e., compress-and-forward with joint decompression and decoding\cite{Sanderovich2008:IT:ComViaDesc,Park:2013:SPLett}, and noisy network coding \cite{Lim:IT:2011:NoisyNetwork} are optimal. For the proof of the converse part, we utilize useful connections with the Chief Executive Officer (CEO) source coding problem under logarithmic loss distortion measure~\cite{Courtade2014LogLoss}.  Extensions to general discrete memoryless channels are also investigated. In this case, we establish inner and outer bounds on the capacity region. For memoryless Gaussian channels within the studied class, we provide a full characterization of the capacity region under \textit{Gaussian signaling}, i.e., when the users' channel inputs are restricted to be Gaussian. In doing so, we also investigate the role of time-sharing.

\subsection*{Outline and Notation}

The rest of this paper is organized as follows. Section II provides a formal description of the model, as well as some definitions that are related to it. Section III contains the main result of this paper, which is a single-letter characterization of the capacity region of a class of discrete memoryless CRANs with oblivious processing at relays and enabled time-sharing in which the channel outputs at the relay nodes are independent conditionally on the users' channel inputs. This section also provides inner and outer bounds on the capacity region of general discrete memoryless CRANs with constrained relays, as well as some discussions on the suboptimality of successive decompression and decoding and the role of time-sharing. Finally, in Section IV, we study a memoryless vector Gaussian CRAN model with oblivious processing at relays and enabled time-sharing, for which we characterize the capacity region under Gaussian signaling.
% Finally, in Section V, we characterize the rate-information region of the vector Gaussian distributed information bottleneck problem.

Throughout this paper, we use the following notation. Upper case letters are used to denote random variables, e.g., $X$; lower case letters are used to denote realizations of random variables $x$; and calligraphic letters denote sets, e.g., $\mathcal{X}$.  The cardinality of a set $\mc X$ is denoted by $|\mc X|$. The length-$n$ sequence $(X_1,\ldots,X_n)$ is denoted as  $X^n$; and, for integers $j$ and $k$ such that $1 \leq k \leq j \leq n$, the sub-sequence $(X_k,X_{k+1},\ldots, X_j)$ is denoted as  $X_{k}^j$. Probability mass functions (pmfs), are denoted by $p_X(x)=\mathrm{Pr}\{X=x\}$; or for short, as $p(x)=\mathrm{Pr}\{X=x\}$. Boldface upper case letters denote vectors or matrices, e.g., $\dv X$, where context should make the distinction clear. For an integer $L \geq 1$, we denote the set of integers smaller or equal $L$ as $\mc L := \{ l \in \mathbb{N} :  1 \leq l \leq L\}$. Sometimes, this set will also be denoted as $[1\!:\!L]$. For a set of integers $\mc K \subseteq \mc L$, the notation $X_{\mc K}$ designates the set of random variables $\{X_k\}$ with indices $k$ in the set $\mc K$, i.e., $X_{\mc K}=\{X_k\}_{k \in \mc K}$.
 We denote the covariance of a zero mean vector $\mathbf{X}$ by $\mathbf{\Sigma}_{\mathbf{x}}:=\mathrm{E}[\mathbf{XX}^H]$;  $\mathbf{\Sigma}_{\mathbf{x},\mathbf{y}}$ is the cross-correlation  $\mathbf{\Sigma}_{\mathbf{x},\mathbf{y}}:= \mathrm{E}[\mathbf{XY}^H]$, and the conditional correlation matrix of $\mathbf{X}$ given $\mathbf{Y}$ as $\mathbf{\Sigma}_{\mathbf{x}|\mathbf{y}}:= \mathbf{\Sigma}_{\mathbf{x}}-\mathbf{\Sigma}_{\mathbf{x},\mathbf{y}}\mathbf{\Sigma}_{\mathbf{y}}^{-1}\mathbf{\Sigma}_{\mathbf{y},\mathbf{x}}$.

\vspace{-3mm}
\section{System Model}\label{sec:System}

Consider the discrete memoryless (DM) CRAN model shown in Figure~\ref{fig:Schm}. In this model, $L$ users communicate with a common destination or central processor (CP) through $K$ relay nodes, where $L \geq 1$ and $K \geq 1$. Relay node $k$, $ 1 \leq k \leq K$, is connected to the CP via an error-free finite-rate fronthaul link of capacity $C_k$. In what follows, we let $\mc L := [1\!:\!L]$ and $\mc K := [1\!:\!K]$ indicate the set of users and relays, respectively.

\noindent Similar to~\cite{Simeone2011:IT:CodebookInfoOutBand}, the relay nodes are constrained to operate without knowledge of the users' codebooks and only know a time-sharing sequence $Q^n$, i.e., a set of time instants at which users switch among different codebooks. The obliviousness of the relay nodes to the actual codebooks of the users is modeled via the notion of \textit{randomized encoding} \cite{Sanderovich2008:IT:ComViaDesc} (see also~\cite{Lapidoth:TIT:98:Reliable} for an earlier introduction of this notion in the context of coding for channels with unknown states). That is, users or transmitters select their codebooks at random and the relay nodes are \textit{not} informed about the currently selected  codebooks, while the CP is given such information. Specifically, in this setup, user $l$, $l \in \mc L$, sends codewords $X^n_l(F_l, M_l, Q^n)$ that depend not only on the message $M_l \in [1\!:\!2^{nR_l}]$ of rate $R_l$ that is to be transmitted to the CP by the user and the time-sharing sequence $Q^n$, but also on the index $F_l$ of the codebook selected by this user. This codebook index $F_l$ runs over all possible codebooks of the given rate $R_l$, i.e., $F_l \in [1\!:\!|\mc X_l|^{n2^{nR_l}}]$, and is unknown to the relay nodes. The CP, however, knows all indices of the currently selected codebooks by the users. Also, it is assumed that all terminals know the time-sharing sequence.

\vspace{-3mm}
\subsection{Formal Definitions}

The discrete memoryless CRAN model with oblivious relay processing and enabled time-sharing that we study in this paper is defined as follows.
\begin{enumerate}
\item \textit{Messages and Codebooks:}
Transmitter $l$, $l\in \mathcal{L}$, sends message $M_l\in [1\!:\!2^{nR_l}]$ to the CP using a codebook from a set of codebooks $\{\mathcal{C}_l(F_l)\}$ that is indexed by $F_l\in [1\!:\!|\mathcal{X}_l|^{n2^{nR_l}}]$. The index $F_l$ is picked at random and shared with the CP, but not the relays.
\item \textit{Time-sharing sequence:} All terminals, including the relay nodes, are aware of a time-sharing sequence $Q^n$, distributed as $p_{Q^n}(q^n)=\prod_{i=1}^n p_{Q}(q_i)$ for a pmf $p_{Q}(q)$.

\item  \textit{Encoding functions:} The encoding function at user $l$, $l \in \mc L$, is defined by a pair $(p_{X_l},\phi_l)$ where $p_{X_l}$ is a single-letter pmf and $\phi_l$ is a mapping $\phi_l:[1\!:\!|\mc X_l|^{n2^{nR_l}}]\times [1\!:\!2^{nR_l}]\times \mathcal{Q}^n\rightarrow \mathcal{X}_l^n$ that assigns the given codebook index $F_l$, message $M_l$ and time-sharing variable $Q^n$ to a channel input $X_{l}^n=\phi_l(F_l,M_l, Q^n)$. Conditioned on a time-sharing sequence $Q^n=q^n$,  the probability of selecting a codebook $F_l \in [1\!:\!|\mathcal{X}_l|^{n2^{nR_l}}]$ is given by
\begin{align}
p_{F_l|Q^n}(f_l|q^n) = \prod_{m_l\: \in \: [1\!:\!2^{nR_l}]} p_{X_l^n|Q^n}(\phi_l(f_l,m_l,q^n)|q^n),
\label{probability-distribution-codebook-index}
\vspace{-3mm}
\end{align}
where $p_{X_l^n|Q^n}(x_l^n|q^n) = \prod_{i=1}^n p_{X_l|Q}(x_{l,i}|q_i)$ for some given conditional pmf $p_{X_l|Q}(x_l|q)$.

\item \textit{Relaying functions:} The relay nodes receive the outputs of a memoryless interference channel defined by
\begin{equation}
p_{Y_{\mathcal{K}}^n|X_{\mathcal{L}}^n}(y_{\mathcal{K}}^n|x_{\mathcal{L}}^n)=\prod_{i=1}^np_{Y_{\mc K}|X_{\mathcal{L}}}(y_{\mc K,i}|x_{\mathcal{L},i}).
\label{conditional-distribution-general-DM-CRAN}
\end{equation}
Relay node $k$ , $k \in \mathcal{K}$, is unaware of the codebook indices $F_{\mathcal{L}}= (F_1,\ldots, F_L)$, and maps its received channel output $Y_k^n\in \mathcal{Y}_k^n$ into an index $J_k\in[1\!:\!2^{nC_k}]$ as $J_k=\phi_k^r(Y_k^n,Q^n)$. The index $J_k$ is then sent the to the CP over the error-free link of capacity~$C_k$.

\item \textit{Decoding function:} Upon receiving the indices $J_{\mathcal{K}} := (J_1\ldots,J_K)$, the CP estimates the users' messages $M_{\mathcal{L}}:=(M_1,\ldots,M_L)$ as
\vspace{-3mm}
\begin{align}
(\hat{M}_1,\ldots,\hat{M}_L)= g(F_1,\ldots,F_L,J_1,\ldots,J_K,Q^n),
\end{align}
\vspace{-5mm}
where
\vspace{-3mm}
\begin{align}
g:&[1\!:\!|\mathcal{X}_1|^{n2^{nR_1}}]\times\cdots\times [1\!:\! |\mathcal{X}_L|^{n2^{nR_L}}]\times [1\!:\!2^{nC_1}]\times \cdots \times [1\!:\! 2^{nC_K}]\times \mathcal{Q}^n\nonumber\\
&\rightarrow [1\!:\! 2^{nR_1}] \times \hdots \times [1\!:\! 2^{nR_L}]
\end{align}
 is the decoding function at the CP.
\end{enumerate}

\begin{definition}
A $(n,R_1,\ldots,R_L)$ code for the studied DM CRAN model with oblivious relay processing and enabled time-sharing consists of $L$ encoding functions
%\begin{align}
$\phi_l:[1\!:\!|\mc X_l|^{n2^{nR_l}}]\times [1\!:\!2^{nR_l}]\times \mathcal{Q}^n\rightarrow \mathcal{X}^n_l$,
%\end{align}
$K$ relaying functions
$\phi_k^r: \mathcal{Y}^n_k\times \mathcal{Q}^n\rightarrow [1\!:\!2^{nC_k}]$,
and a decoding function
%\begin{align}
$g:[1\!:\! |\mathcal{X}_1|^{n2^{nR_1}}]\times\cdots\times [1\!:\! |\mathcal{X}_L|^{n2^{nR_L}}]\times [1\!:\! 2^{nC_1}]\times\cdots\times [1\!:\! 2^{nC_K}] \times \mathcal{Q}^n \rightarrow [1\!:\! 2^{nR_1}] \times \hdots \times [1\!:\! 2^{nR_L}]$.
%\end{align}
\end{definition}

\begin{definition}
A rate tuple $(R_{1},\ldots,R_{L})$ is said to be achievable if, for any $\epsilon>0$, there exists a sequence of  $(n,R_1,\ldots,R_L)$ codes such that
\vspace{-2mm}
\begin{align}
 \mathrm{Pr}\{(M_1,\hdots, M_L)\neq (\hat{M}_1,\hdots, \hat{M}_L)\}\leq \epsilon,
 \end{align}
where the probability is taken with respect to a uniform distribution of messages $M_l \in [1\!:\!2^{nR_l}]$, $l=1,\hdots,L$, and with respect to independent indices $F_l$, $l=1,\hdots,L$, whose joint distribution, conditioned on the time-sharing sequence, is given by the product of~\eqref{probability-distribution-codebook-index}.

\noindent For given individual fronthaul constraints $C_{\mathcal{K}} :=(C_1,\hdots,C_K)$, the capacity region $\mathcal{C}(C_{\mathcal{K}})$ is the closure of all achievable rate tuples $(R_{1},\ldots,R_{L})$.
\end{definition}

\vspace{-3mm}
In this work, we are interested in characterizing the capacity region $\mathcal{C}(C_{\mathcal{K}})$.

\subsection{Some Useful Implications}

As shown in~\cite{Simeone2011:IT:CodebookInfoOutBand}, the above constraint of oblivious relay processing with enabled time-sharing means that, in the absence of information regarding the indices $F_{\mathcal{L}}$ and the messages $M_{\mathcal{L}}$, a codeword $x_l^n(f_l,m_l,q^n)$ taken from a $(n,R_l)$ codebook has independent but non-identically distributed entries.
 \vspace{0.2cm}
\begin{lemma}\label{lem:IIDinput}
Without the knowledge of the selected codebooks indices $(F_1,\ldots,F_L)$, the distribution of the transmitted codewords conditioned on the time-sharing sequence are given by
\begin{align}
\mathrm{Pr}\{X_l^n(F_l,W_l,Q^n)=x_l^n|Q^n=q^n\}=\prod_{i=1}^np_{X_l|Q}(x_{l,i}|q_i).
\label{eq:ProdX}
\end{align}
Thus, the channel output  $Y_k^n$ at relay $k\in \mc K$ is distributed as
\begin{align}
&p_{Y_k^n|Q^n}(y^n_k|q^n)\nonumber\!
=\!\prod_{i=1}^n\sum_{x_1,\ldots, x_L}\!\!\!p_{Y_k|X_{\mathcal{L}}}({y}_{k,i}|x_{\mathcal{L},i})\prod_{i=1}^Lp_{X_l|Q}(x_{l,i}|q_i).\nonumber
\end{align}
\end{lemma}

\begin{proof}
The proof of this lemma, whose result was also used in~\cite{Simeone2011:IT:CodebookInfoOutBand}, is along the lines of that of~\cite[Lemma 1]{Sanderovich2008:IT:ComViaDesc} and is therefore omitted for brevity.
\end{proof}

\begin{remark}
Equation~\eqref{eq:ProdX} states that, when averaged over the probability of selecting a codebook $F_l$ and over the uniform distribution of the message set, but conditioned on the time-sharing variable $Q^n$, the transmitted codeword $X_l^n$ has a pmf according to a product distribution $p_{X_l|Q}$ of independent but non-identically distributed entries. That is, in the absence of codebook information, the codewords lack structure. When a node is informed of the codebook index $F_l=f_l$, the codebook structure is provided by the selected codebook.
\end{remark}

\vspace{-3mm}
%----------------------------------------
\section{Discrete Memoryless Model}\label{sec:Main}

\subsection{Capacity Region of a Class of CRANs}\label{ssec:MainClass}

In this section, we establish a single-letter characterization of the capacity region of a class of discrete memoryless CRANs with oblivious relay processing and enabled time-sharing in which the channel outputs at the relay nodes are independent conditionally on the users' inputs. Specifically, consider the following class of DM CRANs in which equation~\eqref{conditional-distribution-general-DM-CRAN} factorizes as
\begin{equation}
p_{Y_{\mathcal{K}}^n|X_{\mathcal{L}}^n}(y_{\mathcal{K}}^n|x_{\mathcal{L}}^n)=\prod_{i=1}^n\prod_{k=1}^Kp_{Y_{k}|X_{\mathcal{L}}}(y_{k,i}|x_{\mathcal{L},i}).
\label{conditional-distribution-class-DM-CRAN}
\end{equation}
Equation~\eqref{conditional-distribution-class-DM-CRAN} is equivalent to that, for all $k \in \mc K$ and all $i \in [1\!:\!n]$,
\begin{align}
Y_{k,i} \mkv X_{\mathcal{L},i} \mkv Y_{\mathcal{K}/k,i}
\label{eq:MKChain_pmf}
\end{align}
forms a Markov chain. The following theorem provides the capacity region of this class of channels.

\vspace{0.1cm}
\begin{theorem}\label{th:MK_C_Main}
For the class of DM CRANs with oblivious relay processing and enabled time-sharing for which~\eqref{eq:MKChain_pmf} holds, the capacity region $\mc C(C_{\mc K})$ is given by the union of all rate tuples $(R_1,\ldots, R_L)$ which satisfy
\begin{align}
\sum_{t\in \mathcal{T}}R_t\leq& \sum_{s\in \mathcal{S}} [C_s-I(Y_{s};U_{s}|X_{\mathcal{L}},Q)]
+ I(X_{\mathcal{T}};U_{\mathcal{S}^c}|X_{\mathcal{T}^c},Q),\label{eq:MK_C_Main}
\end{align}
for all non-empty subsets $\mathcal{T} \subseteq \mathcal{L}$ and all $\mathcal{S} \subseteq \mathcal{K}$, for some joint measure of the form
\begin{align}\label{eq:PMF_Cap}
p(q)\prod_{l=1}^L p(x_l|q)\prod_{k=1}^Kp(y_k|x_{\mathcal{L}})\prod_{k=1}^{K}p(u_k|y_k,q).
\end{align}
\end{theorem}

\vspace{0.1cm}

\begin{proof}
The proof of Theorem~\ref{th:MK_C_Main} appears in Appendix~\ref{app:ConverseClass}.
\end{proof}

\begin{remark}
Our main contribution in Theorem~\ref{th:MK_C_Main} is the proof of the converse part. As mentioned in Appendix~\ref{app:ConverseClass}, the direct part of Theorem~\ref{th:MK_C_Main} can be obtained by a coding scheme in which each relay node compresses its channel output by using Wyner-Ziv binning \cite{Wyner1978} to exploit the correlation with the channel outputs at the other relays, and forwards the bin index to the CP over its rate-limited link. The CP jointly decodes the compression indices (within the corresponding bins) and the transmitted messages, i.e., Cover-El Gamal compress-and-forward~\cite[Theorem 3]{Cover:1979} with joint decompression and decoding (CF-JD)\footnote{The rate region achievable by this scheme for a general DM CRAN, i.e., without the Markov chain~\eqref{eq:MKChain_pmf}, is given by Theorem~\ref{th:NNC_all_MK_inner}.}.
Alternatively, the rate region of Theorem~\ref{th:MK_C_Main} can also be obtained by a direct application of the noisy network coding (NNC) scheme of~\cite[Theorem 1]{Lim:IT:2011:NoisyNetwork}. Observe that the fact that the two operations of decompression and decoding are performed jointly in the scheme CF-JD is critical to achieve the full rate-region of Theorem~\ref{th:MK_C_Main}, in the sense that if the CP first jointly decodes the compression indices and then jointly decodes the users' messages, i.e., the two operations are performed successively, this results in a region that is generally strictly suboptimal. Similar observations can be found \mbox{in~\cite{Sanderovich2008:IT:ComViaDesc}, \cite{DBLP:journals/corr/ZhouX0C16} and \cite{Park:2013:SPLett}.}
\end{remark}

\begin{remark}
Key element to the proof of the converse part of Theorem~\ref{th:MK_C_Main} is the connection with the Chief Executive Officer (CEO) source coding problem\footnote{Because the relay nodes are connected to the CP through error-free finite-rate links, the scenario, as seen by the relay nodes, is similar to one in which a remote vector source $(X^n_1,\hdots,X^n_L)$ needs to be compressed distributively and conveyed to a single decoder. There are important differences, however, as the vector source is not i.i.d. here but given by a codebook that is subject to design.}. For the case of $K \geq 2$ encoders, while the characterization of the optimal rate-distortion region of this problem for general distortion measures has eluded the information theory for now more than four decades, a characterization of the optimal region in the case of logarithmic loss distortion measure has been provided recently in~\cite{ Courtade2014LogLoss}. A key step in  \cite{ Courtade2014LogLoss} is that the log-loss distortion measure admits a lower bound in the form of the entropy of the source conditioned on the decoders input. Leveraging on this result, in our converse proof of Theorem~\ref{th:MK_C_Main} we derive a single letter  upper-bound on the entropy of the channel inputs conditioned on the indices $J_{\mc K}$ that are sent by the relays, in the absence of knowledge of the codebooks indices $F_{\mc L}$. (Cf. the step~\eqref{eq:MK_C_BoundGamma} in Appendix~\ref{app:ConverseClass}).
%The connection with the CEO problem is discussed further in Section~\ref{sec:LogLoss}.
\end{remark}
\begin{remark}
In the special case in which $K=L$ and the memoryless channel~\eqref{conditional-distribution-class-DM-CRAN} is such that $Y_k=X_k$ for $k\in \mc K$, the source coding counter-part of the problem treated in this section reduces to a distributed source coding setting with \textit{independent} sources (recall that the users input symbols are independent here) under logarithmic loss distortion measure. Note that, for $K>2$ and general, i.e., arbitrarily correlated, sources, the problem appears to be of remarkable complexity, and is still to be solved. In fact, the Berger-Tung coding scheme~\cite{Berger1978} can be suboptimal in this case, as is known to be so for Korner-Marton's modulo-two adder problem~\cite{Korner:IT:1979HowToEncode}.
\end{remark}

\subsection{Inner and Outer Bounds for the General DM CRAN Model}\label{sec:GenOCRAN}

In this section, we study the general DM CRAN model~\eqref{conditional-distribution-general-DM-CRAN}. That is, the Markov chains given by~\eqref{eq:MKChain_pmf} are not necessarily assumed to hold. In this case, we establish inner and outer bounds on the capacity region that do not coincide in general. The bounds extend those of~\cite{Sanderovich2008:IT:ComViaDesc}, which are established therein for a setup with a single transmitter and no time-sharing, to the case of multiple transmitters and enabled time-sharing.

\noindent The following theorem provides an inner bound on the capacity region of the general DM CRAN model~\eqref{conditional-distribution-general-DM-CRAN} with oblivious relay processing and time-sharing.
\begin{theorem}\label{th:NNC_all_MK_inner}
For the general DM CRAN model~\eqref{conditional-distribution-general-DM-CRAN} with oblivious relay processing and enabled time-sharing, the achievable rate
region $\mc{R}_{\mathrm{CF-JD}}$ of the scheme CF-JD is given by the union of all rate tuples $(R_1,\ldots,R_L)$ that satisfy, for all non-empty subsets $\mathcal{T} \subseteq \mathcal{L}$ and all $ \mathcal{S} \subseteq \mathcal{K}$,
\begin{align}\label{eq:NNC_all_MK_inner}
\sum_{t\in \mathcal{T}}R_t\leq& \sum_{s\in \mathcal{S}} C_s-I(Y_{S};U_{\mathcal{S}}|X_{\mathcal{L}},U_{\mathcal{S}^c},Q)
+ I(X_{\mathcal{T}};U_{\mathcal{S}^c}|X_{\mathcal{T}^c},Q),
\end{align}
for some joint measure of the form
\begin{align}
p(q)\prod_{l=1}^{L}p(x_l|q) p(y_{\mathcal{K}}|x_{\mathcal{L}})\prod_{k=1}^{K}p(u_k|y_k,q).
\end{align}
\end{theorem}

\begin{proof}
The proof of Theorem~\ref{th:NNC_all_MK_inner} appears in Appendix~\ref{app:NNC_all_MK_inner}.
\end{proof}

\begin{remark}
The coding scheme that we employ for the proof of Theorem~\ref{th:NNC_all_MK_inner}, which we denote by
compress-and-forward with joint decompression and decoding (CF-JD), is one in which every relay node compresses its output \`a-la Cover-El Gamal compress-and-forward~\cite[Theorem 3]{Cover:1979}. The CP jointly decodes the compression indices and users' messages. The scheme, as detailed in Appendix~\ref{app:NNC_all_MK_inner}, generalizes~\cite[Theorem 3]{Sanderovich2008:IT:ComViaDesc} to the case of multiple users and enabled time-sharing.
\end{remark}

\noindent We now provide an outer bound on the capacity region of the general DM CRAN model with oblivious relay processing and time-sharing. The following theorem states the result.

\begin{theorem}\label{th:NNC_all_MK_outer}
For the general DM CRAN model~\eqref{conditional-distribution-general-DM-CRAN} with oblivious relay processing and enabled time-sharing, if a rate tuple $(R_1,\ldots,R_L)$ is achievable then for all non-empty subsets $\mathcal{T} \subseteq \mathcal{L}$ and $ \mathcal{S} \subseteq \mathcal{K}$ it holds that
\vspace{-2mm}
\begin{align}\label{eq:NNC_all_MK_outer}
\sum_{t\in \mathcal{T}}R_t\leq &\sum_{s\in \mathcal{S}} C_s-I(Y_{S};U_{\mathcal{S}}|X_{\mathcal{L}},U_{\mathcal{S}^c},Q)
 + I(X_{\mathcal{T}};U_{\mathcal{S}^c}|X_{\mathcal{T}^c},Q),
\end{align}
for some $(Q,X_{\mathcal{L}},Y_{\mathcal{K}},U_{\mathcal{K}},W)$ distributed according to
\begin{align}
p(q)\prod_{l=1}^{L}p(x_l|q) ~ p(y_{\mathcal{K}}|x_{\mathcal{L}})~p(w|q),
\end{align}
where $u_{k} = f_k(w,y_k,q)$ for $k\in \mc K$; for some random variable $W$ and deterministic functions $\{f_{k}\}$, for $k\in \mc K$.
\end{theorem}
%\noindent\textbf{Proof:} The proof of Theorem~\ref{th:NNC_all_MK_outer} appears in Appendix~\ref{app:NNC_all_MK_outer}. \qed

\begin{proof}
The proof of Theorem~\ref{th:NNC_all_MK_outer} appears in Appendix~\ref{app:NNC_all_MK_outer}.
\end{proof}

\begin{remark}
The inner bound of Theorem~\ref{th:NNC_all_MK_inner} and the outer bound of Theorem \ref{th:NNC_all_MK_outer} do not coincide in general. This is because in Theorem \ref{th:NNC_all_MK_inner}, the auxiliary random variables $U_1,\ldots, U_K$ satisfy the Markov chains $U_k \mkv (Y_k,Q) \mkv (X_{\mathcal{L}}, Y_{\mathcal{L}/k},U_{\mathcal{K}/k})$, while in Theorem~\ref{th:NNC_all_MK_outer} each $U_k$ is a function of $Y_k$ but also of a ``common'' random variable $W$. In particular,  the Markov chains $U_k \mkv (Y_k,Q) \mkv U_{\mathcal{K}/k}$ do not necessarily hold for the auxiliary random variables of the outer bound.
\end{remark}

\begin{remark}
As we already mentioned, the class of DM CRAN models satisfying~\eqref{eq:MKChain_pmf} connects with the CEO problem under logarithmic loss distortion measure. The rate-distortion region of this problem is characterized in the excellent contribution~\cite{Courtade2014LogLoss} for an arbitrary number of (source) encoders (see \cite[Theorem 3]{Courtade2014LogLoss}  therein). For general DM CRAN channels, i.e., \textit{without} the Markov chain~\eqref{eq:MKChain_pmf} the model connects with the distributed source coding problem under logarithmic loss distortion measure. While a solution of the latter problem for the case of two encoders has been found in~\cite[Theorem 6]{Courtade2014LogLoss}, generalizing the result to the case of arbitrary number of encoders poses a significant challenge. In fact, as also mentioned in~\cite{Courtade2014LogLoss},  the Berger-Tung inner bound is known to be generally suboptimal (e.g., see the Korner-Marton lossless modulo-sum problem~\cite{Korner:IT:1979HowToEncode}). Characterizing the capacity region of the general DM CRAN model under the constraint of oblivious relay processing and enabled time-sharing poses a similar challenge, even for the case of two relays.
Finally, we mention that in the context of multi-terminal distributed source coding  with general distortion measure, an outer bound has been derived in \cite{Wagner:IT:2008}; and is shown to be tight in certain cases. The proof technique therein is based on introducing a random source $X$ such that the observations at the encoders are conditionally independent on $X$, i.e., a Markov chain similar to that in~\eqref{eq:MKChain_pmf} holds. Note however that the connection of the outer bound that we develop here for the uplink CRAN model with oblivious relay processing with that of \cite{Wagner:IT:2008} is only of high level nature as the proof techniques are different.
\end{remark}

\subsection{On the Suboptimality of Separate Decompression-Decoding and Role of Time-Sharing}\label{sec:SumRateOptimality}

For the general DM CRAN model~\eqref{conditional-distribution-general-DM-CRAN}, the scheme CF-JD of Theorem~\ref{th:NNC_all_MK_inner} is based on a \textit{joint} decoding of the compression indices and users' messages. That is, the CP performs the operations of the decoding of the quantization codewords and the decoding of the users' messages \textit{simultaneously}. A more practical strategy, considered also in~\cite{Sanderovich2008:IT:ComViaDesc} and~\cite{DBLP:journals/corr/ZhouX0C16}, consists in having the CP first decode the quantization codewords (jointly), and then decode the users' messages (jointly). That is, compress-and-forward with \textit{separate} decompression and decoding operations. In what follows, we refer to such a scheme as CF-SD. The following proposition provides the rate-region allowed by this scheme for the DM CRAN model~\eqref{conditional-distribution-general-DM-CRAN}.

\begin{proposition}{(~\cite[Theorem 1]{Sanderovich2008:IT:ComViaDesc})}\label{prop:CFSD_achi}
For the general DM CRAN model \eqref{conditional-distribution-general-DM-CRAN} with oblivious relay processing and enabled time-sharing, the achievable rate region $\mc {R}_{\mathrm{CF-SD}}$ of the scheme CF-SD is the union of all rate tuples $(R_1,\ldots, R_L)$ that satisfy, for all non-empty $\mc T\subseteq \mc L$ and $\mathcal{S}\subseteq\mathcal{K}$
\begin{subequations}
\begin{align}
\sum_{t\in \mc T}R_{t} &\leq I(X_{\mathcal{T}};U_{\mathcal{K}}|X_{\mathcal{T}^c},Q)\label{eq:SD_constMAC}\\
\sum_{s\in \mathcal{S}}C_s &\geq I(U_{\mathcal{S}};Y_{\mathcal{S}}|U_{\mathcal{S}^c},Q),\label{eq:SD_const2}
\end{align}
\end{subequations}
for some pmf $p(q)\prod_{l=1}^L p(x_l|q)p(y_{\mathcal{K}}|x_{\mathcal{L}})\prod_{k=1}^{K}p(u_k|y_k,q)$.
\end{proposition}

It is clear that the rate region $\mc R_{\text{CF-SD}}$ of Proposition~\ref{prop:CFSD_achi} is contained in that, $\mc R_{\text{CF-JD}}$, of Theorem~\ref{th:NNC_all_MK_inner}.
\iffalse
The following example shows that this containment is in general \textit{strict}, i.e., $ \mc R_{\text{CF-JD}} \subsetneq \mc R_{\text{CF-SD}}$.

\begin{example}
Consider the example CRAN model with $L=2$ and $K=2$ shown in Figure~\ref{fig-example-suboptimality-of-separate-decompression-decoding}.
\end{example}
\fi

As a special instance of the scheme CF-SD, we consider compress-and-forward with \textit{successive} separate decompression-decoding performs \textit{sequential} decoding of the quantization codewords first, followed by \textit{sequential} decoding of the users' messages. More specifically, let ${\pi}_r\: : \: \mc K \rightarrow \mc K$ and ${\pi}_u \: : \: \mc L \rightarrow \mc L $ be two permutations that are defined on the set of quantization codewords and the set of user message codewords, respectively. An outline of this scheme, which we denote as CF-SSD, is as follows. The relays compress their outputs  sequentially, starting by relay node $\pi_r(1)$. In doing so, they utilize Wyner-Ziv binning~\cite{Wyner1978}, i.e., relay node $\pi_r(k)$, $k\in \mc K$, quantizes its channel output $Y^n_{\pi_r(k)}$ into a description $U^n_{\pi_r(k)}$ taking into account $(U^n_{\pi_r(1)}, \hdots, U^n_{\pi_r(k-1)})$ as decoder side information. The CP first recovers the quantization codewords in the same order, and then decodes the users' messages sequentially, in the order indicated by $\pi_u$, starting by user $\pi_u(1)$. That is, the codeword of user $l$, $l\in \mc L$, is estimated using all compression codewords $(Y^n_{\pi_r(1)},\hdots, Y^n_{\pi_r(K)})$ as well as the previously decoded user codewords $(X^n_{\pi_u(1)},\hdots, X^n_{\pi_u(l-1)})$. The rate-region obtained with a given decoding order $(\pi_r,\pi_u)$ as well as that of the scheme CF-SSD, obtained by considering all possible permutations, are given in the following proposition.
\begin{proposition}\label{prop:CFSWZ_achi}
For the general DM CRAN model~\eqref{conditional-distribution-general-DM-CRAN} with oblivious relay processing and enabled time-sharing, the achievable rate region $\mc R_{\text{CF-SSD}}(\pi_r,\pi_u)$ of the scheme CF-SSD with decoding order  $(\pi_r,\pi_u)$ is the union of all rate tuples $(R_1,\hdots,R_L)$ that satisfy, for all $l \in \mc L$ and $k \in \mc K$,
\begin{subequations}
\begin{align}
R_{\pi_u(l)} &\leq I(X_{\pi_u(l)};U_{\mc K}|X_{\pi_u(1)},\hdots,X_{\pi_u(l-1)},Q)\\
C_{\pi_r(k)} &\geq I(U_{\pi_r(k)};Y_{\pi_(k)}|U_{\pi_r(1)},\hdots,U_{\pi_r(k-1)},Q),
\end{align}
\label{eq:SWZ_const2}
\end{subequations}
for some pmf $p(q)\prod_{l=1}^L p(x_l|q)p(y_{\mathcal{K}}|x_{\mathcal{L}})\prod_{k=1}^{K}p(u_k|y_k,q)$. The rate region $\mc R_{\text{CF-SSD}}$ achievable by the scheme CF-SSD is defined as the union of the regions $\mc R_{\text{CF-SSD}}(\pi_r,\pi_u)$ over all possible permutations $\pi_r$ and $\pi_u$, i.e.,
\begin{equation}
\mc R_{\text{CF-SSD}} = \bigcup_{\pi_r,\:\pi_u} \mc R_{\text{CF-SSD}}(\pi_r,\pi_u).
\end{equation}
\end{proposition}

While successive separate decompression and decoding results in a rate region that is generally strictly smaller than that of joint decoding, i.e., with CF-JD, in what follows we show that the maximum sum-rate that is achievable by this specific separate decompression-decoding is the same as that achieved by joint decoding. That is, the schemes CF-SSD and CF-JD achieve the same sum-rate (and, so, so does also the scheme CF-SD). Specifically, let the maximum sum-rate achieved by the scheme CF-JD be defined as
\[
R_{\text{sum, CF-JD}}=
\left\{\begin{array}{l}
\max \sum_{i=1}^{L} R_i \\
\text{s.t.}\: (R_1,\hdots,R_L) \in \mc R_{\text{CF-JD}}.
\end{array}
\right.
\]

\noindent Similarly, let the maximum sum rate for the scheme CF-SD be defined as
$$
R_{\text{sum, CF-SD}}=
\left\{
\begin{array}{l}
\max \sum_{i=1}^{L} R_i \\
\text{s.t.}\: (R_1,\hdots,R_L) \in \mc R_{\text{CF-SD}},
\end{array}
\right.
$$
and that of the scheme CF-SSD defined as
$$
R_{\text{sum, CF-SSD}}=
\left\{
\begin{array}{l}
\max \sum_{i=1}^{L} R_i \\
\text{s.t.}\: (R_1,\hdots,R_L) \in \mc R_{\text{CF-SSD}}.
\end{array}
\right.
$$

\begin{theorem}\label{th:SWZSumRate}
For the general DM CRAN model~\eqref{conditional-distribution-general-DM-CRAN} with oblivious relay processing and enabled time-sharing in Figure~\ref{fig:Schm}, we have
\begin{equation}
R_{\text{sum, CF-JD}} = R_{\text{sum, CF-SD}} = R_{\text{sum, CF-SSD}}.
\end{equation}
\end{theorem}

\begin{proof}
The proof of Theorem~\ref{th:SWZSumRate} appears in Appendix~\ref{app:SWZSumRate}.
\end{proof}

\begin{remark}
The proof of Theorem~\ref{th:SWZSumRate} uses properties of submodular optimization; and is similar to that of \cite[Theorem 2]{DBLP:journals/corr/ZhouX0C16} which shows that CF-JD and CF-SD achieve the same sum-rate for the class of CRANs that satisfy~\eqref{eq:MKChain_pmf}. Thus, in a sense, Theorem~\ref{th:SWZSumRate} can be thought of as a generalization of \cite[Theorem 2]{DBLP:journals/corr/ZhouX0C16} to the case of general channels~\eqref{conditional-distribution-general-DM-CRAN}.
 A generalized successive decompression-decoding scheme (CF-GSD) which allows arbitrary interleaved decoding orders between quantization codewords and users' messages is proposed in \cite{DBLP:journals/corr/ZhouX0C16}, which under the sum-rate constraint is also optimal. In general, CF-GSD achieves a larger rate-region that CF-SD and achieves the same rate-region as CF-JD under sum-fronthaul constraint \cite[Theorem 2]{DBLP:journals/corr/ZhouX0C16}.
\end{remark}

\begin{remark}\label{rm:Successive}
Theorem~\ref{th:SWZSumRate} shows that the three schemes CF-JD, CF-SD and CF-SSD achieve the same sum-rate and that, in general, the use of time-sharing is required for the three schemes to achieve the maximum sum-rate. Note that the uplink CRAN is a multiple-source, multiple-relay, single-destination network. If all fronthaul capacities were infinite, then the model would reduce to a standard multiple access channel (MAC) and it follows from standard results that time-sharing is not needed to achieve the optimal sum-rate in this case \cite{elGamal:book}. The reader may wonder whether it is also so in the case of finite-rate fronthaul links, i.e., whether one can optimally set $Q=\emptyset$ in the region $\mc C(C_{\mc K})$ for sum-rate maximization. The answer to this question is negative for finite fronthaul capacities $\{C_l\}$, as shown in Section~\ref{ssec:Gauss}. This is reminiscent of the fact that time-sharing generally increase rates in relay channels, e.g., \cite{ElGamal:IT:2006,Kim:Allerton:2007}.
In addition, when the three schemes CF-JD, CF-SD and CF-SSD are  restricted to operate without time-sharing, i.e., $Q=\emptyset$, CF-SSD might perform strictly worse than CF-JD and CF-SD. To see this, the reader may find it useful to observe that while time-sharing is not required for sum-rate maximization in a regular MAC, as successive decoding (in any order) is sum-rate optimal in this case, it is beneficial when the sum-rate maximization is subjected to constraints on the users' message rates such as when the users' rates need to be symmetric~\cite{Rimoldi:IT:1996}, i.e., the operation point is not in a corner point of the MAC region. Similarly, standard successive Wyner-Ziv (in any order, without time-sharing) is known to achieve any corner point of the Berger-Tung region~\cite{BergerChen:IT:2008, Wagner2008}, but time-sharing (or rate-splitting \`a-la~\cite{BergerChen:IT:2008}) is beneficial if the compression rates are subjected to constraints such as when the compression rates are symmetric. An example which illustrates these aspects for memoryless Gaussian CRAN is provided in Section~\ref{ssec:Gauss}.

\end{remark}

\section{Memoryless MIMO Gaussian CRAN}\label{ssec:Gauss}

In this section, we consider a memoryless Gaussian MIMO CRAN with oblivious relay processing and enabled time-sharing. Relay node $k$, $k \in \mc K$, is equipped with $M_k$ receive antennas and has channel output
\vspace{-2mm}
\begin{equation}
\mathbf{Y}_k = \mathbf{H}_{k,\mathcal{L}}\mathbf{X}+\mathbf{N}_k,
\label{mimo-gaussian-model}
\end{equation}
where $\mathbf{X}:=[\mathbf{X}_1^T,\ldots,\mathbf{X}_L^T]^T$, $\mathbf{X}_{l}\in \mathds{C}^{N_l}$ is the channel input vector of user $l \in \mc L$, $N_l$ is the number of antennas at user $l$, $\mathbf{H}_{k,\mathcal{L}} := [\mathbf{H}_{k,1},\ldots, \mathbf{H}_{k,L}]$ is the matrix obtained by concatenating the $\mathbf{H}_{k,l}$, $l\in \mc L$, horizontally, with $\mathbf{H}_{k,l}\in \mathds{C}^{M_k\times N_l}$ being the channel matrix connecting user $l$ to relay node $k$, and $\mathbf{N}_k\in\mathds{C}^{M_k}$ is the noise vector at relay $k$, assumed to be memoryless Gaussian with covariance matrix $\mathbf{N}_k\sim\mathcal{CN}(\mathbf{0},\mathbf{\Sigma}_k)$ and independent from other noises and from the channel inputs $\{\dv X_l\}$. The transmission from user $l \in \mc L$ is subjected to the covariance constraint,
\begin{equation}
\mathrm{E}[\mathbf{X}_l\mathbf{X}^H_{l}]\preceq \mathbf{K}_{l},
\label{input-covariance-matrix-Gaussian-model}
\end{equation}
where $\dv K_l$ is a given $N_l {\times} N_l$ positive semi-definite matrix, and the notation $\preceq$ indicates that the matrix $(\dv K_l - \mathrm{E}[\mathbf{X}_l\mathbf{X}^H_{l}])$ is positive semi-definite.

\subsection{Capacity Region under Time-Sharing of Gaussian Inputs}\label{ssec:TimeSharingGaussInputs}

The memoryless MIMO Gaussian model with oblivious relay processing described by~\eqref{mimo-gaussian-model} and~\eqref{input-covariance-matrix-Gaussian-model} clearly falls into the class of CRANs studied in Section~\ref{ssec:MainClass}, since $\dv Y_k \mkv (\dv X_1,\hdots,\dv X_L) \mkv (\dv Y_1,\hdots,\dv Y_{k-1},\dv Y_{k+1},\hdots,\dv Y_K)$ forms a Markov chain in this order for all $k \in \mc K$. Thus, Theorem~\ref{th:MK_C_Main}, which can be extended to continuous channels using standard techniques, characterizes the capacity region of this model. The computation of the region of Theorem~\ref{th:MK_C_Main}, i.e., $\mc C(C_{\mc K})$, for the model described by~\eqref{mimo-gaussian-model} and~\eqref{input-covariance-matrix-Gaussian-model}, however, is not easy as it requires  finding the optimal choices of channel inputs $(\dv X_1,\hdots,\dv X_L)$ and the involved auxiliary random variables $(U_1, \hdots, U_K)$. In this section, we find an explicit characterization of the capacity region of the model described by~\eqref{mimo-gaussian-model} and~\eqref{input-covariance-matrix-Gaussian-model} in the case in which the users are constrained to time-share only among \textit{Gaussian} codebooks. That is, for all $q \in Q$ and all $l \in \mc L$, the distribution of the input $\dv X_l$ conditionally on $Q=q$ is Gaussian (with covariance matrix that can be optimized over so as to satisfy~\eqref{input-covariance-matrix-Gaussian-model}). We denote that region by $\mc C_{\mathrm{G}}(C_{\mc K})$.  Although Gaussian input may generally be suboptimal for uplink CRAN~\cite{Sanderovich2008:IT:ComViaDesc}, i.e., in general $\mc C_{\mathrm{G}}(C_{\mc K})\subset \mc C(C_{\mc K})$,  restricting to Gaussian input for every $Q=q$ is appreciable because it leads to rate regions that are less difficult to evaluate. In doing so, we also show that time-sharing Gaussian compression at the relay nodes is optimal if the users' channel inputs are restricted to be Gaussian for all $q \in Q$.

Let, for all $l \in \mc L$, the input $\dv X_l$ be restricted to be distributed such that for all $Q=q$,
\vspace{-3mm}
\begin{equation}\label{eq:GaussInput}
\dv X_l|Q=q \sim \mathcal{CN}(\mathbf{0},\mathbf{K}_{l,q}),
\end{equation}
 where the matrices $\{\dv K_{l,q}\}_{q=1}^{|\mc Q|}$ are chosen to satisfy
\vspace{-3mm}
\begin{equation}
\sum_{q\in\mathcal{Q}}p_Q(q)\mathbf{K}_{l,q} \preceq \mathbf{K}_l.
\label{eq:powConst}
\end{equation}

The following theorem characterizes the capacity region of the model with oblivious relay processing described by~\eqref{mimo-gaussian-model} and~\eqref{input-covariance-matrix-Gaussian-model} under the constraint of fixed Gaussian input and given \mbox{fronthaul capacities $C_{\mc K}$.}

\begin{theorem}~\label{th:GaussSumCap}
 The capacity region $\mc C_{\mathrm{G}}(C_{\mc K})$ of the memoryless Gaussian MIMO model with oblivious relay processing described by~\eqref{mimo-gaussian-model} and~\eqref{input-covariance-matrix-Gaussian-model} under time-sharing of Gaussian inputs is given by the set of all rate tuples $(R_1,\ldots,R_L)$ that satisfy
\begin{align}
\sum_{t\in\mathcal{T}}R_{t} \leq& \sum_{k\in \mathcal{S}}\left[C_k-\mathrm{E}_{Q}\left[\log\frac{|\mathbf{\Sigma}_k^{-1}|}{|\mathbf{\Sigma}^{-1}_k-\mathbf{B}_{k,Q}|}\right]\right] \nonumber\\
 & + \mathrm{E}_{Q}\left[ \log \frac{|\sum_{k\in\mathcal{S}^{c}}\mathbf{H}_{k,\mathcal{T}}^{H}
\mathbf{B}_{k,Q} \mathbf{H}_{k,\mathcal{T}}+\mathbf{K}^{-1}_{\mathcal{T},Q}|}{|\mathbf{K}_{\mathcal{T},Q}^{-1}|}\right],
\label{eq:GaussSumCap}
\end{align}
for all non-empty $ \mathcal{T} \subseteq \mathcal{L}$ and all $\mathcal{S} \subseteq \mathcal{K}$, for some pmf $p_Q(q)$ and matrices $\dv K_{q,l}$ and $\mathbf{B}_{k,q}$ such that $\mathrm{E}_Q [\dv K_{l,Q}]\preceq \dv K_{l}$ and $\mathbf{0}\preceq \mathbf{B}_{k,q}\preceq \mathbf{\Sigma}_{k}^{-1}$; and where, for $q \in Q$ and $\mc T \subseteq \mc L$, the matrix $\mathbf{K}_{\mathcal{T},q}$ is defined as $\mathbf{K}_{\mathcal{T},q} := \text{diag}[\{\mathbf{K}_{t,q}\}_{t\in\mathcal{T}}]$.
\end{theorem}

\begin{proof}
The proof of Theorem~\ref{th:GaussSumCap} appears in Appendix~\ref{app:GaussSumCap}.
\end{proof}

\begin{remark}
Theorem~\ref{th:GaussSumCap} extends the result with oblivious relay processing of~\cite[Theorem 5]{Sanderovich2008:IT:ComViaDesc} to the MIMO setup  with $L$ users and enabled time-sharing, and shows that under the constraint of Gaussian signaling, the quantization codewords can be chosen optimally to be Gaussian. Recall that, as shown through an example in~\cite{Sanderovich2008:IT:ComViaDesc},  restricting to Gaussian input signaling can be a severe constraint and is generally suboptimal.
\end{remark}

\subsection{On the Role of Time-Sharing }~\label{sec-role-time-sharing}
In Remark~\ref{rm:Successive} in Section~\ref{sec:SumRateOptimality} we commented on the utility of time-sharing for sum-rate maximization in the uplink of DM CRAN with oblivious relay processing. In this section we investigate further the role of time-sharing. Specifically, we first provide an example in which time-sharing increases capacity; and then discuss some scenarios in which time-sharing does \textit{not} enlarge the capacity region of the memoryless MIMO Gaussian CRAN model with oblivious relay processing described by~\eqref{mimo-gaussian-model} and~\eqref{input-covariance-matrix-Gaussian-model}.

For convenience, let us denote by $\mc C_{\text{G}}^{\text{no-ts}}(C_{\mc K})$ the rate region obtained by setting $Q=\emptyset$, i.e, without enabled time-sharing, in the region of Theorem~\ref{th:GaussSumCap}. That is, $\mc C_{\text{G}}^{\text{no-ts}}(C_{\mc K})$ is given by the set of all rate tuples $(R_1,\ldots,R_L)$ that for all non-empty $ \mathcal{T} \subseteq \mathcal{L}$ and all $\mathcal{S} \subseteq \mathcal{K}$
\begin{align}
\sum_{t\in\mathcal{T}}R_{t} \leq&
 \sum_{k\in \mathcal{S}}\left[C_k-\log\frac{|\mathbf{\Sigma}_k^{-1}|}{|\mathbf{\Sigma}^{-1}_k-\mathbf{B}_k|}\right]
 + \log \frac{
|\sum_{k\in\mathcal{S}^{c}}\mathbf{H}_{k,\mathcal{T}}^{H}
\mathbf{B}_{k}
\mathbf{H}_{k,\mathcal{T}}+\mathbf{K}^{-1}_{\mathcal{T}}|
}{
|\mathbf{K}_{\mathcal{T}}^{-1}|
},\label{eq:GaussSumCap_EqualK}
\end{align}
for some $\mathbf{0}\preceq \mathbf{B}_{k}\preceq \mathbf{\Sigma}_{k}^{-1}$, $k\in \mc K$.

\noindent The following example shows that $\mc C_{\text{G}}^{\text{no-ts}}(C_{\mc K})$ may be contained \textit{strictly} in $\mc C_{G}(C_{\mc K})$.

\begin{example}\label{ex:Example1}
Consider an instance of the memoryless MIMO Gaussian CRAN described by~\eqref{mimo-gaussian-model} and~\eqref{input-covariance-matrix-Gaussian-model} in which $L=1$, $K=2$, $M_1=M_2=N_1=1$ (all devices are equipped with single-antennas), the relay nodes have equal fronthaul capacities, i.e.,  $C_1 = C_2 = C$, and
\begin{align}
Y_k = a X + N_k,\quad \text{for}\:\: k=1,2,
\end{align}
where $\mathrm{E}[|X|^2]\leq P$ and $N_k \sim \mc{CN}(0,1)$, for $k=1,2$.

\vspace{0.2cm}

\noindent The capacity $C_{G}(C)$ of this one-user Gaussian CRAN example can be obtained from Theorem~\ref{th:GaussSumCap} as the following optimization problem
\begin{align}
C_{G}(C) = \max_{\alpha_q, b_q, P_q } &\min_{\mc S\subseteq \{1,2\}}
 \Bigg\{|\mathcal{S}| [C+\sum_{q = 1}^{|\mc Q|}\alpha_q\log(1-b_q)] + \sum_{q = 1}^{|\mc Q|}\alpha_q\log \left(|\mc S^c| P_q  a^2 b_q  + 1 \right)\Bigg\}
\label{capacity-Gaussian-example-with-time-sharing}
\end{align}
where the maximization is over $0\leq b_q\leq 1$, $0\leq \alpha_q\leq 1$ and $P_q \geq 0$, such that $\sum_{q=1}^{|\mc Q|}\alpha_q = 1$  and $\sum_{q=1}^{|\mc Q|}\alpha_q P_q \leq  P$.
Due to Theorem \ref{th:SWZSumRate}, $C_{G}(C)$ is achievable with CF-JD, CF-SD and CD-SSD by using time-sharing.
\noindent Without time-sharing, i.e., $Q=\emptyset$, the capacity $C_{\text{G}}^{\text{no-ts}}(C)$ of this one-user Gaussian CRAN example is achievable with the CF-JD scheme and can be obtained easily from~\eqref{eq:GaussSumCap_EqualK}, as
\begin{align}
C_{\text{G}}^{\text{no-ts}}(C)&=\max_{0\leq b\leq 1}
\min_{\mc S\subseteq \{1,2\}}
\Bigg\{
 |\mathcal{S}| \left[C+\log(1-b)\right]
+ \log \left(|\mc S^c| P a^2 b  + 1 \right)\Bigg\}\\
&=\log\left(1+ 2  a^2 P 2^{-2C} \left(2^{2C} + a^2P - \sqrt{a^4P^2+(1+2Pa^2) 2^{2C}}\right) \right).
\label{capacity-Gaussian-example-no-time-sharing}
\end{align}

With time-sharing with, say $\mc Q = \{1,2\}$, the user can communicate at larger rates with CF-JD, as follows.  The transmission time is divided into two periods or phases, of duration ${\alpha}n$  and $(1-\alpha)n$ respectively, where $0 < \alpha < 1$. The user transmits symbols only during the first phase, with power $ P/\alpha$; and it remains silent during the second phase. The two relay nodes operate as follows. During the first phase, relay node $k$, $k=1,2$, compresses its output to the fronthaul constraint $C/\alpha$; and it remains silent during the second phase. Observe that with such transmission scheme the input constraint~\eqref{eq:powConst} and fronthaul constraints
%$\lim_{n \to \infty} (1/n)\sum_{i=1}^n C_k(i) \leq C_k$
are satisfied. Evaluating the rate-region of Theorem~\ref{th:GaussSumCap} with the choice $p_Q(1)=\alpha$, $p_Q(2)=(1-\alpha)$, $\dv K_{k,1}=P/{\alpha}$ and $\dv K_{k,2}=0$, $k=\{1,2\}$ yields in this case
\begin{align}
R_{\text{G,CF-JD}}^{\text{two-ph}}(C):= \max_{0\leq \alpha \leq 1} \max_{0\leq b\leq 1} \min_{\mc S\subseteq \{1,2\}}
\alpha \Bigg\{ |\mathcal{S}| \left[\frac{C}{\alpha}+\log(1-b)\right] + \log \left(|\mc S^c| \frac{P}{\alpha}  a^2 b  + 1 \right)\Bigg\}
\label{lower-bound-capacity-Gaussian-example-with-time-sharing}
\end{align}

\noindent Figure~\ref{fig:ExampleI} depicts the evolution of the capacity enabled with time-sharing $C_{G}(C)$, the capacity without time-sharing $C_{\text{G}}^\text{no-ts}(C)$, as well as the cut-set upper bound, for $a=1$ and $C=0.5$, as function of the user transmit power~$P$.  Also shown for comparison is the achievable rate $R_{\text{G,CF-JD}}^{\text{two-ph}}(C)$ as given by \eqref{lower-bound-capacity-Gaussian-example-with-time-sharing}, which is a lower bound on  $C_{\text{G}}(C)$. Observe that while restricting to CF-JD with two-phases might be suboptimal,  $R_{\text{G,CF-JD}}^{\text{two-ph}}(C)$ is very close to $C_{G}(C)$. As it can be seen from the figure,  the utility of time-sharing (to increase rate) is visible mainly at small average transmit power. The intuition for this gain is that, for small $P$,  the observations at the relay nodes become too noisy and the relay mostly forwards noise. It is therefore more advantageous to increase the power at $P/\alpha$ for a fraction $\alpha$ of the transmission. Accordingly, the effective compression rate is increased to $C/\alpha$, therefore reducing the compression noise. This observation is reminiscent of similar ones in~\cite{ElGamal:IT:2006} in the context of relay channels with orthogonal components and in~\cite{Kim:Allerton:2007} in the context of primitive relay channels.

%----------------------------------------
\begin{figure}[t!]
\centering
\includegraphics[width=0.65\textwidth]{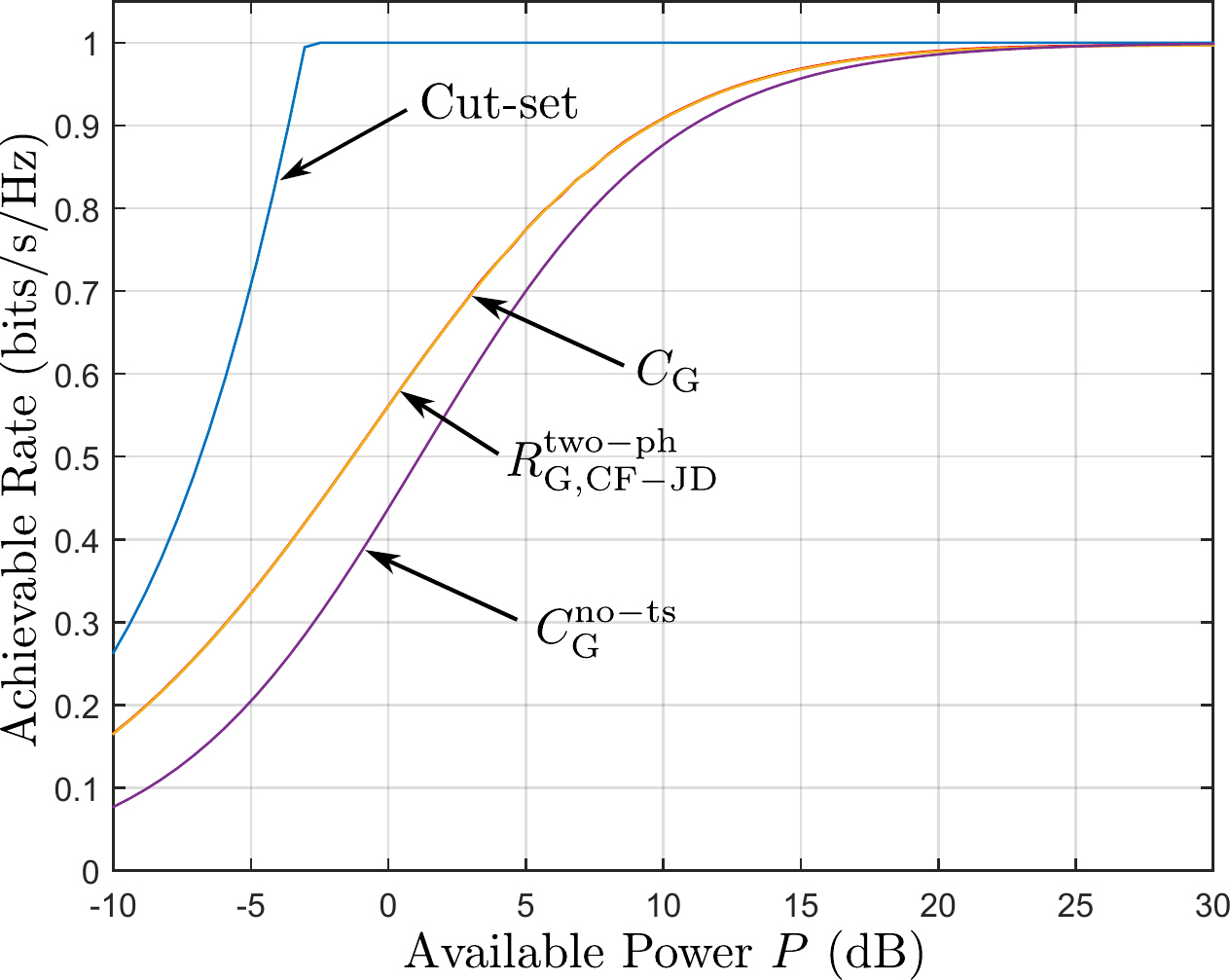}
\vspace{-2mm}
\caption{Capacity with enabled time-sharing and without time-sharing as well as the achievable rate $R_{\text{G,CF-JD}}(C)$ for the model of Example~\ref{ex:Example1}. Numerical values are $L=1$, $K=2$, $M_1=M_2=N_1=1$, $a=1$ and $C=0.5$.}
\label{fig:ExampleI}
\end{figure}

When the three schemes CF-JD, CF-SD and CF-SSD are restricted to operate without time-sharing, i.e., $Q=\emptyset$, and Gaussian signaling, CF-SD and CF-SSD might perform strictly worse than CF-JD. The rate achievable by the CF-SD scheme without time-sharing  follows by Proposition~\ref{prop:CFSD_achi}, and it is easy to show that it coincides with $C_{\text{G}}^{\text{no-ts}}(C)$ in \eqref{capacity-Gaussian-example-no-time-sharing}, i.e.,  in this example, CF-JD and CF-SD achieve the  capacity $C_{\text{G}}^{\text{no-ts}}(C)$ without time-sharing.
The rate achievable by CF-SSD without time-sharing and Gaussian test channels $U_k \sim \mc{CN}(Y_k,\sigma_k^2)$, $k\in \mc K$, can be obtained from Proposition~\ref{prop:CFSWZ_achi}, as
\begin{align}\label{lower-bound-CFSSD}
R_{\text{G,CF-SSD}}^{\text{no-ts}}(C):= \log\left( 1 + P a^2\left((1+\sigma_1^{-2})^{-1}+(1+\sigma_2^{-2})^{-1}\right)\right),
\end{align}
where  $\sigma^2_{1}= (a^2P+1)/(2^{C}-1)$ and $\sigma^2_{2}= (a^2P+1 -  a^4P^2(a^2 P + 1 +\sigma^2_{1})^{-1})/(2^{C}-1)$.

\noindent Figure~\ref{fig:ExampleII} shows  the capacities $C_{\text{G}}(C)$, $C_{\text{G}}^{\text{no-ts}}(C)$  and the achievable rates $R_{\text{G,CF-JD}}^{\text{two-ph}}(C)$ and $R_{\text{G,CF-SSD}}^{\text{no-ts}}(C)$ for $a = 1$ and $C=6$, as function of the transmit power $P$. Note that CF-SSD, when restricted not to use time-sharing performs strictly worse than CF-JD and CF-SD without time-sharing, i.e., $C_{\text{G}}^{\text{no-ts}}(C)$. Observe that in this scenario, the gains due to time-sharing are limited. This observation is in line with the fact that for large fronthaul values, the CRAN model reduces to a MAC, for which time-sharing is not required to achieve the optimal sum-rate.
\qed
\end{example}

%----------------------------------------
\begin{figure}[t!]
\centering
\includegraphics[width=0.65\textwidth]{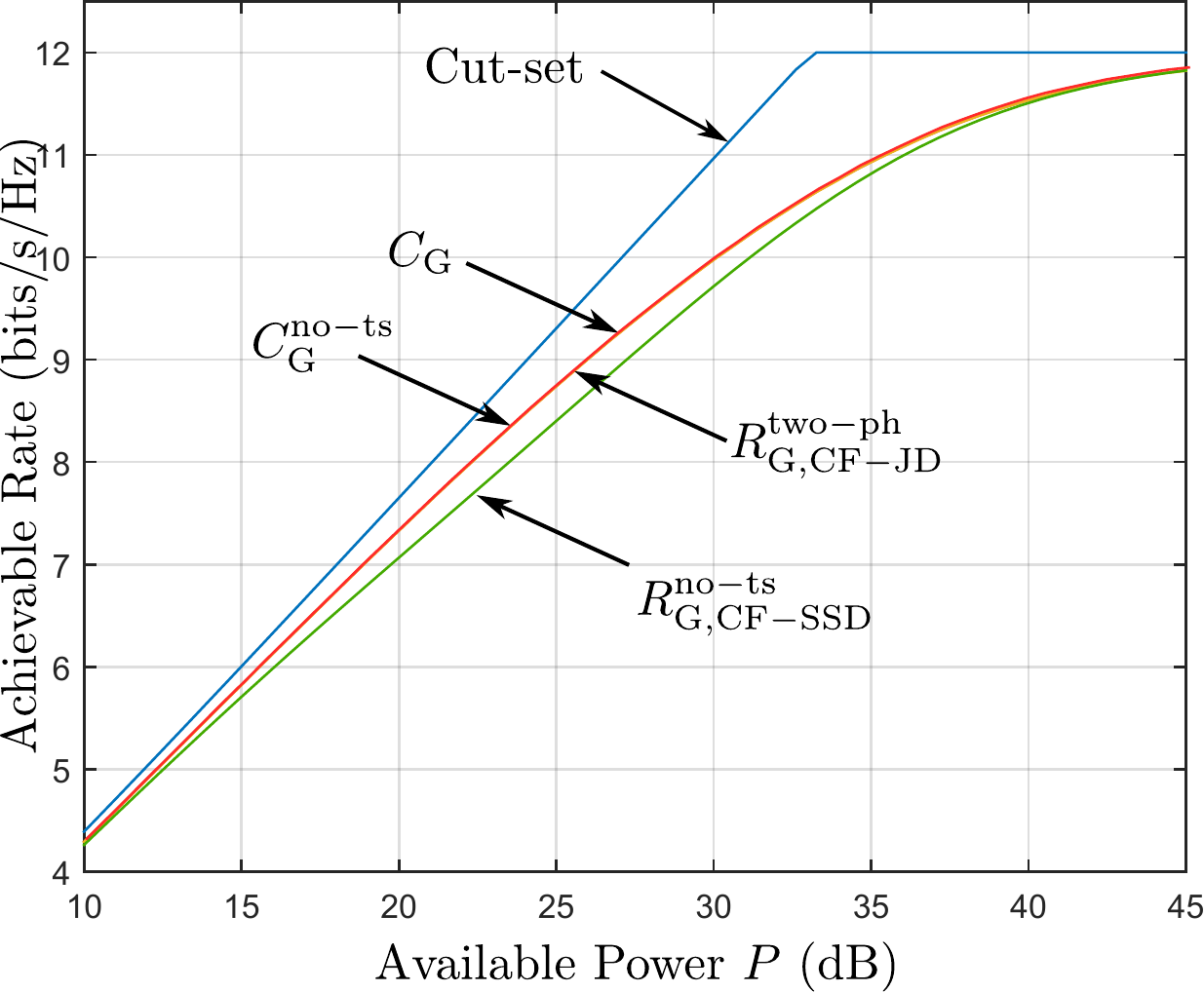}
\vspace{-2mm}
\caption{Capacity with $C_{\text{G}}(C)$ and $C_{\text{G}}^{\text{no-ts}}(C)$ and rates achievable by CF-JD, CF-SD and CF-SSD without time-sharing  for the model of Example~\ref{ex:Example1}. Numerical values are: $a=1$, $C=6$.}
\label{fig:ExampleII}
\end{figure}

The above shows that in general time-sharing increases rates for the memoryless MIMO Gaussian CRAN model described by~\eqref{mimo-gaussian-model} and~\eqref{input-covariance-matrix-Gaussian-model}, i.e., $\mc C_{\text{G}}^{ \text{no-ts}}(C_{\mc K}) \subsetneq \mc C_{G}(C_{\mc K})$. In what follows, we discuss two scenarios in which time-sharing does \textit{not} enlarge the capacity region of the model given by~\eqref{mimo-gaussian-model} and~\eqref{input-covariance-matrix-Gaussian-model}, i.e., $\mc C_{\text{G}}^{\text{no-ts}}(C_{\mc K})=\mc C_{G}(C_{\mc K})$.

\subsubsection{Case of Fixed Gaussian Codebook at User Side}

Consider the scenario in which the users are not allowed to time-share among several Gaussian codebooks, but they are constrained to use each a single, possibly different, Gaussian codebook. This may be relevant, e.g., for contexts in which signaling overhead reduction among the users and relays is of prime interest. Conceptually, this corresponds to equalizing all the covariance matrices $\{\dv K_{l,q}\}$ for given $l$ and all $q=1,\hdots,|\mc Q|$. Let
\begin{align}
\tilde{\dv K}_l := \dv K_{l,1}=\cdots = \dv K_{l,|\mc Q|} \preceq \dv K_l.
\label{eq:EqualCov}
\end{align}

The reader may wonder whether allowing the relay nodes to time-share among compression codebooks can be beneficial in this case. Note that the answer to this question is not clear a-priori, because time-sharing in general increases the Berger-Tung rate region if constraints on the rates are imposed. (See Remark~\ref{rm:Successive}). The following proposition shows that for the model described by~\eqref{mimo-gaussian-model} and~\eqref{input-covariance-matrix-Gaussian-model} this does not hold under the constraint~\eqref{eq:EqualCov}.

\begin{proposition}\label{th:GaussSumCap_EqualK}
For the model with oblivious relay processing described by~\eqref{mimo-gaussian-model} and~\eqref{input-covariance-matrix-Gaussian-model}, if~\eqref{eq:EqualCov} holds for all $l\in \mc L$ then $\mc C_{\text{G}}^{\text{no-ts}}(C_{\mc K})=\mc C_{G}(C_{\mc K})$.
\end{proposition}

\begin{proof}
The proof of Proposition~\ref{th:GaussSumCap_EqualK} appears in Appendix~\ref{app:GaussSumCap_EqualK}.
\end{proof}

\subsubsection{High SNR Regime}

Consider again the model described by~\eqref{mimo-gaussian-model} and~\eqref{input-covariance-matrix-Gaussian-model}. Assume that for all $k \in \mc K$ the vector Gaussian noise at relay node $k$ has covariance matrix
\begin{align}
\dv\Sigma_{k}=\epsilon \tilde{\dv\Sigma}_k
\label{noise-covariance-matrix-high-SNR-regime}
\end{align}
for some $\epsilon \geq 0$ and $\tilde{\dv\Sigma}_k\succeq \dv 0$ that is independent from $\epsilon$.

The following proposition shows that, in this case, the benefit of time-sharing in terms of increasing rates vanishes for arbitrarily small $\epsilon$.

\begin{proposition}\label{th:GaussSumCap_highSNR}
For the model with oblivious relay processing described by~\eqref{mimo-gaussian-model} and~\eqref{input-covariance-matrix-Gaussian-model}, if for all $k \in \mc K$ the vector Gaussian noise at relay node $k$ has covariance matrix that can be put in the form given by~\eqref{noise-covariance-matrix-high-SNR-regime} for some $\epsilon \geq 0$ and $\tilde{\dv\Sigma}_k\succeq \dv 0$ that is independent from $\epsilon$, then the following holds: If $(R_1,\ldots, R_L)\in \mc C_{G}(C_{\mc K})$, then $(R_1-\Delta_{\epsilon},\ldots, R_L-\Delta_{\epsilon})\in \mc C_{\text{G}}^{ \text{no-ts}}(C_{\mc K})$ for some $\Delta_{\epsilon} \geq 0$. In addition
\begin{align}
\lim_{\epsilon \rightarrow 0}\Delta_{\epsilon} = 0.
\end{align}
\end{proposition}

\begin{proof}
The proof of Proposition~\ref{th:GaussSumCap_highSNR} appears in Appendix~\ref{app:GaussSumCap_highSNR}.
\end{proof}

\subsection{Price of Non-Awareness: Bounded Rate Loss}

In this section, we show that for the memoryless MIMO Gaussian model that is given by~\eqref{mimo-gaussian-model} and~\eqref{input-covariance-matrix-Gaussian-model} allowing the relay nodes to be fully aware of the users' codebooks (i.e., the non-constrained or non-oblivious setting) increases rates by at most a bounded constant (only !).
In other terms, restricting the relay nodes not to know/utilize the users' codebooks causes only a bounded rate loss in comparison with maximum rate that would be achievable in the non-oblivious setting. The constant depends on the network size, but is independent of the channel gain matrix, powers and noise levels. The result is an easy combination of a recent improved constant-gap result of Ganguly and Lim in~\cite{GangulyKim:ISIT:2017} (which tightens further that of Zhou \textit{et al.}~\cite{DBLP:journals/corr/ZhouX0C16}, see Remark~\ref{remark-on-constant-gap-result} below) with our Theorem~\ref{th:GaussSumCap}.

 For simplicity, we focus on the case in which $N_l=N$ for all $l \in \mc L$ and $M_k=M$ for all $k \in \mc K$. For the unconstrained case (i.e., with none of the constraints of obliviousness and Gaussian signaling assumed), the capacity region of the model described by~\eqref{mimo-gaussian-model} and~\eqref{input-covariance-matrix-Gaussian-model}, which we denote hereafter as $\mc C^{\text{uncons}}(C_{\mc K})$, is still to be found in general; and an easy outer bound on it is given by the maximum-flow min-cut bound, i.e., the set $\mc R^{\text{up}}(C_{\mc K})$ of all rate tuples $(R_1,\ldots, R_L)$ for which for all $\mathcal{T}\subseteq \mathcal{L}$ and $\mathcal{S}\subseteq\mathcal{K}$
\begin{align}
\sum_{t\in\mathcal{T}}R_{t} \leq&
 \sum_{k\in \mathcal{S}}C_k+ \log \frac{
|\sum_{k\in\mathcal{S}^{c}}\mathbf{H}_{k,\mathcal{T}}^{H}
\mathbf{\Sigma}_{k}^{-1}
\mathbf{H}_{k,\mathcal{T}}+\mathbf{K}^{-1}_{\mathcal{T}}|
}{
|\mathbf{K}_{\mathcal{T}}^{-1}|
}\label{eq:Cut-set}.
\end{align}

\noindent The following theorem shows that the rate-region of Theorem~\ref{th:GaussSumCap}  is within a constant gap from  $\mc R^{\text{up}}(C_{\mc K})$, and so from the capacity region of the unconstrained setting $C^{\text{uncons}}(C_{\mc K})$.

\begin{theorem}~\label{th:ConstantGapResult}
If $(R_1,\ldots, R_L) \in \mc C^{\text{uncons}}(C_{\mc K})$, then there exists a constant $\Delta \geq 0$ such that $(R_1-\Delta,\ldots, R_L-\Delta)\in \mc C_{G}(C_{\mc K})$, with
\begin{align}
\Delta\leq \begin{cases}
\frac{N}{2}(2.45+\log(\frac{KM}{N})), &\text{for } KM > 2N,\\
\frac{KM+N}{2}&\text{for } KM\leq 2N.
\end{cases}
\end{align}
\end{theorem}

\begin{remark}~\label{remark-on-constant-gap-result}
In the unconstrained case with no time-sharing, Zhou \textit{et al.} show in~\cite{DBLP:journals/corr/ZhouX0C16} (see Theorem 3 therein) that the rate region $\mc C_{\text{G}}^{ \text{no-ts}}(C_{\mc K})$ achievable with the scheme CF-JD with Gaussian input and Gaussian quantization is within a constant gap $\eta=(KM+N)$ of the capacity region $\mc C^{\text{uncons}}(C_{\mc K})$. Specifically, for any rate tuple $(R_1,\ldots, R_L) \in \mc R^{\text{up}}(C_{\mc K})$, the tuple $(R_1-\eta,\ldots, R_L-\eta) \in \mc C_{\text{G}}^{\text{no-ts}}(C_{\mc K})$. As we already mentioned, our Theorem~\ref{th:GaussSumCap} shows that under the constraint of Gaussian signaling and oblivious relay processing CF-JD is in fact optimal from a capacity viewpoint. Also, our Theorem~\ref{th:ConstantGapResult} improves the gap to the cut-set bound of~\cite[Theorem 3]{DBLP:journals/corr/ZhouX0C16}, which in our context can be interpreted as tightening the rate loss that is caused by restricting the relay nodes not to know/utilize the users' codebooks.
\end{remark}

\subsection{Numerical Results: Circular Symmetric Wyner Model for CRAN}

In this section, we evaluate and compare the performance of some oblivious and non-oblivious schemes for a simple Gaussian CRAN example, the circular symmetric Wyner model shown in Figure~\ref{fig:CW_C_Fix}. There are $K$ cells, with each cell containing a single-antenna user and a single antenna RU. Inter-cell interference takes place only between adjacent cells; and intra-cell and inter-cell channel gains are given by $1$ and $\gamma \in [0,1]$, respectively. All RUs have a fronthaul capacity of $C$. In this model, the channel output at RU or relay node $k\in \mc K$ is given by
\begin{align}
Y_k = \gamma X_{[k-1]_K} + X_k + \gamma X_{[k+1]_K} + N_k,
\end{align}
where $[\cdot]_K:=[\cdot]\!\!\!\mod K$,  $\mathrm{E}[|X_k|^2]\leq P$ and $N_k \sim \mc{CN}(0,1)$, for all $k\in \mc K$. For convenience, we write
$\dv Y = \dv H \dv X + \dv N$, where $\dv X=[X_1,\hdots,X_K]^T$, $\dv N=[N_1,\hdots,N_K]^T$ and $\dv H$ is the $K{\times}K$ matrix with the element $(k,l)$ given by
\begin{align}
h_{k,l} =
\begin{cases}
1& \text{if } l=k\\
\gamma & \text{if } k = [l+1]_K \text{ or } [l-1]_K\\
 0 & \text{otherwise}.
\end{cases}
\end{align}

\noindent Although seemingly simple, the capacity region of this model is still to be found in the case in which the relay nodes are not constrained, i.e., are allowed to perform non-oblivious processing. In what follows, we restrict to studying the maximum per-cell -sum-rates that are offered by various schemes, some of which use only oblivious relay processing and others not. A straightforward upper bound on those per-cell rates is given by the cut-set bound,
\begin{align}
R_{\text{cut-set}}(C) = \min\left\{  C ,\frac{1}{K} \log\det(\dv I + P  \dv H\dv H^H )  \right\}.
\end{align}

%----------------------------------------
\begin{figure}[t!]
\centering
\includegraphics[width=0.5\textwidth]{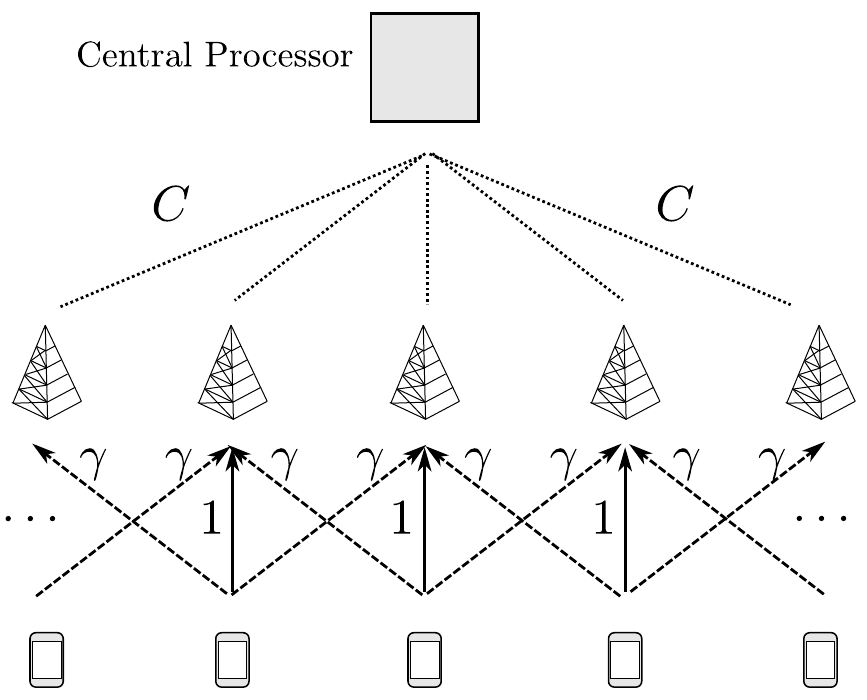}
\vspace{-2mm}
\caption{Circular Wyner model with $K$ users and $K$ relay nodes or remote units (RUs).}
\label{fig:CW_C_Fix}
\end{figure}

This  model is clearly an instance of the memoryless MIMO Gaussian CRAN described by~\eqref{mimo-gaussian-model} and~\eqref{input-covariance-matrix-Gaussian-model}. Thus, its performance, in terms of per-cell capacity  $C_{\text{G}}(C)$,  under oblivious relay processing with time-sharing of Gaussian inputs can be obtained easily using Theorem~\ref{th:GaussSumCap} as
\begin{equation}
C_{\text{G}}(C) = \max_{b_q,\alpha_q,P_q}\min_{\mc S\subseteq \mc K} \left\{ |\mc S|(C + \sum_{q=1}^{|\mc Q|}\log(1-b_q)) + \sum_{q=1}^{|\mc Q|}\log\det(\dv I + P_q b_q \dv H_{\mc S^c}\dv H_{\mc S^c}^H )  \right\}
\end{equation}
where $\dv H_{\mc S^c}$ is the submatrix of $\dv H$ composed by only those rows of $\dv H$ that are in the subset $\mc S^c$, and the maximization is over $0 \leq b_q \leq 1$, $0 \leq \alpha_q \leq 1$ and $P_q \geq 0$ such that $\sum_{q=1}^{|\mc Q|}\alpha_q = 1$ and $\sum_{q=1}^{|\mc Q|}\alpha_q P_q \leq  P$. If time-sharing is not enabled, i.e., $Q=\text{constant}$, $C_{\text{G}}(C)$ reduces to
\begin{align}
C_{\text{G}}^{\text{no-ts}}(C) = \max_{0\leq b\leq 1}\min_{\mc S\subseteq \mc K} \left\{ |\mc S|(C + \log(1-b)) + \log\det(\dv I + P b \dv H_{\mc S^c}\dv H_{\mc S^c}^H )  \right\}.
\end{align}

For non-oblivious schemes, we consider mainly the following two schemes:
\begin{enumerate}
\item \textit{Decode-and-Forward (DF)}: This scheme proposed in \cite{Sanderovich:2009:IT} is based on the fact that the output at each relay node can be seen as that of a three user Gaussian multiple-access channel. Relay $k$ decodes the message from user $k$ by either treating interference from users $[k-1]_K$ and $[k+1]_K$ as noise, or by jointly decoding all three messages. Then, it forwards message $k$ to the CP. This scheme yields the per-cell rate \cite{Sanderovich:2009:IT}
\begin{subequations}
\begin{align}
R_{\text{DF}}(C)& :=  \min\{\max\{R_{\text{tin}} , R_{\text{joint}}\}, C \}\\
R_{\text{tin}} &= \log\left(1+\frac{P}{1+2\gamma^2 P}\right)\\
R_{\text{joint}} &= \min \left\{\frac{1}{2} \log \left( 1 + 2\gamma^2 P \right), \frac{1}{3}\log ( 1 + (1+2\gamma^2 ) P) \right\}.
\end{align}
\end{subequations}

\item \textit{Compute-and-Forward (CoF)}: This scheme, proposed in \cite{Nazer:IT:2011}, is based on nested lattice codes. The users transmit using the same lattice code.  Then, each relay node  decodes
one equation (with integer-valued coefficients) that relates the users symbols and forwards that equation to the CP. If the collected $K$ equations are linearly independent, the CP can invert the system and obtain the transmitted symbols. For the studied example, this yields~\cite{Nazer:ISIT:2009}
\begin{equation}
R_{\text{CoF}}(C) =\min\left\{\max_{b_1,b_2\in \mc B} -\log\left( b_1^2 + 2b_2^2 - \frac{P(b_1+2\gamma b_2)^2}{1+P(1+2\gamma^2)} \right),C\right\},
\end{equation}
where the set $\mc B$ is given by
\vspace{-2mm}
%\begin{equation}
$\mc B = \{(b_1,b_2): b_1,b_2\in \mathds{Z}, b_1\neq 0, b_1^2+ 2b_2^2\leq 1+P(1+2\gamma^2)\}$.
%\end{equation}
\end{enumerate}

%----------------------------------------
\begin{figure}[t!]
\centering
\includegraphics[width=0.6\textwidth]{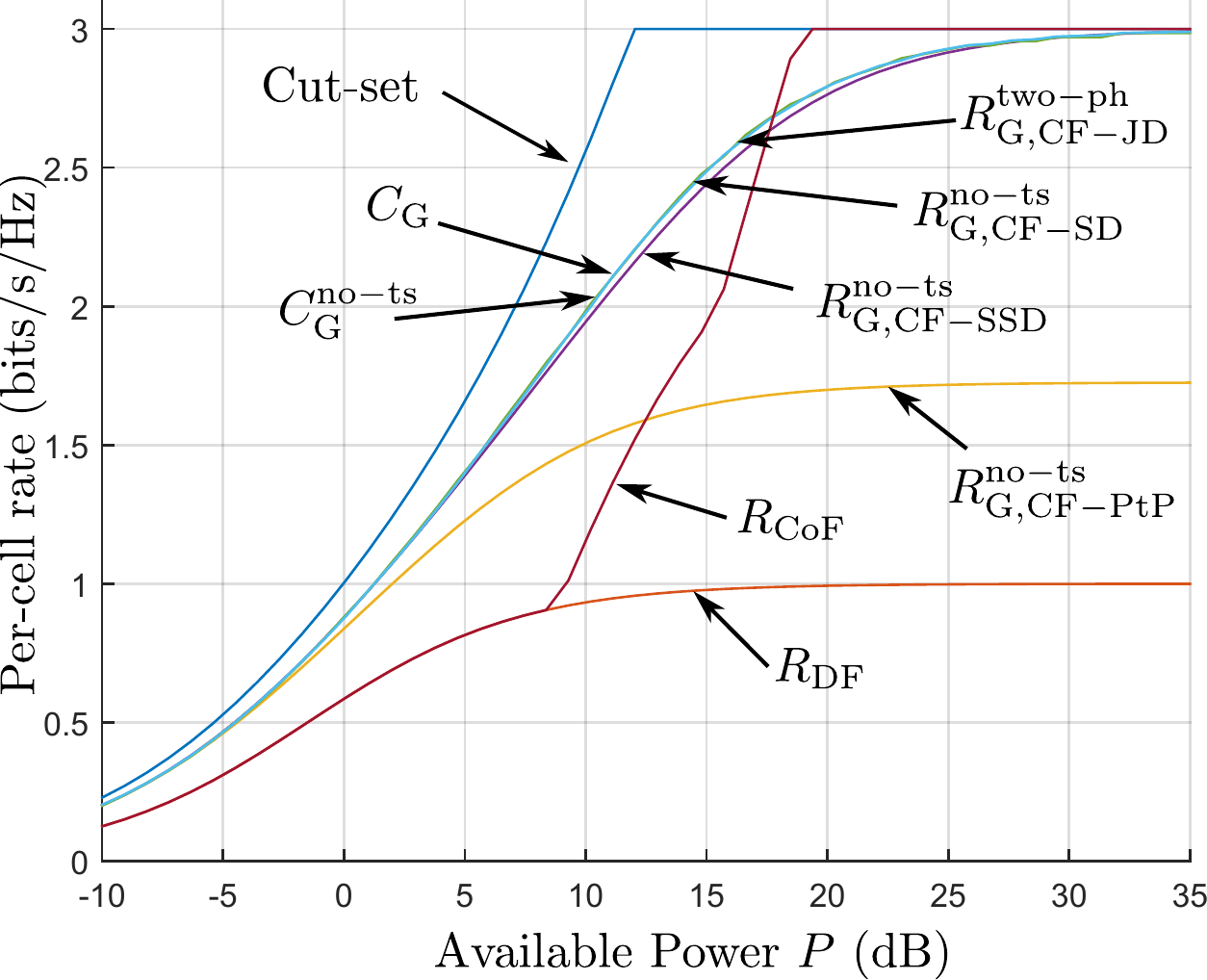}
\vspace{-2mm}
\caption{Bounds on the per-cell rate for the circular symmetric Wyner model of Figure~\ref{fig:CW_C_Fix}, as well as the per-cell capacity under time-sharing of Gaussian signaling. Numerical values are $K=3$, $\gamma = 1/\sqrt{2}$ and $C=3.5$.}
\label{fig:CW_C_Fix}
\end{figure}

\vspace{0.2cm}

\noindent For comparison reasons, we also consider the following oblivious schemes:
\begin{enumerate}
\item \textit{CF-JD with $|\mc Q|=2$:} It is easy to see that the per-cell sum-rate achievable using the CF-JD scheme with time-sharing between two phases in which users and relays are active during the first phase and remain silent in the second as in Example~\ref{ex:Example1} is given by
\begin{align}
R_{\text{G,CF-JD}}^{\text{two-ph}}(C) & \max_{0\leq \alpha\leq 1}\max_{0\leq b\leq 1}\min_{\mc S\subseteq \mc K} \alpha \cdot \left\{ |\mc S|\left(\frac{C}{\alpha}+\log(1-b)\right) + \log\det\left(\dv I + \frac{P}{\alpha} b\dv H_{\mc S^c}\dv H_{\mc S^c}^H \right)   \right\}.\nonumber
\end{align}
\item \textit{CF-SD without time-sharing:} The per-cell rate achievable by CF-SD without time-sharing and Gaussian test channels $U_k\sim \mc {CN}(Y_k, \sigma_*^2)$, $k\in \mc K$ follows from Proposition~\ref{prop:CFSD_achi} as
\begin{align}
R_{\text{CF-SD}}^{\text{no-ts}}(C) &=  \log\det(\dv I + P (1+\sigma^2_*)^{-1}\dv H \dv H^H)
\end{align}
where $\sigma_*^2$ is the unique solution of the equation $KC = \log\det(\dv I + (1/\sigma_*^2)(P \dv H \dv H^H + \dv I))$.

\item \textit{CF-SSD without time-sharing:}
The per-cell rate achievable by CF-SSD without \mbox{time-sharing} and Gaussian test channels $U_k\sim \mc {CN}(Y_k, \sigma_k^2)$, $k\in \mc K$ follows from Proposition~\ref{prop:CFSWZ_achi}, as
\begin{align}
R_{\text{CF-SSD}}^{\text{no-ts}}(C)=  \log\det(\dv I + P \dv D\dv H \dv H^H), \label{eq:CWSSD}
\end{align}
where $\dv D = \mathrm{diag}(1/(1+ \sigma_k^2), k\in \mc K)$ with $\sigma_k^2 = \sigma^2_{Y_k|Y_1^{k-1}}/(2^{C}-1)$; where $\sigma^2_{Y_k|U_1^{k-1}}$ corresponds to the MMSE error of estimating $Y_k$ from $U_1^k$,  given by
\begin{equation}
\sigma^2_{Y_k|Y_1^{k-1}} = (1+2\alpha^2)P +1 - \dv h_k \dv H_{[1:k-1]}( P \dv H_{[1:k-1]}\dv H_{[1:k-1]}^H + \dv I +\text{diag}(\sigma_{[1:k-1]}^{2}) )^{-1}\dv H_{[1:k-1]}^H \dv h_k^H.\nonumber
\end{equation}

\item \textit{CF-PtP without time-sharing:}
A simplified version of CF-SSD, to which we refer as ``Compress-and-Forward with Point-to-Point compression" (CF-PtP), is one in which each relay node compresses its channel output using standard compression, i.e., without binning. The per-cell rate $R_{\text{CF-PtP}}^{\text{no-ts}}(C)$ allowed by this scheme is given as in~\eqref{eq:CWSSD} with
\begin{equation}
\dv D = (2^{C}-1)/ (2^{C}+P(1+2\gamma^2))\dv I.
\end{equation}
\end{enumerate}

\noindent Figure~\ref{fig:CW_C_Fix} depicts the evolution of the per-cell rates obtained using the above discussed oblivious and non-oblivious schemes, as well as the cut-set bound, for numerical values $K=3$, $\gamma = 1/\sqrt{2}$ and $C=3.5$, as function of the user transmit power $P$. As it can be seen from the figure, for this example the loss in performance, in terms of per-cell rate,  that is caused by constraining the relay nodes to implement only oblivious operations is less than $1.7743$ bits.  Also, time-sharing is generally beneficial, in the sense that the discussed oblivious schemes generally suffer some (small) rate-loss when constrained not to employ time-sharing.

%----------------------------------------
\begin{figure}[t!]
\centering
\includegraphics[width=0.6\textwidth]{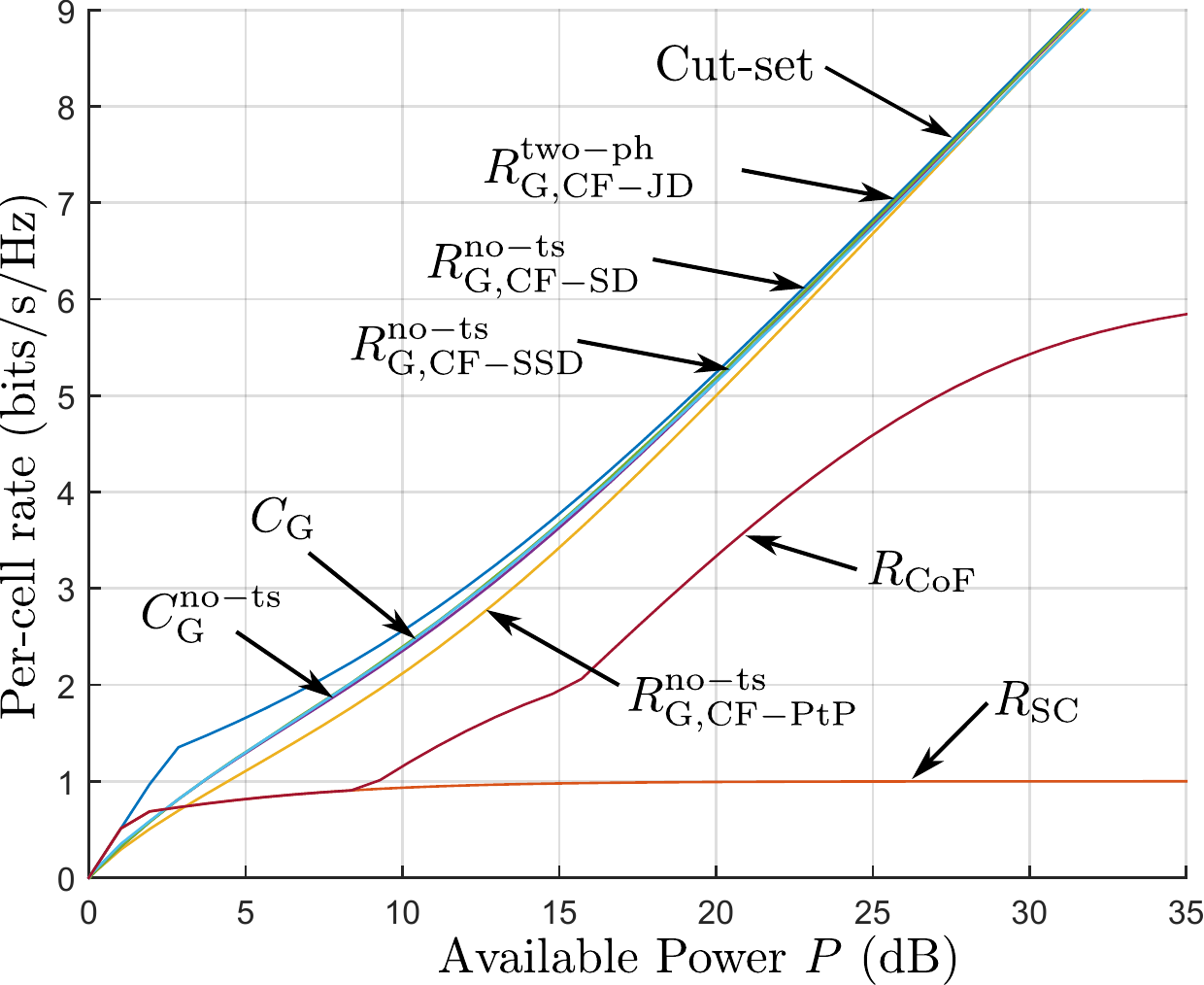}
\vspace{-2mm}
\caption{Degrees of freedom offered by some oblivious and non-oblivious for an example circular symmetric Wyner model of Figure~\ref{fig:CW_C_Fix} with $K=3$ and $\gamma = 1/\sqrt{2}$. The fronthaul capacity scales as $C=5\log_{10}(P)$.}
\label{fig:CW_C_SNR}
\end{figure}

\noindent Figure~\ref{fig:CW_C_SNR} shows how the rates offered by the aforementioned oblivious and non-oblivious schemes scale with the signal-to-noise ratio, when the available per-link fronthaul capacity scales logarithmically with the available user transmit power as $C=5\log_{10}(P)$. As the figure illustrates, in opposition with non-oblivious schemes such as decode-and-forward and compute-and-forward, oblivious processing also has the advantage to cause no loss in terms of degrees of freedom.

\vspace{0.2cm}

\section{Concluding Remarks}

We close this paper with some concluding remarks.  Our results shed light (and sometimes determine exactly) what operations the relay nodes should perform optimally in the case in which transmission over a cloud radio access network is under the framework of oblivious processing at the relays, i.e., the relays are not allowed to know, or cannot acquire, the users' codebooks. In particular, perhaps non-surprisingly, it is shown that compress-and-forward, or variants of it, generally perform well in this case, and are optimal when the outputs at the relay nodes are conditionally independent on the users inputs. Furthermore, in addition to its relevance from a practical viewpoint, restricting the relays not to know/utilize the users' codebooks causes only a bounded rate loss in comparison with the non-oblivious setting (e.g., compress-and-forward and noisy network coding perform to within a constant gap from the cut-set bound in the Gaussian case).

Finally, leveraging on the now known connection of the information bottleneck method (IB)~\cite{Tishby99theinformation} (see \cite{Witsenhausen1975CondEntropy} and \cite{EC98} for an earlier equivalent formulation of the IB problem in the context of source coding and investment theory, respectively) with the CEO source coding problem with logarithmic loss and one that can be established with the CRAN channel coding problem with oblivious relay processing, we note that the results of this paper, and the proof techniques, translate easily into analogous ones for the problem of distributed information bottleneck. In this problem, multiple sensors compress separately their observations in a manner that, collectively, the compressed signals provide as much information as possible about a remote (or \textit{hidden}) source. On this aspect, the reader may refer to~\cite{EZ-IZS18-DIB} and~\cite{aguerri2018distributed} where a full characterization of the optimal tradeoffs among the minimum description lengths at which the features are described (i.e., \textit{complexity}) and the information that the latent variables collectively preserve about the target variable (i.e., \textit{accuracy} or \textit{relevant information}) are established for both DM and Gaussian models, together with Blahut-Arimoto type algorithms and neural network based representation learning algorithms, that allow to compute optimal tradeoffs. The results of~\cite{EZ-IZS18-DIB, aguerri2018distributed} generalize those for the single user DM IB problem~\cite{Tishby99theinformation} and the single-user scalar \cite{EC98} and vector Gaussian IB problem~\cite{journals/jmlr/ChechikGTW05} to the distributed scenario. Since the single-encoder IB method has found application in various contexts of learning and prediction~\cite{C-BL06}, such as word clustering for text classification~\cite{ST01}, community detection~\cite{LYFL15}, neural code analysis~\cite{BM10}, speech recognition~\cite{HT05} and others, distributed IB methods clearly finds usefulness in the extensions of those applications to the distributed case.

Among interesting problems that are left unaddressed in this paper that of characterizing optimal input distributions under rate-constrained compression at the relays where, e.g., discrete signaling is already known to sometimes outperform Gaussian signaling for single-user Gaussian CRAN~\cite{Sanderovich2008:IT:ComViaDesc}. Alternatively, one may consider finding the worst-case noise under given input distributions, e.g., Gaussian, and rate-constrained compression at the relays. Also, although it is still not clear whether the known multiaccess/broadcast (MAC/BC) duality extends to one between uplink and downlink CRAN models in general~\cite{liu2016uplink}, it is expected that the approach of this paper be instrumental towards characterizing the effect of the relay nodes being oblivious to the actual codebooks used by the users in the downlink setting, especially in the case in which the connection between the CP and the relay nodes are not wired.

%\vspace{-0.1cm}

%\onecolumn
\begin{appendices}

\section{Proof of  Theorem~\ref{th:MK_C_Main}}\label{app:ConverseClass}

\subsection{Proof of Direct Part of Theorem~\ref{th:MK_C_Main}}
We derive the rate region achievable by the CF-JD scheme for the class of DM CRAN models satisfying \eqref{eq:MKChain_pmf} using the inner bound derived in Theorem~\ref{th:NNC_all_MK_inner} for the general DM CRAN model.
It follows from Theorem~\ref{th:NNC_all_MK_inner} that the rate region in Theorem~\ref{th:MK_C_Main} is achievable by noting that for the class of DM CRAN models satisfying \eqref{eq:MKChain_pmf}, we have
\begin{align}
I(Y_{\mc S};U_{\mathcal{S}}|X_{\mathcal{L}},U_{\mathcal{S}^c},Q)  = \sum_{s\in \mathcal{S}}I(Y_{s};U_{s}|X_{\mathcal{L}},Q),\label{eq:innerMK}
\end{align}
where \eqref{eq:innerMK} follows due to the Markov chains (given $Q$), $U_k\mkv Y_k \mkv X_{\mathcal{L}} \mkv Y_{\mathcal{K}/k}-U_{\mathcal{K}/k}$, for $k\in \mc K$.  This concludes the proof.

\subsection{Proof of Converse Part of Theorem~\ref{th:MK_C_Main}}
Assume the rate tuple $(R_{1},\ldots,R_L)$ is achievable. Let $\mathcal{T}\subseteq \mc L$, $\mathcal{S}\subseteq \mc K$, with $\mc T,\mc S\neq \emptyset$, and $J_k:= \phi^r_k(Y_k^n,Q^n)$ be the message sent by relay $k$,  $k\in \mathcal{K}$,  $F_{\mathcal L}$ be the codebook indices, and let $\tilde{Q}:=Q^n$ be the time-sharing variable. For simplicity, let
$X_{\mathcal{L}}^n := (X_{1}^n,\ldots,X_{L}^n )$, $R_{\mathcal{T}}:= \sum_{t\in \mathcal{T}}R_t$ and $C_{\mathcal{S}}:= \sum_{k\in \mathcal{S}}C_k$.
Define
\begin{align}\label{eq:Auxiliary}
U_{i,k}:= (J_k,Y_k^{i-1}) \quad \text{and }\quad \bar{Q}_i:= (X_{\mathcal{L}}^{i-1},X_{\mathcal{L},i+1}^n,\tilde{Q}).
\end{align}

From Fano's inequality, we have with $\epsilon_n\rightarrow 0$ for $n\rightarrow \infty$, for all $\mathcal{T}\subseteq \mathcal{L}$,
\begin{align}
H(m_{\mathcal{T}}|J_{\mathcal{K}},F_{\mathcal{L}},\tilde{Q}) \leq H(m_{\mathcal{L}}|J_{\mathcal{K}},F_{\mathcal{L}},\tilde{Q})\leq n \epsilon_n.\label{eq:MK_C_Fano}
\end{align}

\noindent We start by showing the following inequality which will be instrumental in the rest of this proof.
\begin{align}
H(X^n_{\mathcal{T}}|X^n_{\mathcal{T}^c},J_{\mathcal{K}},\tilde{Q})&\leq \sum_{i=1}^nH(X_{\mathcal{T},i}|X_{\mathcal{T}^c,i},\bar{Q}_i) -nR_{\mathcal{ T}}
  := n\Gamma_{\mathcal{T}}  .\label{eq:MK_C_BoundGamma}
\end{align}
Inequality \eqref{eq:MK_C_BoundGamma} can be shown as follows.
\begin{align}
%nR_{\mathcal{T}}n&
nR_{\mathcal{T}} =& H(m_{\mathcal{T}})\label{eq:MK_C_Ineq_00}\\
=& I(m_{\mathcal{T}};J_{\mathcal{K}},F_{\mathcal{L}},\tilde{Q})+ H(m_{\mathcal{T}}|J_{\mathcal{K}},F_{\mathcal{L}},\tilde{Q})\\
=& I(m_{\mathcal{T}};J_{\mathcal{K}},F_{\mathcal{T}}|F_{\mathcal{T}^c},\tilde{Q})+ H(m_{\mathcal{T}}|J_{\mathcal{K}},F_{\mathcal{L}},\tilde{Q})\label{eq:MK_C_Ineq_0}\\
\leq& I(m_{\mathcal{T}};J_{\mathcal{K}},F_{\mathcal{T}}|F_{\mathcal{T}^c},\tilde{Q})+ n\epsilon_n\label{eq:MK_C_Fano_sub}\\
%=& H(J_{\mathcal{K}},F_{\mathcal{T}}|F_{\mathcal{T}^c},\tilde{Q}) - H(J_{\mathcal{K}},F_{\mathcal{T}}|F_{\mathcal{T}^c},m_{\mathcal{T}},\tilde{Q})+ n\epsilon_n\nonumber\\
=& H(J_{\mathcal{K}}|F_{\mathcal{T}^c},\tilde{Q})+H(F_{\mathcal{T}}|F_{\mathcal{T}^c},J_{\mathcal{K}},\tilde{Q}) + n\epsilon_n\\
&- H(F_{\mathcal{T}}|F_{\mathcal{T}^c},m_{\mathcal{T}},\tilde{Q}) - H(J_{\mathcal{K}}|F_{\mathcal{T}^c},m_{\mathcal{T}},F_{\mathcal{T}},\tilde{Q})\\
 =& I(m_{\mathcal{T}},F_{\mathcal{T}};J_{\mathcal{K}}|F_{\mathcal{T}^c},\tilde{Q})-I(F_{\mathcal{T}};J_{\mathcal{K}}|F_{\mathcal{T}^c},\tilde{Q})+ n\epsilon_n \label{eq:MK_C_indep}\\\
\leq &  I(m_{\mathcal{T}},F_{\mathcal{T}};J_{\mathcal{K}}|F_{\mathcal{T}^c},\tilde{Q})+ n\epsilon_n\\
\leq &  I(X_{\mathcal{T}}^n;J_{\mathcal{K}}|F_{\mathcal{T}^c},\tilde{Q})+ n\epsilon_n\label{eq:MK_C_dataproc}\\
 = & H(X_{\mathcal{T}}^n|F_{\mathcal{T}^c},\tilde{Q})-H(X_{\mathcal{T}}^n|F_{\mathcal{T}^c},J_{\mathcal{K}},\tilde{Q})+ n\epsilon_n\label{eq:MK_C_Ineq_10}\\
 % = & H(X_{\mathcal{T}}^n|X_{\mathcal{T}^c}^n,F_{\mathcal{T}^c},\tilde{Q})-H(X_{\mathcal{T}}^n|F_{\mathcal{T}^c},J_{\mathcal{K}},\tilde{Q})+ n\epsilon_n
  \leq & H(X_{\mathcal{T}}^n|X_{\mathcal{T}^c}^n,\tilde{Q})-H(X_{\mathcal{T}}^n|X_{\mathcal{T}^c}^n,F_{\mathcal{T}^c},J_{\mathcal{K}},\tilde{Q})+ n\epsilon_n\label{eq:MK_C_Ineq_11}\\
  = & H(X_{\mathcal{T}}^n|X_{\mathcal{T}^c}^n,\tilde{Q})-H(X_{\mathcal{T}}^n|X_{\mathcal{T}^c}^n,J_{\mathcal{K}},\tilde{Q})+ n\epsilon_n,\label{eq:MK_C_Ineq_12}
\end{align}
where \eqref{eq:MK_C_Ineq_00} follows since $m_{\mathcal{T}}$ are independent; \eqref{eq:MK_C_Ineq_0} follows since $m_{\mathcal{T}}$ is independent of $\tilde{Q}$ and $F_{\mathcal{T}^c}$; \eqref{eq:MK_C_Fano_sub} follows from \eqref{eq:MK_C_Fano}; \eqref{eq:MK_C_indep} follows since $m_{\mathcal{T}}$ is independent of $F_{\mathcal{L}}$; \eqref{eq:MK_C_dataproc} follows from the data processing inequality;\eqref{eq:MK_C_Ineq_11} follows since $X_{\mathcal{T}^c}^n,F_{\mathcal{T}^c}$ are independent from $X_{\mathcal{T}}^n$ and since conditioning reduces entropy and; \eqref{eq:MK_C_Ineq_12} follows  due to the Markov chain
\begin{align}
X_{\mathcal{T}}^n \mkv (X_{\mathcal{T}^c}^n,J_{\mathcal{K}},\tilde{Q}) \mkv F_{\mathcal{T}^c}.
\end{align}
Then, from \eqref{eq:MK_C_Ineq_12} we have \eqref{eq:MK_C_BoundGamma} as follows:
\begin{align}
H(X_{\mathcal{T}}^n|X_{\mathcal{T}^c}^n,J_{\mathcal{K}},\tilde{Q})
%&\leq  H(X_{\mathcal{T}}^n|X_{\mathcal{T}^c}^n,\tilde{Q}) -nR_{\mathcal{T}}- n\epsilon_n\\
&\leq \sum_{i=1}^{n}H(X_{\mathcal{T},i}|X_{\mathcal{T}^c}^n, X_{\mathcal{T}}^{i-1},\tilde{Q}) - nR_{\mathcal{T}}\label{eq:MK_C_Ineq_0_1}\\
&= \sum_{i=1}^{n}H(X_{\mathcal{T},i}|X_{\mathcal{T}^c,i},X_{\mathcal{L}}^{i-1},X_{\mathcal{L},i+1}^n,\tilde{Q}) - nR_{\mathcal{T}}\label{eq:MK_C_Ineq_0_2}\\
&= \sum_{i=1}^{n} H(X_{\mathcal{T},i}|X_{\mathcal{T}^c,i},\bar{Q}_i) -nR_{\mathcal{T}} = n\Gamma_{\mathcal{T}},
\end{align}
where \eqref{eq:MK_C_Ineq_0_2} is due to Lemma \ref{lem:IIDinput}.

 We pause to mention that, for a subset $\mc T \subseteq \mc L$, inequality~\ref{eq:MK_C_BoundGamma} provides a lower bound on the term in the RHS of it in terms of a conditional entropy term; and, as such, it is reminiscent of the result of~\cite[Lemma 1]{Courtade2014LogLoss} which states that for the CEO problem with logarithmic loss fidelity measure the expected distortion admits a lower bound in the form of a conditional entropy term, namely the entropy of the remote source conditioned on the CEO's inputs.

\noindent Continuing from \eqref{eq:MK_C_Ineq_12}, we have
\begin{align}
 nR_{\mathcal{T}} \leq &
 %  H(X_{\mathcal{T}}^n|X_{\mathcal{T}^c}^n,\tilde{Q})-H(X_{\mathcal{T}}^n|X_{\mathcal{T}^c}^n, J_{\mathcal{K}},\tilde{Q})+ n\epsilon_n\\
%
%=
\sum_{i=1}^n H(X_{\mathcal{T},i}|X_{\mathcal{T}^c}^n,\tilde{Q}, X_{\mathcal{T}}^{i-1})
%\nonumber\\
-H(X_{\mathcal{T},i}^n|X_{\mathcal{T}^c}^n, J_{\mathcal{K}},X_{\mathcal{T}}^{i-1},\tilde{Q})+ n\epsilon_n\\
=& \sum_{i=1}^n H(X_{\mathcal{T},i}|X_{\mathcal{T}^c}^n, \tilde{Q}, X_{\mathcal{T}}^{i-1}, X_{\mathcal{T},i+1}^{n})
%\nonumber\\&
-H(X_{\mathcal{T},i}|X_{\mathcal{T}^c}^n, J_{\mathcal{K}},X_{\mathcal{T}}^{i-1},\tilde{Q})+ n\epsilon_n\label{eq:MK_C_ineq_2}\\
\leq&
%\sum_{i=1}^n H(X_{\mathcal{T},i}|X_{\mathcal{T}^c,i},\bar{Q}_i)\label{eq:MK_C_ineq_3}\nonumber\\
%&-H(X_{\mathcal{T},i}|X_{\mathcal{T}^c,i},J_{\mathcal{K}}, Y_{\mathcal{K}}^{i-1},X_{\mathcal{L}}^{i-1},X_{\mathcal{L},i+1}^{n},\tilde{Q})+ n\epsilon_n\\
%
%=&
 \sum_{i=1}^n H(X_{\mathcal{T},i}|X_{\mathcal{T}^c,i},\bar{Q}_i)
% \nonumber\\&
 -H(X_{\mathcal{T},i}|X_{\mathcal{T}^c,i}, U_{\mathcal{K},i},\bar{Q}_{i})+ n\epsilon_n\label{eq:MK_C_ineq_3}\\
=& \sum_{i=1}^n I(X_{\mathcal{T},i};U_{\mathcal{K},i}|X_{\mathcal{T}^c,i},\bar{Q}_{i})+ n\epsilon_n\label{eq:MK_C_ineq_4},
\end{align}
where \eqref{eq:MK_C_ineq_2} follows due to Lemma \ref{lem:IIDinput};  and \eqref{eq:MK_C_ineq_3} follows since conditioning reduces entropy.

On the other hand, we have the following equality
\begin{align}
I(Y_{\mathcal{S}}^n;J_{\mathcal{S}}|X_{\mathcal{L}}^n, J_{\mathcal{S}^c},\tilde{Q})
&= \sum_{k\in\mathcal{S}}I(Y_k^n;J_{k}|X_{\mathcal{L}}^n,\tilde{Q})\label{eq:MK_C_Ineq_2_2}\\
&=\sum_{k\in\mathcal{S}}\sum_{i=1}^nI(Y_{k,i};J_k|X_{\mathcal{L}}^n,Y_{k}^{i-1},\tilde{Q})\\
&=\sum_{k\in\mathcal{S}}\sum_{i=1}^nI(Y_{k,i};J_k,Y_{k}^{i-1}|X_{\mathcal{L}}^n,\tilde{Q})\label{eq:MK_C_Ineq_2_4}\\
&=  \sum_{k\in\mathcal{S}}\sum_{i=1}^nI(Y_{k,i};U_{k,i}|X_{\mathcal{L},i},\bar{Q}_i), \label{eq:MK_C_Ineq_2_41}
\end{align}
where \eqref{eq:MK_C_Ineq_2_2} follows due to the Markov chain, for $k\in \mc K$,
\begin{align}
J_k \mkv Y_k^n \mkv X_{\mathcal{L}}^n \mkv Y_{\mathcal{S}\setminus k}^n \mkv J_{\mathcal{S}\setminus k},
\end{align}
 and since $J_k$ is a function of $Y_k^n$; and \eqref{eq:MK_C_Ineq_2_4} follows due to the Markov chain $Y_{k,i}-X_{\mathcal{L}}^n-Y_{k}^{i-1}$ which follows since the channel is memoryless.

Then, from the relay side we have, for $\mc S\neq \emptyset$
\begin{align}
%n\sum_{k\in \mathcal{S}}C_k
nC_{\mathcal{S}}\geq& \sum_{k\in \mathcal{S}}H(J_k)
\geq  H(J_{\mathcal{S}})\\
%\geq & H(J_{\mathcal{S}}|X_{\mathcal{T}^c}^n,J_{\mathcal{S}^c},\tilde{Q})\\
\geq& I(Y_{\mathcal{S}}^n;J_{\mathcal{S}}|X_{\mathcal{T}^c}^n,J_{\mathcal{S}^c},\tilde{Q})\\
=& I(X_{\mathcal{T}}^n,Y_{\mathcal{S}}^n;J_{\mathcal{S}}|X_{\mathcal{T}^c}^n,J_{\mathcal{S}^c},\tilde{Q})\label{eq:MK_C_Ineq_2_1}\\
%=& I(X_{\mathcal{T}}^n;J_{\mathcal{S}}|X_{\mathcal{T}^c}^n,J_{\mathcal{S}^c},\tilde{Q})+I(Y_{\mathcal{S}}^n;J_{\mathcal{S}}|X_{\mathcal{L}}^n, J_{\mathcal{S}^c},\tilde{Q})\\
%
=& H(X_{\mathcal{T}}^n|X_{\mathcal{T}^c}^n,J_{\mathcal{S}^c},\tilde{Q})-H(X_{\mathcal{T}}^n|X_{\mathcal{T}^c}^n, J_{\mathcal{K}},\tilde{Q})
%\nonumber\\&
+ I(Y_{\mathcal{S}}^n;J_{\mathcal{S}}|X_{\mathcal{L}}^n, J_{\mathcal{S}^c},\tilde{Q})\\
\geq& H(X_{\mathcal{T}}^n|X_{\mathcal{T}^c}^n,J_{\mathcal{S}^c},\tilde{Q})-n\Gamma_{\mathcal{T}}
%\nonumber\\&
+ I(Y_{\mathcal{S}}^n;J_{\mathcal{S}}|X_{\mathcal{L}}^n, J_{\mathcal{S}^c},\tilde{Q})\label{eq:MK_C_Ineq_2_3}\\
%
%=& \sum_{i=1}^nH(X_{\mathcal{T},i}|X_{\mathcal{T}^c}^n,J_{\mathcal{S}^c},X_{\mathcal{T}}^{i-1},\tilde{Q})-n\Gamma_{\mathcal{T}}
%%\nonumber\\&
%+I(Y_{\mathcal{S}}^n;J_{\mathcal{S}}|X_{\mathcal{L}}^n, J_{\mathcal{S}^c},\tilde{Q})\\
%%
\geq & \sum_{i=1}^nH(X_{\mathcal{T},i}|X_{\mathcal{T}^c,i},U_{\mathcal{S}^c,i},\bar{Q}_i)-n\Gamma_{\mathcal{T}}
%\nonumber\\&
+ I(Y_{\mathcal{S}}^n;J_{\mathcal{S}}|X_{\mathcal{L}}^n, J_{\mathcal{S}^c},\tilde{Q})\label{eq:MK_C_Ineq_2_5}\\
%
%=&\sum_{i=1}^nH(X_{\mathcal{T},i}| X_{\mathcal{T}^c,i},U_{\mathcal{S}^c,i},\bar{Q}_i)- H(X_{\mathcal{T},i}|X_{\mathcal{T}^c,i},\bar{Q}_i)\nonumber\\
%& +nR_{\mathcal{T}}  + \sum_{k\in\mathcal{S}}\sum_{i=1}^nI(Y_{k,i};U_{k,i}|X_{\mathcal{L},i},\bar{Q}_i)\label{eq:MK_C_Ineq_3_0}\\
%
=&nR_{\mathcal{T}} - \sum_{i=1}^nI(X_{\mathcal{T},i};U_{\mathcal{S}^c,i}|X_{\mathcal{T}^c,i},\bar{Q}_i)
%\nonumber\\ &
+\sum_{k\in\mathcal{S}}\sum_{i=1}^nI(Y_{k,i};U_{k,i}|X_{\mathcal{L},i},\bar{Q}_i),\label{eq:MK_C_Ineq_3_1}
\end{align}
where~\eqref{eq:MK_C_Ineq_2_1} follows since $J_{\mathcal{S}}$ is a function of $Y_{\mathcal{S}}^n$;  \eqref{eq:MK_C_Ineq_2_3} follows from \eqref{eq:MK_C_BoundGamma};  \eqref{eq:MK_C_Ineq_2_5} follows since conditioning reduces entropy; and \eqref{eq:MK_C_Ineq_3_1} follows from \eqref{eq:MK_C_BoundGamma} and \eqref{eq:MK_C_Ineq_2_41}.

Note that, in general,  $\bar{Q}_i$ is not independent of $X_{\mathcal{L},i},Y_{\mathcal{S},i}$, and that due to Lemma \ref{lem:IIDinput}, conditioned on $\bar{Q}_i$, we have the Markov chain
\begin{align}
U_{k,i}-Y_{k,i}-X_{\mathcal{L},i}-Y_{{\mathcal{K}\setminus k},i}-U_{{\mathcal{K}\setminus k},i}.
\end{align}

%\iffalse
Finally, we define the standard time-sharing variable $Q'$  uniformly distributed over $\{1,\ldots, n\}$, $X_{\mathcal{L}} := X_{\mathcal{L},Q'}$, $Y_k:= Y_{k,Q'}$, $U_k := U_{k,Q'}$ and $Q := [\bar{Q}_{Q'},Q']$ and
we have from \eqref{eq:MK_C_ineq_4}, for $\mc S = \emptyset$
\begin{align}
nR_{\mathcal{T}}&\leq \sum_{i=1}^n I(X_{\mathcal{T},i};U_{\mathcal{K},i}|X_{\mathcal{T},i},\bar{Q}_{i})+ n\epsilon_n\\
&= nI(X_{\mathcal{T},Q'};U_{\mathcal{K},Q'}|X_{\mathcal{T}^c,Q'},\bar{Q}_{Q'},Q')+ n\epsilon_n\\
&= nI(X_{\mathcal{T}};U_{\mathcal{K}}|X_{\mathcal{T}^c},Q)+ n\epsilon_n,
\end{align}
and similarly, from \eqref{eq:MK_C_Ineq_3_1}, we have for $\mc S\neq \emptyset$
\begin{align}
R_{\mathcal{T}}\leq
 C_{\mathcal{S}}  - \sum_{k\in\mathcal{S}}I(Y_{k};U_{k}|X_{\mathcal{L}},Q) + I(X_{\mathcal{T}};U_{\mathcal{S}^c}|X_{\mathcal{T}^c},Q).\nonumber
\end{align}
This completes the proof of Theorem~\ref{th:MK_C_Main}.\qed
%\fi

%\noindent Finally, a standard time-sharing argument completes the proof of Theorem~\ref{th:MK_C_Main}.
%\qed

%\vspace{-0.3cm}

%%%%%%%%%%%%%%%%%%%%%%%%%%%%%%%%%%%%%%%%%%%%%%%%%%%%%%%%%%%%%%%%%%%%%%%%%%%%%%%%%%%%%%%%%%%%%%%%%%%%%%%%%%%%%%%%%%%%%%%%%%%%%%%%%%%%%%%%%%%%%%%%%%%%%%%%%%%%%%%%%%%%%%%%%%
%%%%%%%%%%%%%%%%%%%%%%%%%%%%%%%%%%%%%%%%%%%%%%%%%%%%%%%%%%%%%%%%%%%%%%%%%%%%%%%%%%%%%%%%%%%%%%%%%%%%%%%%%%%%%%%%%%%%%%%%%%%%%%%%%%%%%%%%%%%%%%%%%%%%%%%%%%%%%%%%%%%%%%%%%%

\section{Proof of the Inner Bound in Theorem \ref{th:NNC_all_MK_inner}}\label{app:NNC_all_MK_inner}
The scheme CF-JD employed in Theorem~\ref{th:NNC_all_MK_inner} for the general DM CRAN model generalizes~\cite[Theorem 3]{Sanderovich2008:IT:ComViaDesc} to the case of multiple users and enabled time-sharing. An outline of this scheme is as follows.
User $l$, $l\in \mc L$, sends $X^n_l(m_l,f_l,q^n)$, where $m_l\in [1\!:\!2^{nR_l}]$ is the users' message, $f_l\in [1\!:\!|\mc X_l|^{2^{nR_l}}]$ is the codebook index and $q^n\in \mc Q$ is the time-sharing sequence. Relay node $k$, $k \in \mc K$, compresses its channel output $Y_k^n$ into a description $U_k^n$ of compression rate $\hat{R}_k$ indexed by $i_k\in [1\!:\!2^{n\hat{R}_k}]$. The descriptions are randomly binned into $2^{nC_k}$ bins, indexed by a Wyner-Ziv bin index $j_k\in [1\!:\!2^{nC_k}]$. Relay node $k$ forwards the bin index $j_k$ of the bin containing the description $U_k^n$ to the CP over the error-free link. The CP receives $(j_1,\ldots, j_K)$ and decodes jointly the compression indices and the transmitted messages, i.e., it jointly recovers the indices $(m_1,\ldots, m_L,i_1,\ldots,i_K)$. The detailed proof is as follows.

Fix $\delta>0$, non-negative rates $R_1,\ldots, R_K$ and a joint pmf that factorizes as
\begin{align}
p(q,x_{\mathcal{L}},y_{\mathcal{K}},u_{\mathcal{K}})= p(q)\prod_{l=1}^{L}p(x_l|q) ~ p(y_{\mathcal{K}}|x_{\mathcal{L}})\prod_{k=1}^{K}p(u_k|y_k,q).
\end{align}

\noindent\textit{Codebook Generation:} Randomly generate a time-sharing sequence $q^n$ according to $\prod_{i=1}^np_{Q}(q_i)$. For user $l$, $l\in \mc L$ and every codebook index $F_l$, randomly generate a codebook $\mathcal{C}_l(F_l)$ consisting of a collection of $2^{nR_l}$ independent codewords $\{x_{l}^n(m_l,f_l,q^n)\}$ indexed with $m_l\in [1\!:\!2^{nR_l}]$, where $x_{l}^n(m_l,f_l,q^n)$ has its elements generated i.i.d. according to $\prod_{i=1}^np(x_i|q_i)$.

Let non-negative rates $\hat{R}_1,\ldots, \hat{R}_K$. For relay $k$, $k \in \mc K$, generate a codebook $\mathcal{C}^r_{k}$ consisting of a collection of $2^{n\hat{R}_l}$ independent codewords $\{u_k^n(i_k)\}$ indexed with $i_k\in [1\!:\!2^{n\hat{R}_k}]$, where codeword $u_k^n(i_k)$ has its elements generated i.i.d. according to $\prod_{i=1}^np(u_i|q_i)$. Randomly and independently assign these codewords into $2^{nC_k}$ bins $\{\mathcal{B}_{j_k}\}$, indexed with $j_l\in[1\!:\!2^{nC_k}]$, and containing $2^{n(\hat{R}_k-C_k)}$ codewords each.

\noindent\textit{Encoding at User $l$:} Let $(m_1,\ldots, m_L)$ be the messages to be sent and $(f_1,\ldots,f_L)$ be the selected codebook indexes. User  $l\in \mc L$, transmits the codeword $x_{l}^n(m_l,f_l,q^n)$ in codebook $\mc C_l(f_l)$.

\noindent\textit{Oblivious processing at Relay $k$:} Relay $k$ finds an index $i_k$ such that $u_k^n(i_k)\in \mathcal{C}_k^r$ is strongly $\epsilon$-jointly typical with $y_k^n$.  Using standard arguments,  this can be accomplished with vanishing probability of error as long as $n$ is large and
\begin{align}
\hat{R}_k\geq I(Y_k;U_k|Q).\label{eq:innerCondition_1}
\end{align}
Let $j_k\in [1\!:\!2^{nC_k}]$ be the index such that $u_{k}(i_k)\in \mathcal{B}_{j_k}$. Relay $k$ then forwards the bin index $j_k$ to the CP through the error-free link.

\noindent\textit{Decoding at CP:} The CP collects all the bin indices $j_{\mathcal{K}} = (j_1,\ldots,j_K)$ from the error-free link and finds the set of indices $\hat{i}_{\mathcal{K}} = (\hat{i}_1,\ldots,\hat{i}_K)$ of the compressed vectors $u^n_{\mathcal{K}}$ and the transmitted messages $\hat{m}_{\mathcal{L}} = (\hat{m}_1,\ldots, \hat{m}_L)$, such that
\begin{align}
&(q^n,x_1^n(\hat{m}_1,f_1,q^n),\ldots,x_{L}^n(\hat{m}_L,f_L,q^n),u_{1}^n(\hat{i}_1),\ldots, u_{K}^n(\hat{i}_K))\; \text{strongly } \epsilon-\text{jointly typical},\\
&u^n_{k}(\hat{i}_k)\in \mathcal{B}_{j_k} \quad \text{for } k\in \mc K,\\
&x_l^n(\hat{m}_l,f_l,q^n)\in \mc C_{l}(f_l)\quad \text{for } l\in \mc L.
\end{align}

An error event in the decoding is declared if $\hat{m}_{\mathcal{L}}\neq m_{\mathcal{L}}$ or if there is more than one such $\hat{m}_{\mc L}$. The decoding event can be accomplished with vanishing probability of error for sufficiently long $n$ as shown next.
Assume that for some $\mathcal{T}\subseteq \mathcal{L}$ and $\mathcal{S}\subseteq \mathcal{K}$, we have $\hat{m}_{\mathcal{T}}\neq m_{\mathcal{T}}$ and $\hat{i}_{\mathcal{S}}\neq i_{\mathcal{S}}$, and  $\hat{m}_{\mathcal{T}^c}= m_{\mathcal{T}^c}$ and $\hat{i}_{\mathcal{S}^c}= i_{\mathcal{S}^c}$.
Thus, the tuple $(q^n,x^n_{\mathcal{T}}(\hat{m}_{\mathcal{T}},f_{\mathcal{T}},q^n),x^n_{\mathcal{T}^c}(\hat{m}_{\mathcal{T}^c},f_{\mathcal{T}^c},q^n),u^n_{\mathcal{S}}(i_{\mathcal{S}}),u^n_{\mathcal{S}}(i_{\mathcal{S}}^c) )$ belongs, with high probability, to a typical set with distribution
\begin{align}
\prod_{i=1}^n \left(P_Q(q_i)P_{U_{\mathcal{S}^c}, X_{\mathcal{T}^c} }(u_{\mathcal{S}^c,i}, x_{\mathcal{T}^c,i}|q_i)\prod_{s\in \mathcal{S}} P_{U_{s} }(u_{s,i}|q_i)\prod_{t\in \mathcal{T}} P_{ X_{t} }( x_{t,i}|q_i)\right).
\end{align}

The probability that the tuple $(q^n,x^n_{\mathcal{T}}(\hat{m}_{\mathcal{T}},f_{\mc T},q^n),x^n_{\mathcal{T}^c}(\hat{m}_{\mathcal{T}^c},f_{\mc T^c},q^n),u^n_{\mathcal{S}}(i_{\mathcal{S}}),u^n_{\mathcal{S}}(i_{\mathcal{S}}^c) )$ is strongly $\epsilon$-jointly typical is, according to \cite[Lemma 3]{Sanderovich2008:IT:ComViaDesc},  upper bounded by
\begin{align}
2^{-n[H(U_{\mathcal{S}^c}, X_{\mathcal{T}^c}|Q) -H( U_{\mathcal{K}},X_{\mathcal{L}}|Q) +\sum_{s\in {\mathcal{S}}}H(U_{s}|Q) +\sum_{t\in \mathcal{T}} H(X_t|Q)]}.
\end{align}

Overall, there are
%\begin{align}
$2^{n(\sum_{j\in\mathcal{T}} R_{j} + \sum_{s\in {\mathcal{S}}}[\hat{R}_s-C_s]  )}-1$, of
%\end{align}
such sequences
%$(x^n_{\mathcal{T}}(\hat{m}_{\mathcal{T}},f_{\mc T},q^n),x^n_{\mathcal{T}^c}(\hat{m}_{\mathcal{T}^c},f_{\mc T^c},q^n),u^n_{\mathcal{S}}(i_{\mathcal{S}}),u^n_{\mathcal{S}}(i_{\mathcal{S}}^c) )$
in the set $\mathcal{B}_{j_1}\times \cdots \times \mathcal{B}_{j_K}$.
This means that the CP is able to reliably decode $m_{\mathcal{L}}$ and $i_{\mathcal{K}}$, i.e., that the decoding event has vanishing probability of error for sufficiently long $n$, as long as $(R_1,\ldots, R_L)$ satisfy, for all $\mathcal{T}\subseteq \mathcal{L}$ and for all $\mathcal{S}\subseteq \mathcal{K}$,
\begin{align}
\sum_{t\in \mathcal{T}}R_{t}
&\leq \sum_{s\in \mathcal{S}}[C_s -\hat{R}_s]\!+\!H(U_{\mathcal{S}^c}, X_{\mathcal{T}^c}|Q) -H( U_{\mathcal{K}},X_{\mathcal{L}} |Q) \!+\!\sum_{s\in {\mathcal{S}}}H(U_{s}|Q) +\sum_{t\in \mathcal{T}} H(X_t|Q)\\
&\leq \sum_{s\in \mathcal{S}}[C_s +H(U_s|Y_s,Q)]+H(U_{\mathcal{S}^c}, X_{\mathcal{T}^c}|Q) -H( U_{\mathcal{K}},X_{\mathcal{L}}|Q )  +\sum_{t\in \mathcal{T}} H(X_t|Q)\label{eq:inner_1}\\
%&= \sum_{s\in \mathcal{S}}[C_s +H(U_s|Y_s,Q)]+H(U_{\mathcal{S}^c}, X_{\mathcal{T}^c}|Q) -H( U_{\mathcal{K}},X_{\mathcal{L}}|Q )  %+H(X_{\mathcal{T}}|Q)\label{eq:inner_2}\\
%&= \sum_{s\in \mathcal{S}}[C_s +H(U_s|Y_s,Q)]+H(U_{\mathcal{S}^c}, X_{\mathcal{T}^c}|Q) -H( U_{\mathcal{K}},X_{\mathcal{T}^c}|X_{\mathcal{T}},Q ) \\
%&= \sum_{s\in \mathcal{S}}[C_s +H(U_s|Y_s,Q)]+H( X_{\mathcal{T}^c}|Q)+H(U_{\mathcal{S}^c}| X_{\mathcal{T}^c},Q) - H(X_{\mathcal{T}^c}|X_{\mathcal{T}},Q ) - H( U_{\mathcal{K}}|X_{\mathcal{L}},Q ) \\
&= \sum_{s\in \mathcal{S}}[C_s +H(U_s|Y_s,Q)]+H(U_{\mathcal{S}^c}| X_{\mathcal{T}^c},Q) - H( U_{\mathcal{K}}|X_{\mathcal{L}},Q ) \label{eq:inner_3}\\
&= \sum_{s\in \mathcal{S}}[C_s +H(U_s|Y_s,Q)]+I(U_{\mathcal{S}^c};X_{\mathcal{T}}| X_{\mathcal{T}^c},Q) - H( U_{\mathcal{S}}|X_{\mathcal{L}},U_{\mathcal{S}^c},Q ) \\
&= \sum_{s\in \mathcal{S}}C_s +H(U_{\mathcal{S}}|Y_{\mathcal{S}},X_{\mathcal{L}},U_{\mathcal{S}^c},Q)+I(U_{\mathcal{S}^c};X_{\mathcal{T}}| X_{\mathcal{T}^c},Q) - H( U_{\mathcal{S}}|X_{\mathcal{L}},U_{\mathcal{S}^c},Q ) \label{eq:inner_4}\\
&= \sum_{s\in \mathcal{S}}C_s -I(U_{\mathcal{S}};Y_{\mathcal{S}}|X_{\mathcal{L}},U_{\mathcal{S}^c},Q)+I(U_{\mathcal{S}^c};X_{\mathcal{T}}| X_{\mathcal{T}^c},Q),%\label{eq:inner_4}
\end{align}
where \eqref{eq:inner_1} follows from \eqref{eq:innerCondition_1} and due to the independence of $X_t$ with $X_l$, $l\neq t$; \eqref{eq:inner_3} is due to the independence of $X_{\mathcal{T}^c}$ and $X_{\mathcal{T}}$;  and \eqref{eq:inner_4} follows due to the Markov chains (given $Q$)
$U_{k}\mkv Y_k -(X_{\mathcal{L}},U_{\mathcal{K}/k})$, $k\in \mc K$.
This completes the proof of Theorem~\ref{th:NNC_all_MK_inner}.  \qed

\section{Proof of the Outer Bound in Theorem \ref{th:NNC_all_MK_outer}}\label{app:NNC_all_MK_outer}
The proof of this theorem is along the lines of that of Theorem~\ref{th:MK_C_Main}. In the following, we outline the similar steps and highlight the differences.
Suppose the tuple $(R_{1},\ldots,R_L)$ is achievable. Let $\mathcal{T}$ be a set of $\mathcal{L}$, $\mathcal{S}$ be a non-empty set of $\mathcal{K}$,  and $J_k:= \phi^r_k(Y_k^n,q^n)$ be the message sent by relay $k\in \mathcal{K}$, and let $\tilde{Q}:=Q^n$ be the time-sharing variable. Define for $k\in \mc K$ and $i\in [1\!:\!n]$,
\begin{align}
U_{i,k}:= (J_k,Y_{\mathcal{K}}^{i-1}) \quad \text{and }\quad \bar{Q}_i:= (X_{\mathcal{L}}^{i-1},X_{\mathcal{L},i+1}^n,\tilde{Q}).
\end{align}
%Note that the definition of $U_{i,k}$ differs from that in \eqref{eq:Auxiliary}.

From Fano's inequality, we have with $\epsilon_n\rightarrow 0$ for $n\rightarrow \infty$, for all $\mc T\subseteq \mc L$,
\begin{align}
H(m_{\mathcal{T}}|J_{\mathcal{K}},F_{\mathcal{L}},\tilde{Q}) \leq H(m_{\mathcal{L}}|J_{\mathcal{K}},F_{\mathcal{L}},\tilde{Q})\leq n \epsilon_n.\label{eq:NOMK_C_Fano}
\end{align}

Similarly to \eqref{eq:MK_C_BoundGamma}, we have the following inequality
\begin{align}\label{eq:Inequ2}
H(X^n_{\mathcal{T}}|X^n_{\mathcal{T}^c},J_{\mathcal{K}},\tilde{Q})&\leq \sum_{i=1}^nH(X_{\mathcal{T},i}|X_{\mathcal{T}^c,i},\bar{Q}_i) -nR_{\mathcal{ T}}  := n\Gamma_{\mathcal{T}}.
\end{align}

Then, we have
\begin{align}
R_{\mathcal{T}} =&H(m_{\mathcal{T}})\\
\leq &  H(X_{\mathcal{T}}^n|X_{\mathcal{T}^c}^n,\tilde{Q})-H(X_{\mathcal{T}}^n|X_{\mathcal{T}^c}^n, J_{\mathcal{K}},\tilde{Q})+ n\epsilon_n\label{eq:1}\\
=& \sum_{i=1}^n H(X_{\mathcal{T},i}|X_{\mathcal{T}^c}^n,\tilde{Q}, X_{\mathcal{T}}^{i-1})
%\nonumber\\
%&
-H(X_{\mathcal{T},i}|X_{\mathcal{T}^c}^n, J_{\mathcal{K}},X_{\mathcal{T}}^{i-1},\tilde{Q})+ n\epsilon_n\\
=& \sum_{i=1}^n H(X_{\mathcal{T},i}|X_{\mathcal{T}^c}^n, \tilde{Q}, X_{\mathcal{T}}^{i-1}, X_{\mathcal{T},i+1}^{n})
-H(X_{\mathcal{T},i}|X_{\mathcal{T}^c}^n, J_{\mathcal{K}},X_{\mathcal{T}}^{i-1},\tilde{Q})+ n\epsilon_n\label{eq:NOMK_C_ineq_2}\\
\leq& \sum_{i=1}^n H(X_{\mathcal{T},i}|X_{\mathcal{T}^c,i},\bar{Q}_i)\label{eq:NOMK_C_ineq_3}
-H(X_{\mathcal{T},i}|X_{\mathcal{T}^c,i},J_{\mathcal{K}}, Y_{\mathcal{K}}^{i-1},X_{\mathcal{L}}^{i-1},X_{\mathcal{L},i+1}^{n},\tilde{Q})+ n\epsilon_n\\
=& \sum_{i=1}^n H(X_{\mathcal{T},i}|X_{\mathcal{T}^c,i},\bar{Q}_i)-H(X_{\mathcal{T},i}|X_{\mathcal{T}^c,i}, U_{\mathcal{K},i},\bar{Q}_{i})+ n\epsilon_n\\
=& \sum_{i=1}^n I(X_{\mathcal{T},i};U_{\mathcal{K},i}|X_{\mathcal{T}^c,i},\bar{Q}_{i})+ n\epsilon_n\label{eq:NOMK_C_ineq_4},
\end{align}
where \eqref{eq:1} follows as in \eqref{eq:MK_C_Ineq_00}-\eqref{eq:MK_C_Ineq_12}; \eqref{eq:NOMK_C_ineq_2} follows due to Lemma \ref{lem:IIDinput}  and \eqref{eq:NOMK_C_ineq_3} follows since conditioning reduces entropy.

On the other hand, we have the following inequality
\begin{align}
I(Y_{\mathcal{K}}^n;J_{\mathcal{S}}|X_{\mathcal{L}}^n,J_{\mathcal{S}^c},\tilde{Q})
&=\sum_{i=1}^n I(Y_{\mathcal{K},i};J_{\mathcal{S}}|X_{\mathcal{L}}^n,J_{\mathcal{S}^c},\tilde{Q},Y_{\mathcal{K}}^{i-1})\\
&=\sum_{i=1}^n I(Y_{\mathcal{K},i};J_{\mathcal{S}},Y_{\mathcal{K}}^{i-1}|X_{\mathcal{L}}^n,J_{\mathcal{S}^c},\tilde{Q},Y_{\mathcal{K}}^{i-1})\\
& = \sum_{i=1}^n I(Y_{\mathcal{K},i};U_{\mathcal{S},i}|X_{\mathcal{L},i},U_{\mathcal{S}^c,i},\bar{Q}_i)\\
& \geq \sum_{i=1}^n I(Y_{\mathcal{S},i};U_{\mathcal{S},i}|X_{\mathcal{L},i},U_{\mathcal{S}^c,i},\bar{Q}_i)\label{eq:NOMK_C_Ineq_2_41}.
\end{align}

Then, from the relay nodes side we have,
\begin{align}
C_{\mc S}\geq& \sum_{k\in \mathcal{S}}H(J_k)\label{eq:NOMK_C_Ineq_2_0001}\geq  H(J_{\mathcal{S}})\\
\geq & H(J_{\mathcal{S}}|X_{\mathcal{T}^c}^n,J_{\mathcal{S}^c},\tilde{Q})\\
\geq& I(Y_{\mathcal{K}}^n;J_{\mathcal{S}}|X_{\mathcal{T}^c}^n,J_{\mathcal{S}^c},\tilde{Q})\\
=& I(X_{\mathcal{T}}^n,Y_{\mathcal{K}}^n;J_{\mathcal{S}}|X_{\mathcal{T}^c}^n,J_{\mathcal{S}^c},\tilde{Q})\label{eq:NOMK_C_Ineq_2_1}\\
%=& I(X_{\mathcal{T}}^n;J_{\mathcal{S}}|X_{\mathcal{T}^c}^n,J_{\mathcal{S}^c},\tilde{Q})+I(Y_{\mathcal{K}}^n;J_{\mathcal{S}}|X_{\mathcal{L}}^n, J_{\mathcal{S}^c},\tilde{Q})\\
%
=& H(X_{\mathcal{T}}^n|X_{\mathcal{T}^c}^n,J_{\mathcal{S}^c},\tilde{Q})-H(X_{\mathcal{T}}^n|X_{\mathcal{T}^c}^n, J_{\mathcal{K}},\tilde{Q})
+ I(Y_{\mathcal{K}}^n;J_{\mathcal{S}}|X_{\mathcal{L}}^n, J_{\mathcal{S}^c},\tilde{Q})\\
\geq& H(X_{\mathcal{T}}^n|X_{\mathcal{T}^c}^n,J_{\mathcal{S}^c},\tilde{Q})-n\Gamma_{\mathcal{T}}
+ I(Y_{\mathcal{K}}^n;J_{\mathcal{S}}|X_{\mathcal{L}}^n, J_{\mathcal{S}^c},\tilde{Q})\label{eq:NOMK_C_Ineq_2_3}\\
%
%=& \sum_{i=1}^nH(X_{\mathcal{T},i}|X_{\mathcal{T}^c}^n,J_{\mathcal{S}^c},X_{\mathcal{T}}^{i-1},\tilde{Q})-n\Gamma_{\mathcal{T}}
%+I(Y_{\mathcal{K}}^n;J_{\mathcal{S}}|X_{\mathcal{L}}^n, J_{\mathcal{S}^c},\tilde{Q})\\
%%
\geq & \sum_{i=1}^nH(X_{\mathcal{T},i}|X_{\mathcal{T}^c,i},U_{\mathcal{S}^c,i},\bar{Q}_i)-n\Gamma_{\mathcal{T}}
+ I(Y_{\mathcal{K}}^n;J_{\mathcal{S}}|X_{\mathcal{L}}^n, J_{\mathcal{S}^c},\tilde{Q})\label{eq:NOMK_C_Ineq_2_5}\\
\geq&\sum_{i=1}^nH(X_{\mathcal{T},i}| X_{\mathcal{T}^c,i},U_{\mathcal{S}^c,i},\bar{Q}_i)- H(X_{\mathcal{T},i}|X_{\mathcal{T}^c,i},\bar{Q}_i)\\
& +nR_{\mc T}  + \sum_{i=1}^n I(Y_{\mathcal{S},i};U_{\mathcal{S},i}|X_{\mathcal{L},i},U_{\mathcal{S}^c,i},\bar{Q}_i)\label{eq:NOMK_C_Ineq_3_0}\\
=& nR_{\mc T} +\sum_{i=1}^n I(Y_{\mathcal{S},i};U_{\mathcal{S},i}|X_{\mathcal{L},i},U_{\mathcal{S}^c,i},\bar{Q}_i)
 - \sum_{i=1}^nI(X_{\mathcal{T},i};U_{\mathcal{S}^c,i}|X_{\mathcal{T}^c,i},\bar{Q}_i)\label{eq:NOMK_C_Ineq_3_1}
\end{align}
where:~\eqref{eq:NOMK_C_Ineq_2_1} follows since $J_{\mathcal{S}}$ is a function of $Y_{\mathcal{S}}^n$;  \eqref{eq:NOMK_C_Ineq_2_3} follows from \eqref{eq:Inequ2};  \eqref{eq:NOMK_C_Ineq_2_5} follows since conditioning reduces entropy; and \eqref{eq:NOMK_C_Ineq_3_0} follows from \eqref{eq:Inequ2} and \eqref{eq:NOMK_C_Ineq_2_41}.

We define the standard time-sharing variable $Q'$  uniformly distributed over $\{1,\ldots, n\}$, $X_{\mathcal{L}} := X_{\mathcal{L},Q'}$, $Y_k:= Y_{k,Q'}$, $U_k := U_{k,Q'}$ and $Q := [\bar{Q}_{Q'},Q']$ and
we have from \eqref{eq:NOMK_C_ineq_4} and \eqref{eq:NOMK_C_Ineq_3_1},
\begin{align}
n\sum_{t\in \mathcal{T}} R_t&\leq  nI(X_{\mathcal{T}};U_{\mathcal{K}}|X_{\mathcal{T}^c},Q)+ n\epsilon_n\\
n\sum_{t\in \mathcal{T}} R_t\leq
& \sum_{k\in \mathcal{S}}C_k  - I(Y_{\mathcal{S}};U_{\mathcal{S}}|X_{\mathcal{L}},U_{\mathcal{S}^c},Q) + I(X_{\mathcal{L}};U_{\mathcal{S}^c}|X_{\mathcal{T}^c},Q).
\end{align}

%\begin{align}
%n\sum_{t\in \mathcal{T}} R_t&\leq \sum_{i=1}^n I(X_{\mathcal{T},i};U_{\mathcal{K},i}|X_{\mathcal{T},i},\bar{Q}_{i})+ n\epsilon_n\\
%&= nI(X_{\mathcal{T},Q'};U_{\mathcal{K},Q'}|X_{\mathcal{T}^c,Q'},\bar{Q}_{Q'},Q')+ n\epsilon_n\\
%&= nI(X_{\mathcal{T}};U_{\mathcal{K}}|X_{\mathcal{T}^c},Q)+ n\epsilon_n
%\end{align}
%and similarly, from \eqref{eq:NOMK_C_Ineq_3_1}, we have
%\begin{align}
%n\sum_{t\in \mathcal{T}} R_t\leq
%& \sum_{k\in \mathcal{S}}C_k  - I(Y_{\mathcal{S}};U_{\mathcal{S}}|X_{\mathcal{L}},U_{\mathcal{S}^c},Q) + I(X_{\mathcal{L}};U_{\mathcal{S}^c}|X_{\mathcal{T}^c},Q).
%\end{align}

Define  $W_{Q'}:= (Y_{\mathcal{K}}^{Q'-1},Y_{\mathcal{K}, Q'+1}^n)$, and note that, due to Lemma~\ref{lem:IIDinput}, $X_{\mathcal{L},Q'}$ and $Y_{\mathcal{K},Q'}$ are independent of $W:= W_{Q'}$ when not conditioned on $F_{\mathcal{L}}$. Note that in general,  $\bar{Q}_{Q'}$ is not independent of $X_{\mathcal{L},Q'},Y_{\mathcal{K},Q'}$. Then, conditioned on $Q$, the auxiliary variables $U_{k,Q'}$ satisfies
\begin{align}
U_{k,Q'} &= (J_k,Y_{\mathcal{K}}^{Q'-1})
%& = f'_k(W_{Q'},Y_{k,Q'})\\
 = f_k(W,Y_k,Q).
\end{align}
%where $f'_k(W_{Q'},Y_{k,Q'}):=(J_k,Y_{\mathcal{K}}^{Q'-1})$.
Therefore, conditioned on $\bar{Q}_i$,  for $k\in \mc K$ the following Markov chains hold
\begin{align}
U_{k}\mkv Y_{k}\mkv (X_{\mathcal{L}},Y_{{\mathcal{K}\setminus k}}),\\
U_{k}\mkv (Y_{k},W)\mkv (X_{\mathcal{L}},Y_{{\mathcal{K}\setminus k}},U_{{\mathcal{K}\setminus k}}).
\end{align}

This completes the proof of Theorem~\ref{th:NNC_all_MK_outer}. \qed

\section{ Proof of Theorem~\ref{th:SWZSumRate}}\label{app:SWZSumRate}

Since $R_{\text{sum, CF-SSD}}\leq R_{\text{sum, CF-SD}}\leq R_{\text{sum, CF-JD}}$, to prove that CF-SD and CF-SSD achieve the same sum-rate as CF-JD, it suffices to show $\mc R_{\text{sum, CF-SSD}}\geq R_{\text{sum, CF-JD}}$.
To that end, let us define the following regions, representing the sum-rate achievable by CF-JD and CF-SSD.
\begin{definition}
Let $\mc R_{\text{sum, CF-JD}}$ be the union of tuples $(R,C_1,\ldots,C_K)$ that satisfy, for all $ \mathcal{S} \subseteq \mathcal{K}$,
\begin{align}
R\leq& \sum_{s\in \mathcal{S}} C_s-I(Y_{S};U_{\mathcal{S}}|X_{\mathcal{L}},U_{\mathcal{S}^c},Q)
+ I(X_{\mathcal{L}};U_{\mathcal{S}^c}|Q),
\end{align}
for some joint measure of the form
$p(q)\prod_{l=1}^{L}p(x_l|q) p(y_{\mathcal{K}}|x_{\mathcal{L}})\prod_{k=1}^{K}p(u_k|y_k,q)$.
\end{definition}

\begin{definition}\label{def:SSD_sum}
The region $\mc R_{\text{sum, CF-SSD}}$  is defined as the union of the regions $\mc R_{\text{sum, CF-SSD}}(\pi_r)$ over all possible permutations $\pi_r$, i.e., $\mc R_{\text{CF-SSD}} = \bigcup_{\pi_r} \mc R_{\text{CF-SSD}}(\pi_r)$, where we let $\mc R_{\text{sum, CF-SSD}}(\pi_r)$ with decoding order $(\pi_r)$ be the union of tuples $(R,C_1,\ldots,C_K)$ that satisfy, for all $ \mathcal{S} \subseteq \mathcal{K}$,
\begin{subequations}
\begin{align}
R &\leq I(X_{\mc{L}};U_{\mc K}|Q)\\
C_{\pi_r(k)} &\geq I(U_{\pi_r(k)};Y_{\pi_(k)}|U_{\pi_r(1)},\hdots,U_{\pi_r(k-1)},Q),
\end{align}
\label{eq:SSD_SUM}
\end{subequations}
for some pmf $p(q)\prod_{l=1}^L p(x_l|q)p(y_{\mathcal{K}}|x_{\mathcal{L}})\prod_{k=1}^{K}p(u_k|y_k,q)$.

\end{definition}

We prove  $\mc R_{\text{sum, CF-SSD}}\supseteq \mc R_{\text{sum, CF-JD}}$ using the properties of submodular optimization. To this end,  assume $(R_{\mathrm{sum}},C_1,\ldots, C_K)\in \mc R_{\text{sum, CF-JD}}$ for a joint pmf $p(q)\prod_{l=1}^L p(x_l|q)\prod_{k=1}^{K}p(u_k|y_k,q)$. For such pmf, let $\mc{P}_{R}\in \mathds{R}^{K}_+$  be the polytope formed by the set of pairs $(C_1,\ldots, C_K)$ that satisfy, for all $\mathcal{S} \subseteq \mathcal{K}$, \begin{align}
\sum_{s\in \mathcal{S}}C_s\geq \left[R_{\mathrm{sum}} +I(U_{\mathcal{S}};Y_{\mathcal{S}}|X_{\mathcal{L}},U_{\mathcal{S}^c},Q)-I(U_{\mathcal{S}^c};X_{\mathcal{L}}|Q)\right]^+.
\end{align}

\begin{definition}
For a  pmf $p(q)\prod_{l=1}^L p(x_l|q)\prod_{k=1}^{K}p(u_k|y_k,q)$ we say a point $(R_{\mathrm{sum}},C_1,\ldots, C_K)\in \mc R_{\text{sum, CF-JD}}$ is dominated by a point in $\mc R_{\text{sum, CF-SSD}}$ if there exists  $(R_{\mathrm{sum}}',C_1',\ldots, C_K')\in \mc R_{\text{sum, CF-SSD}}$ for which  $C_k'\leq C_k$, for $k\in \mc K$, and $R'_{\mathrm{sum}}\geq R_{\mathrm{sum}}$.
\end{definition}

To show $(R_{\mathrm{sum}},C_1,\ldots, C_K)\in \mc R_{\text{sum, CF-SSD}}$, it suffices to show that each extreme point of $\mathcal{P}_{R}$ is dominated by a point in $\mc R_{\text{sum, CF-SSD}}$ that achieves a sum-rate $\bar{R}_{\mathrm{sum}}$ satisfying $\bar{R}_{\mathrm{sum}}\geq R_{\mathrm{sum}}$.

Next, we characterize the extreme points of $\mathcal{P}_{R}$.  Let us define the set function $g:2^{\mathcal{K}}\rightarrow \mathds{R}$:
\begin{align}
g(\mathcal{S})&:= R_{\mathrm{sum}} +I(U_{\mathcal{S}};Y_{\mathcal{S}}|U_{\mathcal{S}^c},Q)-I(U_{\mathcal{K}};X_{\mathcal{L}}|Q), \quad \text{for each }\quad \mathcal{S} \subseteq \mathcal{K}.\label{eq:gFunction}
\end{align}

It can be verified that the function $g^+(\mathcal{S}):= \max\{g(\mathcal{S}),0\}$ is a supermodular function (see \cite[Appendix C, Proof of Lemma 6]{Courtade2014LogLoss}\footnote{The proof in \cite[Appendix C, Proof of Lemma 6]{Courtade2014LogLoss} showing that $g'(\mathcal{S}):= I(U_{\mathcal{S}};Y_{\mathcal{S}}|U_{\mathcal{S}^c},Q)$ is supermodular for a model satisfying $Y_{k}\mkv X_{\mathcal{L}}\mkv Y_{\mathcal{K}/k}$, also applies in our setup in which $Y_{k}\mkv X_{\mathcal{L}}\mkv Y_{\mathcal{K}/k}$ does not hold in general.}).

We can rewrite \eqref{eq:gFunction} as follows. For each $\mathcal{S} \subseteq \mathcal{K}$, we have
\begin{align}
g(\mathcal{S}) &= R_{\mathrm{sum}} +I(U_{\mathcal{S}};Y_{\mathcal{S}}|U_{\mathcal{S}^c},Q)-I(U_{\mathcal{K}};X_{\mathcal{L}}|Q)\label{eq:MKChainXY_0}\\
&= R_{\mathrm{sum}} +I(U_{\mathcal{S}};X_{\mathcal{L}},Y_{\mathcal{S}}|U_{\mathcal{S}^c},Q)-I(U_{\mathcal{S}^c};X_{\mathcal{L}}|Q)-I(U_{\mathcal{S}};X_{\mathcal{L}}|U_{S^c},Q)\label{eq:MKChainXY}\\
&= R_{\mathrm{sum}} +I(U_{\mathcal{S}};Y_{\mathcal{S}}|X_{\mathcal{L}}, U_{\mathcal{S}^c},Q)-I(U_{\mathcal{S}^c};X_{\mathcal{L}}|Q),%\\
%&=R_{\mathrm{sum}} + \sum_{s\in \mathcal{S}}I(U_s;Y_s|X_{\mathcal{L}},U_{\mathcal{S}^c},Q)-I(U_{\mathcal{S}^c};X_{\mathcal{L}}|Q),
\end{align}
where \eqref{eq:MKChainXY} follows due to the Markov chain
$U_{\mathcal{S}}\mkv Y_{\mathcal{S}} \mkv (X_{\mathcal{L}},U_{\mathcal{S}^c})$.

Then, by construction, $\mathcal{P}_R$ is equal to the set of $(C_1,\ldots,C_K)$ satisfying for all $\mathcal{S} \subseteq \mathcal{K}$,
\begin{align}
\sum_{s\in \mathcal{S}}C_s\geq g^{+}(\mathcal{S}).
\end{align}

Following the results in submodular optimization \cite[Appendix B, Proposition 6]{DBLP:journals/corr/ZhouX0C16}, we have that for a linear ordering $i_1\prec i_2\prec \cdots \prec i_K$ on the set $\mathcal{K}$, an extreme point of $\mathcal{P}_R$ can be computed as follows for $k = 1,\dots, K$:
\begin{align}
\tilde{C}_{i_k} = g^+(\{i_1,\ldots, i_k\})-g^+(\{i_1,\ldots, i_{k-1}\}).
\end{align}
All the $K!$ extreme points of $\mathcal{P}_R$ can be enumerated by looking over all linear orderings $i_1\prec i_2\prec \cdots \prec i_K$ of $\mathcal{K}$. Each ordering of $\mathcal{K}$ is analyzed in the same manner and, therefore, for notational simplicity, the only ordering we consider is the natural ordering $i_k=k$. By construction,
\begin{align}\label{eq:ExtremePoints}
\tilde{C}_k =& \left[ R_{\mathrm{sum}} + I(U_1^k;Y_1^k|X_{\mathcal{L}},U_{k+1}^{K},Q)-I(U_{k+1}^{K};X_{\mathcal{L}}|Q) \right]^+\\
&-\left[ R_{\mathrm{sum}} + I(U_1^{k-1};Y_1^{k-1}|X_{\mathcal{L}},U_{k}^{K},Q)-I(U_{k}^K;X_{\mathcal{L}}|Q) \right]^+.\nonumber
\end{align}

Let $j$ be the first index for which $\tilde{C}_j >0$, i.e., the first $k$ for which $g(\{1,\ldots, j\})>0$. Then, it follows from \eqref{eq:ExtremePoints} that
\begin{align}
\tilde{C}_k  =& I(U_1^k;Y_1^k|X_{\mathcal{L}},U_{k+1}^{K},Q)-I(U_{k+1}^{K};X_{\mathcal{L}}|Q)\\
&- I(U_1^{k-1};Y_1^{k-1}|X_{\mathcal{L}},U_{k}^{K},Q)+I(U_{k}^K;X_{\mathcal{L}}|Q)\\
=&I(U_1^k;Y_1^k|U_{k+1}^K,Q)-I(U_1^{k-1};Y_1^{k-1}|U_{k}^K,Q)\label{eq:eq1_MK_JD_SD_0}\\
=& I(Y_k;U_k|U_{k+1}^K,Q), \qquad  \text{for all } k>j, \label{eq:eq1_MK_JD_SD}
\end{align}
where \eqref{eq:eq1_MK_JD_SD_0} follows from  \ref{eq:MKChainXY_0}; and \eqref{eq:eq1_MK_JD_SD} follows due to the Markov Chain
\begin{align}
U_k \mkv Y_k \mkv (X_{\mathcal{L}},Y_{\mathcal{K}/k}, U_{\mathcal{K}/k}). \label{eq:MK_alpha}
\end{align}

Moreover, since  we must have $g(\{1,\ldots, j'\})\leq 0$ for $j'< j$, $\tilde{C}_j$ can be expressed as
\begin{align}
\tilde{C}_j&= R_{\mathrm{sum}} + I(U_1^j;Y_1^j|X_{\mathcal{L}},U_{j+1}^K,Q)-I(U_{j+1}^K;X_{\mathcal{L}}|Q) \\
&=  I(Y_j;U_j|U_{j+1}^K,Q) + g(\{1,\ldots, j-1\}),\\
& = (1-\alpha) I(Y_{j};U_{j}|U_{j+1}^{K},Q),
\end{align}
where $\alpha\in (0,1]$ is defined as
\begin{align}
\alpha&:= \frac{-g(\{1,\ldots, j-1\})}{I(Y_{j};U_{j}|U_{j+1}^{K},Q)}
=\frac{I(U_{j}^{K};X_{\mathcal{L}}|Q)-R_{\mathrm{sum}} - I(U_1^{j-1};Y_1^{j-1}|X_{\mathcal{L}},U_{j}^K,Q)}{I(Y_{j};U_{j}|U_{j+1}^{L},Q)}.\label{eq:alpha_def}
\end{align}

Therefore, for the natural ordering, the extreme point $(\tilde{C}_1,\ldots, \tilde{C}_K)$ is given as
\begin{align}
(\tilde{C}_1,\ldots, \tilde{C}_K) = \left(0,\ldots, 0, (1-\alpha)I(Y_j;U_j|U_{j+1}^K,Q),I(Y_{j+1};U_{j+1}|U_{j+2}^K,Q),\right.\\
\left. \ldots,I(Y_{K-1};U_{K-1}|U_{K},Q),  I(Y_K;U_K|Q)\right).
\end{align}

Next, we show that $(\tilde{C}_1,\ldots, \tilde{C}_K)\in \mathcal{P}_{R}$, is dominated by a point $(\bar{R}_{\mathrm{sum}},C_1,\ldots, C_K)\in \mc R_{\text{sum, CF-SDD}}$ that achieves a sum-rate $\bar{R}_{\mathrm{sum}}\geq R_{\mathrm{sum}}$.

We consider an instance of the CF-SSD in which for a fraction $\alpha$ of the time, the CP decodes $U_{j+1}^n,\ldots, U^n_K$ while relays $k=1,\ldots, j$ are inactive.  For the remaining fraction of time $(1-\alpha)$, the CP decodes $U_{j}^n,\ldots, U^n_K$ and relays $k=1,\ldots, j-1$ are inactive. Then, the CP decodes $X_{\mathcal{L}}$.

Formally, we consider the pfm $p(q')\prod_{l=1}^L p(x_l'|q')\prod_{k=1}^{K}p(u_k'|y_k,q')$ for CF-SSD as follows. Let $B$ denote a Bernoulli random variable with parameter $\alpha\in(0,1]$, i.e., $B=1$ with probability $\alpha$ and $B=0$ with probability $(1-\alpha)$. We let $\alpha$ as in \eqref{eq:alpha_def}. We consider the reverse ordering $\pi_r$ such that $\pi_r(1)=K, \pi_r(2) = K-1, \ldots, \pi_r(K) = 1$, i.e., compression is done from relay $K$ to relay $1$. Then, we let $Q'=(B,Q)$ and the tuple of random variables be distributed as
\begin{align}
(Q',X'_{\mathcal{L}}, U_{\mathcal{K}}') =
\begin{cases}
((1,Q),X_{\mathcal{L}},\emptyset,\ldots,\emptyset,U_{j+1},\ldots, U_K)&\text{if }B = 1,\\
((0,Q),X_{\mathcal{L}},\emptyset,\ldots,\emptyset,U_{j},\ldots, U_K)&\text{if }B = 0.
\end{cases}
\end{align}

From Definition~\ref{def:SSD_sum}, we have $(\bar{R}_{\mathrm{sum}},C_1,\ldots, C_K)\in \mc{R}_{\text{sum,CF-SSD}}$,  where
\begin{align}
C_k&= I(Y_k;U_k'|U_{k+1}',\ldots, U_{K}',Q'), \quad\text{for }\quad k=1,\ldots, K,\\
\bar{R}_{\mathrm{sum}}&=I(X'_{\mathcal{L}};U'_{\mathcal{K}}|Q').
\end{align}

Then, for $k=1,\ldots, j-1$, we have
\begin{align}
C_k&= I(Y_k;U_k'|U_{k+1}',\ldots, U_{K}',Q') = 0 = \tilde{C}_{k},\label{eq:Domitated_1}
\end{align}
where \eqref{eq:Domitated_1} follows since $U_k'=\emptyset$ for $k<j$ independently of $B$.
For $k=j+1,\ldots, K$, we have
\begin{align}
C_k&= I(Y_k;U_k'|U_{k+1}',\ldots, U_{K}',Q') \\
&=\alpha I(Y_k;U_k|U_{k+1},\ldots, U_{K},Q,B=1) + (1-\alpha) I(Y_k;U_k|U_{k+1},\ldots, U_{K},Q,B=0) \\
&= I(Y_k;U_k|U_{k+1},\ldots, U_{K},Q) = \tilde{C}_{k},\label{eq:Domitated_2}
\end{align}
where \eqref{eq:Domitated_2} follows since $U_k'=U_k$ for $k>j$ independently of $B$.
For $k=j$, we have
\begin{align}
C_j&= I(Y_j;U_j'|U_{j+1}',\ldots, U_{K}',Q') \\
&=\alpha I(Y_j;U_j|U_{j+1},\ldots, U_{K},Q,B=1) + (1-\alpha) I(Y_j;U_j|U_{j+1},\ldots, U_{K},Q,B=0) \\
&= (1-\alpha)I(Y_j;U_j|U_{j+1},\ldots, U_{K},Q) = \tilde{C}_{j}; \label{eq:Domitated_3}
\end{align}
where \eqref{eq:Domitated_3} follows since $U_j'=\emptyset$ for $B=1$ and $U_j'=U_j$ for $B=0$.

On the other hand, the sum-rate satisfies
\begin{align}
\bar{R}_{\mathrm{sum}}=& I(X_{\mathcal{L}}';U_{\mathcal{K}}'|Q')\\
%=&\alpha I(X_{\mathcal{L}};U_{j+1}^{K}|Q,B=1) +(1-\alpha)I(X_{\mathcal{L}};U_{j}^{K}|Q,B=0)\\
=&I(X_{\mathcal{L}};U_{j}^{K}|Q) -\alpha  I(X_{\mathcal{L}};U_{j}|U_{j+1}^{K},Q)\\
=&I(X_{\mathcal{L}};U_{j}^{K}|Q) -\frac{ I(X_{\mathcal{L}};U_{j}|U_{j+1}^{K},Q)}{I(Y_{j};U_{j}|U_{j+1}^{K},Q)}
\cdot \left[I(U_{j}^K;X_{\mathcal{L}}|Q)-R_{\mathrm{sum}} - I(U_1^{j-1};Y_1^{j-1}|X_{\mathcal{L}},U_j^K,Q)\right]\label{eq:alpha_sub}\\
%\geq &I(X_{\mathcal{L}};U_{j}^{K}|Q) - I(U_{j}^{K};X_{\mathcal{L}}|Q) +R_{\mathrm{sum}} \\
%&+ \sum_{s = 1}^{j}I(U_s;Y_s|X_{\mathcal{L}},U_{j+1},Q)\label{eq:alpha_ineq}\\
\geq &R_{\mathrm{sum}} + I(U_1^{j-1};Y_1^{j-1}|X_{\mathcal{L}},U_{j}^K,Q)\label{eq:alpha_ineq}\\
\geq &R_{\mathrm{sum}},\label{eq:Domitated_4}
\end{align}
where \eqref{eq:alpha_sub} follows from \eqref{eq:alpha_def}; and \eqref{eq:alpha_ineq} follows since $I(Y_{j};U_{j}|U_{j+1}^{L},Q)\geq I(X_{\mathcal{L}};U_{j}|U_{j+1}^{K},Q)$ due to the Markov Chain \eqref{eq:MK_alpha}.

Therefore, from \eqref{eq:Domitated_1}, \eqref{eq:Domitated_2}, \eqref{eq:Domitated_3} and \eqref{eq:Domitated_4}, it follows that the extreme point $(\tilde{C}_1,\ldots,\tilde{C}_K)\in \mathcal{P}_R$ is dominated by the point $(\bar{R}_{\mathrm{sum}},C_1,\dots, C_K)\in \mc R_{\text{sum, CF-SSD}}$ satisfying $\bar{R}_{\mathrm{sum}}\geq R_{\mathrm{sum}}$.
Similarly, considering all possible orderings, each extreme point of $\mathcal{P}_R$ can be shown to be dominated by a point $(R_{\mathrm{sum}},C_1,\ldots, C_K)$ which lies in $\mc R_{\text{sum, CF-SSD}}$ (associated to a permutation $\pi_r$).
This completes the proof of Theorem~\ref{th:SWZSumRate}. \qed

%%%%%%%%%%%%%%%%%%%%%%%%%%%%%%%%%%%%%%%%%%%%%%%%%%%%%%%%%%%%%%%%%%%%%%%%%%%%%%%%%%%%%%%%%%%%%%%%%%%%%%%%%%%%%%%%%%%%%%%%%%%%%%%%%%%%%%%%%%%%%%%%%%%%%%%%%%%%%%%%%%%%%%%%%%
%%%%%%%%%%%%%%%%%%%%%%%%%%%%%%%%%%%%%%%%%%%%%%%%%%%%%%%%%%%%%%%%%%%%%%%%%%%%%%%%%%%%%%%%%%%%%%%%%%%%%%%%%%%%%%%%%%%%%%%%%%%%%%%%%%%%%%%%%%%%%%%%%%%%%%%%%%%%%%%%%%%%%%%%%%
%%%%%%%%%%%%%%%%%%%%%%%%%%%%%%%%%%%%%%%%%%%%%%%%%%%%%%%%%%%%%%%%%%%%%%%%%%%%%%%%%%%%%%%%%%%%%%%%%%%%%%%%%%%%%%%%%%%%%%%%%%%%%%%%%%%%%%%%%%%%%%%%%%%%%%%%%%%%%%%%%%%%%%%%%%

%\onecolumn

\section{Proof of Theorem \ref{th:GaussSumCap}}\label{app:GaussSumCap}

The proof is along the lines of the proofs of \cite[Theorem 8]{Ekrem:TIT:14} and \cite[Theorem 4]{DBLP:journals/corr/ZhouX0C16}, and uses the relations between the MMSE and the Fischer information matrix developed in \cite{Ekrem:TIT:14} and a reparametrization of the MMSE matrix from \cite{DBLP:journals/corr/ZhouX0C16}, but differs from them to account for the time-sharing variable $Q$. We  will use the following lemmas.
\begin{lemma}{\cite{Dembo:IT:91,Ekrem:TIT:14}}\label{lem:FI_Ineq}
Let $(\mathbf{X,Y})$  be a pair of random vectors with pmf $p(\mathbf{x},\mathbf{y})$. We have
\begin{align}
\log|(\pi e) \mathbf{J}^{-1}(\mathbf{X}|\mathbf{U})|\leq h(\mathbf{X}|\mathbf{U})\leq \log|(\pi e) \mathrm{mmse}(\mathbf{X}|\mathbf{U})|.
\end{align}
where the Fischer information matrix of $\mathbf{X}$ conditional on $\mathbf{Y}$ is defined as
\begin{align}
\mathbf{J}(\mathbf{X}|\mathbf{Y}) := \mathrm{E}[\nabla \log p(\mathbf{X}|\mathbf{Y})\nabla\log p(\mathbf{X}|\mathbf{Y})^T],
\end{align}
and the minimum mean squared error (MMSE) matrix  is defined as
\begin{align}
\mathrm{mmse}(\mathbf{X}|\mathbf{Y}) := \mathrm{E}[(\dv X-\mathrm{E}[\dv X|\dv Y])(\dv X-\mathrm{E}[\dv X|\dv Y])^H].
\end{align}
\end{lemma}

\begin{lemma}{\cite{Ekrem:TIT:14}}\label{lemm:Brujin}
Let $\mathbf{V}_1,\mathbf{V}_2$ be an arbitrary random vector with finite second moments, and $\mathbf{N}\sim \mc{CN}(\dv 0,\boldsymbol\Lambda_N)$. Assume $(\mathbf{V}_1,\mathbf{V}_2)$ and $\mathbf{N}$ are independent. We have
\begin{align}
\mathrm{mmse}(\mathbf{V}_2|\mathbf{V}_1,\mathbf{V}_2+\mathbf{N}) = \boldsymbol\Lambda_N -\boldsymbol\Lambda_N\mathbf{J}(\mathbf{V}_2+\mathbf{N}|\mathbf{V}_1)\boldsymbol\Lambda_N.
\end{align}
\end{lemma}

First, we derive an outer bound on the capacity region of the memoryless Gaussian MIMO model described by~\eqref{mimo-gaussian-model} and~\eqref{input-covariance-matrix-Gaussian-model} under time-sharing of Gaussian inputs by deriving an outer bound on the rate region given in Theorem~\ref{th:MK_C_Main} under input constraints \eqref{eq:GaussInput} and \eqref{eq:powConst}. Then, we show that this outer bound is achievable by time-sharing of Gaussian inputs.

For a fixed $Q=q$, let us define $\dv Y_{k,q} :=  \dv H_{k,\mc L}\dv X_{\mc L,q} + \dv N_k$ and $\dv X_{\mc L,q} :=[\dv X_{1}^n,\ldots, \dv X_{L}^n|Q=q]^T$. For fixed Gaussian distribution $\mathbf{X}_{\mathcal{L},q}\sim\mathcal{CN}(\mathbf{0},\mathbf{K}_{\mathcal{L},q})$ and distribution $\prod_{k=1}^{K}p(\mathbf{\hat{y}}_k|\mathbf{y}_k,q)$, let us choose $\mathbf{B}_{k,q}$ satisfying $\mathbf{0}\preceq\mathbf{B}_{k,q}\preceq\mathbf{\Sigma}_{k}^{-1}$ such that for  $k\in \mc K$,
\begin{align}
\mathrm{mmse}(\mathbf{Y}_{k,q}|\mathbf{X}_{\mathcal{L},q},\mathbf{U}_{k,q}) = \mathbf{\Sigma}_{k}-\mathbf{\Sigma}_{k}\mathbf{B}_{k,q}\mathbf{\Sigma}_{k}.\label{eq:covB}
\end{align}

Such $\mathbf{B}_{k,q}$ always exists since $ \mathbf{0}\preceq\mathrm{mmse}(\mathbf{Y}_{k,q}|\mathbf{X}_{\mathcal{L},q},\mathbf{U}_{k,q})\preceq \mathbf{\Sigma}_{k}$ for all $q\in \mathcal{Q}$ and $k\in \mc K$.

Next, we derive the following equality. For $q\in \mathcal{Q}$, and for all $\mathcal{T}\subseteq \mathcal{L}$ and $\mathcal{S}\subseteq \mathcal{K}$, we have
\begin{align}
\mathbf{J}(\mathbf{X}_{\mathcal{T},q}|\mathbf{X}_{\mathcal{T}^c,q},\mathbf{U}_{S^c,q}) = \sum_{k\in\mathcal{S}^{c}}\mathbf{H}_{k,\mathcal{T}}^{H}
\mathbf{B}_{k,q}
\mathbf{H}_{k,\mathcal{T}}+\mathbf{K}^{-1}_{\mathcal{T},q}\label{eq:Fischerequality}.
\end{align}
Equality \eqref{eq:Fischerequality} is obtained as follows. Let us define, for all $\mathcal{T}\subseteq \mathcal{L}$ and $\mathcal{S}\subseteq \mathcal{K}$
\begin{align}
\mathbf{Y}_{\mathcal{S}^c,q} := \mathbf{H}_{{S}^c,\mathcal{T}}\mathbf{X}_{\mathcal{T},q}+\mathbf{H}_{\mathcal{S}^c,\mathcal{T}^c}\mathbf{X}_{\mathcal{T}^c,q}+\mathbf{N}_{\mathcal{S}^c}.
\end{align}
It follows from the MMSE estimation of Gaussian random vectors \cite{elGamal:book}, that
\begin{align}
\mathbf{X}_{\mathcal{T},q}
 &= \mathrm{E}[\mathbf{X}_{\mathcal{T},q}|\mathbf{X}_{\mathcal{T}^c,q},\mathbf{Y}_{\mathcal{S}^c,q}]+\mathbf{Z}_{\mathcal{T},\mathcal{S}^c} = \sum_{k\in \mathcal{S}^c}\mathbf{G}_{\mathcal{T},k}(\mathbf{Y}_{k,q}-\mathbf{H}_{k,\mathcal{T}^c}\mathbf{X}_{\mathcal{T}^c})+\mathbf{Z}_{\mathcal{T},\mathcal{S}^c,q},
\end{align}
where $\mathbf{Z}_{\mathcal{T},\mathcal{S}^c,q}\sim\mathcal{CN}(\mathbf{0},\boldsymbol\Lambda_{Z,q})$ is the estimation error,  with covariance matrix
\begin{align}
\boldsymbol\Lambda_{Z,q} =  \left( \mathbf{K}_{\mathcal{T},q}^{-1} +\sum_{k\in \mathcal{S}^c}\mathbf{H}_{k,\mathcal{T}}^H\mathbf{\Sigma}_{k}^{-1}\mathbf{H}_{k,\mathcal{T}}  \right)^{-1},
\end{align}
and
\begin{align}\label{eq:MMSEfilter}
\mathbf{G}_{\mathcal{T},k,q} = \boldsymbol\Lambda_{Z,q} \mathbf{H}^H_{k,\mathcal{T}}\mathbf{\Sigma}_k^{-1}.
\end{align}

Note that since  $\mathbf{Z}_{\mathcal{T},\mathcal{S}^c,q}$ and $\mathbf{X}_{\mathcal{T}^c,q},\mathbf{Y}_{\mathcal{S}^c,q}$ are Gaussian distributed, $\mathbf{Z}_{\mathcal{T},\mathcal{S}^c,q}$ and  $\mathbf{X}_{\mathcal{T}^c,q},\mathbf{Y}_{\mathcal{S}^c,q}$ are uncorrelated due to the orthogonality principle of the MMSE estimator \cite{elGamal:book}, and hence independent.
Therefore,  $\mathbf{Z}_{\mathcal{T},\mathcal{S}^c,q}$ is also independent of $\mathbf{U}_{\mathcal{S}^c,q}$. Then, by Lemma~\ref{lemm:Brujin}, we have
\begin{align}
\mathbf{J}(\mathbf{X}_{\mathcal{T},q}|\mathbf{X}_{\mathcal{T}^c,q},\mathbf{U}_{S^c,q})
&=\boldsymbol\Lambda_{Z,q}^{-1} - \boldsymbol\Lambda_{Z,q}^{-1} \text{mmse}
\left(
 \sum_{k\in \mathcal{S}^c}\mathbf{G}_{\mathcal{T},k}
  (\mathbf{Y}_k-\mathbf{H}_{k,\mathcal{T}^c}
  \mathbf{X}_{\mathcal{T}^c,q})
  \Big|\mathbf{X}_{\mathcal{L},q}, \mathbf{U}_{\mathcal{S}^c,q}
  \right)  \boldsymbol\Lambda_{Z,q}^{-1} \\
  &=\boldsymbol\Lambda_{Z,q}^{-1} - \boldsymbol\Lambda_{Z,q}^{-1} \text{mmse}
\left(
 \sum_{k\in \mathcal{S}^c}\mathbf{G}_{\mathcal{T},k}
\mathbf{Y}_k
  \Big|\mathbf{X}_{\mathcal{L},q}, \mathbf{U}_{\mathcal{S}^c,q}
  \right)  \boldsymbol\Lambda_{Z,q}^{-1} \\
  &=\boldsymbol\Lambda_{Z,q}^{-1} - \boldsymbol\Lambda_{Z,q}^{-1}
\left(
 \sum_{k\in \mathcal{S}^c}\mathbf{G}_{\mathcal{T},k}
\text{mmse}\left(\mathbf{Y}_k
  \Big|\mathbf{X}_{\mathcal{L},q}, \mathbf{U}_{\mathcal{S}^c,q} \right)\mathbf{G}_{\mathcal{T},k}^H
  \right)  \boldsymbol\Lambda_{Z,q}^{-1} \label{eq:CrossTerms}\\
  &=\boldsymbol\Lambda_{Z,q}^{-1} -
 \sum_{k\in \mathcal{S}^c}\mathbf{H}_{k,\mathcal{T}}^H
\left(\mathbf{\Sigma}_{k}^{-1}-\mathbf{B}_{k} \right)\mathbf{H}_{k,\mathcal{T}}\label{eq:CovSubs}
 \\
 &=\mathbf{K}_{\mathcal{T},q}^{-1} +
 \sum_{k\in \mathcal{S}^c}\mathbf{H}_{k,\mathcal{T}}^H
\mathbf{B}_{k}\mathbf{H}_{k,\mathcal{T}},
 \end{align}
where \eqref{eq:CrossTerms} follows since the cross terms are zero due to the Markov chains,
\begin{align}
(\mathbf{U}_{k,q},\mathbf{Y}_k)\mkv \mathbf{X}_{\mc L, q} \mkv (\mathbf{U}_{\mathcal{K}/k,q},\mathbf{Y}_{\mathcal{K}/k}),
\end{align}
 for $q\in \mc Q$ and $k\in \mc K$; and \eqref{eq:CovSubs} is due to \eqref{eq:covB} and \eqref{eq:MMSEfilter}.

We proceed to derive the outer bound. From  \eqref{eq:MK_C_Main}, we have for $k\in \mc K$ and $q\in \mathcal{Q}$,
\begin{align}
I(\mathbf{Y}_{k,q};\mathbf{U}_{k,q}|\mathbf{X}_{\mathcal{L},q},Q=q)& = \log|(\pi e)\boldsymbol\Sigma_k| -h(\mathbf{Y}_{k,q}|\mathbf{X}_{\mathcal{L},q},\mathbf{U}_{k,q},Q=q)\\
& \geq \log|(\pi e)\boldsymbol\Sigma_k| -\log|(\pi e )~\mathrm{mmse}(\mathbf{Y}_{k,q}|\mathbf{X}_{\mathcal{L},q},\mathbf{U}_{k,q}) |\\
&\geq \log\frac{|\mathbf{\Sigma}_{k}^{-1}|}{|\mathbf{\Sigma}_{k}^{-1}-\mathbf{B}_{k,q}|}.\label{eq:firstIneq}
\end{align}
On the other hand,
\begin{align}
I(\mathbf{X}_{\mathcal{T},q};\mathbf{U}_{S^c,q}|\mathbf{X}_{\mathcal{T}^c,q},Q=q)& = h(\mathbf{X}_{\mathcal{T},q}|Q=q)-h(\mathbf{X}_{\mathcal{T},q}|\mathbf{X}_{\mathcal{T}^c,q},\mathbf{U}_{S^c,q},Q=q)\\
&\leq \log|\mathbf{K}_{\mathcal{T},q}|-\log|\mathbf{J}^{-1}(\mathbf{X}_{\mathcal{T},q}|\mathbf{X}_{\mathcal{T}^c,q},\mathbf{U}_{S^c,q})|\label{eq:FI_Ineq},\\
& \leq \log|\mathbf{K}_{\mathcal{T},q}|+
\log
\left|\sum_{k\in\mathcal{S}^{c}}\mathbf{H}_{k,\mathcal{T}}^{H}
\mathbf{B}_{k,q}
\mathbf{H}_{k,\mathcal{T}}+\mathbf{K}^{-1}_{\mathcal{T},q}\right|\label{eq:secondtIneq},
\end{align}
where \eqref{eq:FI_Ineq} is due to Lemma \ref{lem:FI_Ineq}; and \eqref{eq:secondtIneq} is due to \eqref{eq:Fischerequality}.

Substituting \eqref{eq:firstIneq} and \eqref{eq:secondtIneq} in \eqref{eq:MK_C_Main} for each $\mathcal{T}\subseteq \mathcal{L}$, we have
\begin{align}
I(\mathbf{Y}_k;\mathbf{U}_k|\mathbf{X}_{\mathcal{L}},Q)
&= \sum_{q\in \mathcal{Q}}p(q)I(\mathbf{Y}_{k,q};\mathbf{U}_{k,q}|\mathbf{X}_{\mathcal{L},q},Q=q)\\
&\geq \mathrm{E}_Q\left[\log\frac{|\mathbf{\Sigma}_{k}^{-1}|}{|\mathbf{\Sigma}_{k}^{-1}-\mathbf{B}_{k,q}|}\right]\label{eq:boundQ1},
\end{align}
and
\begin{align}
I(\mathbf{X}_{\mathcal{T}};\mathbf{U}_{S^c}|\mathbf{X}_{\mathcal{T}^c,q},Q) &= \sum_{q\in \mathcal{Q}}p(q)I(\mathbf{X}_{\mathcal{T},q};\mathbf{U}_{S^c,q}|\mathbf{X}_{\mathcal{T}^c,q},Q=q)\\
&\leq \mathrm{E}_Q\left[ \log|\mathbf{K}_{\mathcal{T},q}|+
\log
\left|\sum_{k\in\mathcal{S}^{c}}\mathbf{H}_{k,\mathcal{T}}^{H}
\mathbf{B}_{k,q}
\mathbf{H}_{k,\mathcal{T}}+\mathbf{K}^{-1}_{\mathcal{T},q}\right|\right]\label{eq:boundQ2}.
\end{align}

This gives an outer bound on the capacity region of the memoryless Gaussian MIMO model described by~\eqref{mimo-gaussian-model} and~\eqref{input-covariance-matrix-Gaussian-model} under time-sharing of Gaussian inputs as given in \eqref{eq:GaussSumCap}.

The direct part of Theorem~ \ref{th:GaussSumCap} follows by noting that this outer bound is achieved by evaluating \eqref{eq:MK_C_Main}, for $Q=q$, with
$\mathbf{X}_{l}|Q=q\sim\mathcal{CN}(\mathbf{0},\mathbf{K}_{l,q})$ and $\mathbf{U}_{k,q}\sim\mathcal{CN}(\mathbf{Y}_{k,q},\mathbf{Q}_{k,q})$, where $\mathbf{B}_{k,q} = (\mathbf{\Sigma}_k+\mathbf{Q}_{k,q})^{-1}$ for some $\mathbf{0}\preceq\mathbf{B}_{k,q}\preceq\mathbf{\Sigma}_k^{-1}$ as given in \eqref{eq:covB}. \qed

%%%%%%%%%%%%%%%%%%%%%%%%%%%%%%%%%%%%%%%%%%%%%%%%%%%%%%%%%%%%%%%%%%%%%%%%%%%%%%%%%%%%%%%%%%%%%%%%%%%%%%%%%%%%%%%%%%%%%%%%%%%%%%%%%%%%%%%%%%%%%%%%%%%%%%%%%%%%%%%%%%%%%%%%%%
%%%%%%%%%%%%%%%%%%%%%%%%%%%%%%%%%%%%%%%%%%%%%%%%%%%%%%%%%%%%%%%%%%%%%%%%%%%%%%%%%%%%%%%%%%%%%%%%%%%%%%%%%%%%%%%%%%%%%%%%%%%%%%%%%%%%%%%%%%%%%%%%%%%%%%%%%%%%%%%%%%%%%%%%%%
%%%%%%%%%%%%%%%%%%%%%%%%%%%%%%%%%%%%%%%%%%%%%%%%%%%%%%%%%%%%%%%%%%%%%%%%%%%%%%%%%%%%%%%%%%%%%%%%%%%%%%%%%%%%%%%%%%%%%%%%%%%%%%%%%%%%%%%%%%%%%%%%%%%%%%%%%%%%%%%%%%%%%%%%%%
\section{Proof of Proposition \ref{th:GaussSumCap_EqualK}}\label{app:GaussSumCap_EqualK}

To prove Proposition~\ref{th:GaussSumCap_EqualK} we show that under channel input constraint \eqref{eq:EqualCov}, the capacity region $\mc C_{\text{G}}(C_{\mc K})$ in Theorem \ref{th:GaussSumCap} is outer bounded by the region  $\mc C_{\text{G}}^{\text{no-ts}}(C_{\mc K})$ in \eqref{capacity-Gaussian-example-no-time-sharing}. Then, the result follows since $\mc C_{\text{G}}^{ \text{no-ts}}(C_{\mc K})$ is achievable with CF-SD without time-sharing, i.e.,  $Q=\emptyset$, and Gaussian channel inputs satisfying \eqref{eq:EqualCov}.
We use the following lemma, which can be readily proven by the application of Weyl's inequality \cite{Horn1985}.
% We use the following lemma proved in Appendix \ref{app:AB}.

\begin{lemma}\label{lem:AB}
Let $\mathbf{A}$ and $\mathbf{B}$ be two $m\times m$ positive-definite matrices  satisfying $\mathbf{B}\succeq \mathbf{A}$. Then for any $m\times m$ positive-definite matrix $\mathbf{C}$, we have
$|\mathbf{I}+\mathbf{B}\mathbf{C}|\geq|\mathbf{I}+\mathbf{A}\mathbf{C}|$.
\end{lemma}

In the following we show $\mc C_{\text{G}}(C_{\mc K}) \subseteq \mc C_{\text{G}}^{\text{no-ts}}(C_{\mc K}) $. Let us  define $\bar{\mathbf{B}}_k:= \sum_{q\in \mathcal{Q}}p(q)\mathbf{B}_{k,q}$ for $ \dv 0\preceq \dv B_{k,q}\preceq \dv \Sigma_{k}^{-1}$, $k\in \mc K$, as in Theorem \ref{th:GaussSumCap}. Note that  $ \dv 0\preceq \bar{\mathbf{B}}_k\preceq \dv \Sigma_{k}^{-1}$. We have, from \eqref{eq:GaussSumCap}
\begin{align}
\sum_{q\in \mathcal{Q}}p(q)\log\frac{|\mathbf{\Sigma}_{k}^{-1}|}{|\mathbf{\Sigma}_{k}^{-1}-\mathbf{B}_{k,q}|}
&\geq \log\frac{|\mathbf{\Sigma}_{k}^{-1}|}{|\mathbf{\Sigma}_{k}^{-1}-\sum_{q\in \mathcal{Q}}p(q)\mathbf{B}_{k,q}|}\label{eq:logDetProp}\\
&=\log\frac{|\mathbf{\Sigma}_{k}^{-1}|}{|\mathbf{\Sigma}_{k}^{-1}-\bar{\mathbf{B}}_k|},\label{eq:logDetProp2}
\end{align}
where \eqref{eq:logDetProp} follows from the concavity of the log-det function and Jensen's Inequality \cite{Boyd2004}.

Similarly, from \eqref{eq:GaussSumCap} we have
\begin{align}
\sum_{q\in \mathcal{Q}}&p(q) \left(
\log \frac{
\left|\sum_{k\in\mathcal{S}^{c}}\mathbf{H}_{k,\mathcal{T}}^{H}
\mathbf{B}_{k,q}
\mathbf{H}_{k,\mathcal{T}}+\mathbf{K}^{-1}_{\mathcal{T},q}\right|}{\left|\mathbf{K}_{q,\mathcal{T}}^{-1}\right|}\right)\label{eq:secondtIneq_app}\\
&= \sum_{q\in \mathcal{Q}}p(q) \left(\log\left|\mathbf{\tilde{K}}_{\mathcal{T}}\right|+
\log
\left|\sum_{k\in\mathcal{S}^{c}}\mathbf{H}_{k,\mathcal{T}}^{H}
\mathbf{B}_{k,q}
\mathbf{H}_{k,\mathcal{T}}+\mathbf{\tilde{K}}^{-1}_{\mathcal{T}}\right|\right)\label{eq:secondEqualK}\\
&\leq\log \left|\tilde{\mathbf{K}}_{\mathcal{T}}\sum_{k\in\mathcal{S}^{c}}\mathbf{H}_{k,\mathcal{T}}^{H}
\bar{\mathbf{B}}_{k}
\mathbf{H}_{k,\mathcal{T}}+\mathbf{I}\right|\label{eq:secondtIneq_Jensen}\\
&\leq\log \left|\mathbf{K}_{\mathcal{T}}\sum_{k\in\mathcal{S}^{c}}\mathbf{H}_{k,\mathcal{T}}^{H}
\bar{\mathbf{B}}_{k}
\mathbf{H}_{k,\mathcal{T}}+\mathbf{I}\right|\label{eq:secondtIneq_4},
\end{align}
where \eqref{eq:secondEqualK} follows from the channel input constraint \eqref{eq:EqualCov}; \eqref{eq:secondtIneq_Jensen} is due to the concavity of the log-det function and Jensen's inequality; \eqref{eq:secondtIneq_Jensen} follows due to the definition of $\mathbf{\bar{B}}_{k}$; \eqref{eq:secondtIneq_4} follows due to Lemma \ref{lem:AB}, since $\sum_{k\in\mathcal{S}^{c}}\mathbf{H}_{k,\mathcal{T}}^{H}
\bar{\mathbf{B}}_{k}
\mathbf{H}_{k,\mathcal{T}}$ is positive-definite and
$\tilde{\mathbf{K}}_{\mathcal{T}} \preceq \mathbf{K}_{\mathcal{T}}$.

This shows that  $\mc C_{\text{G}}(C_{\mc K}) \subseteq \mc C_{\text{G}}^{\text{no-ts}}(C_{\mc K})$. The proof is completed by noting that  $\mc C_{\text{G}}^{\text{no-ts}}(C_{\mc K}) \subseteq \mc C_{\text{G}}(C_{\mc K})$, and therefore we have $\mc C_{\text{G}}^{\text{no-ts}}(C_{\mc K}) = \mc C_{\text{G}}(C_{\mc K})$. \qed

\section{Proof of Proposition \ref{th:GaussSumCap_highSNR}}\label{app:GaussSumCap_highSNR}

In order to prove Proposition \ref{th:GaussSumCap_highSNR}, we derive an outer bound on the capacity region under time-sharing of Gaussian inputs $\mc C_{\text{G}}(C_{\mc K})$ in Theorem \ref{th:GaussSumCap}, which we denote by $\mc C_{\text{G}}^{\mathrm{out}}(C_{\mc K})$. Then, we derive an inner bound on $\mc C_{\text{G}}^{ \text{no-ts}}(C_{\mc K})$, the  region obtained by setting $Q=\emptyset$ in the region of Theorem~\ref{th:GaussSumCap}, denoted by $\mc C_{\text{G}}^{\text{in,no-ts}}(C_{\mc K})$. We show that if $(R_1,\ldots, R_L)$ lies in the outer bound $\mc C_{\text{G}}^{\mathrm{out}}(C_{\mc K})$, then the rate tuple $((R_1-\Delta_{\epsilon}),\ldots, (R_L-\Delta_{\epsilon}))$ lies in the inner bound $\mc C_{\text{G}}^{\mathrm{in}}(C_{\mc K})$, where $\Delta_{\epsilon}\geq 0$. Finally we show, that in the high SNR regime, i.e., for $\epsilon\rightarrow 0$, the gap vanishes, i.e.,  $\Delta_{\epsilon}\rightarrow 0$.

The  derivations of the bounds in this section use the following equality
\begin{align}
\log(\epsilon^{-1} \lambda + 1 ) = \log (\epsilon^{-1} \lambda) + \log(\epsilon \lambda^{-1}+1) \quad \text{for } \lambda,\epsilon > 0,\label{eq:logEquality}
\end{align}
 and the following upper and lower bound:
\begin{align}
\frac{x}{1+x}\leq \log (1 + x)\leq x \quad \text{for } x>-1.\label{eq:logEquality2}
\end{align}

First, let us define the outer bound $\mc C_{\text{G}}^{\mathrm{out}}(C_{\mc K})$ as the set of rate tuples $(R_1,\ldots, R_L)$ satisfying that for all $ \mathcal{T} \subseteq \mathcal{L}$ and all $\mathcal{S} \subseteq \mathcal{K}$,
{\small
\begin{align}
\sum_{t\in\mathcal{T}}R_{t} &\leq
 \sum_{k\in \mathcal{S}}\left[C_k + \log|\mathbf{I}-\bar{\mathbf{B}}_k| \right]
+ \log
\left|\mathbf{K}_{\mathcal{T}}\sum_{k\in\mathcal{S}^{c}}\mathbf{H}_{k,\mathcal{T}}^{H}
\dv\Sigma_k^{-1/2}\bar{\mathbf{B}}_{k}\dv\Sigma_k^{-1/2}
\mathbf{H}_{k,\mathcal{T}}\right|+
\epsilon \sum_{q\in \mathcal{Q}}p(q)\mathrm{Tr}\{\mathbf{A}_{\mc T,\mc S,q}^{-1}\}\\
&:= f_{\mathrm{out}}(\mc{T}, \mc S)\label{eq:HighSNR_OutRegion},
\end{align}}
for some $\mathbf{0}\preceq \tilde{\mathbf{B}}_{k,q}\preceq \mathbf{I}$ for $k\in \mc K$, $q\in \mc Q$,  and
where we define $\bar{\mathbf{B}}_k:= \sum_{q\in \mathcal{Q}}p(q)\mathbf{\tilde{B}}_{k,q}$ and the $M\times M$ matrix, for $q\in \mc{Q}$ and all $\mc T\subseteq \mc L$ and $\mc S\subseteq \mc K$  given by
\begin{align}
\mathbf{A}_{\mc T, \mc S,q} := \mathbf{K}_{\mathcal{T},q}^{1/2}\sum_{k\in\mathcal{S}^{c}}\mathbf{H}_{k,\mathcal{T}}^{H}\tilde{\dv\Sigma}_k^{-1/2}
\tilde{\mathbf{B}}_{k,q}\tilde{\dv\Sigma}_k^{-1/2}
\mathbf{H}_{k,\mathcal{T}}\mathbf{K}_{\mathcal{T},q}^{1/2}\label{eq:defAq}.
\end{align}
It follows from  \eqref{eq:OBound_highSNR_eigen} below, that we can assume $\mathbf{A}_{\mc T,\mc S,q}\succ \mathbf{0}$ without loss in generality.

Next, we show $\mc C_{\text{G}}(C_{\mc K})\subseteq\mc C_{\text{G}}^{\mathrm{out}}(C_{\mc K})$. Let us define $\tilde{\dv B}_{k,q} := \dv\Sigma_k^{1/2}\dv B_{k,q}\dv\Sigma_k^{1/2}$. Note that $\tilde{\dv B}_{k,q}$ satisfies  $\dv 0 \preceq  \tilde{\dv B}_{q,k}\preceq \dv I$ for $q\in \mc Q$ and $k \in \mc K$.
We have, from Theorem \ref{th:GaussSumCap},
\begin{align}
\mathrm{E}_{Q}\left[\log\frac{|\mathbf{\Sigma}_k^{-1}|}{|\mathbf{\Sigma}^{-1}_k-\mathbf{B}_{k,Q}|}\right]
&=\sum_{q\in \mathcal{Q}}p(q)\log\frac{|\mathbf{\Sigma}_{k}^{-1}|}{|\mathbf{\Sigma}_{k}^{-1}-\mathbf{B}_{k,q}|}\\
&=-\sum_{q\in \mathcal{Q}}p(q)\log|\dv I-\tilde{\mathbf{B}}_{k,q}|\geq -\log|\dv I-\bar{\mathbf{B}}_{k}|,\label{eq:logDetProp22}
\end{align}
where \eqref{eq:logDetProp22} follows from the concavity of the log-det function and Jensen's Inequality \cite{Boyd2004}.

On the other hand, we have
\begin{align}
\mathrm{E}_{Q}&\left[ \log \frac{
|\sum_{k\in\mathcal{S}^{c}}\mathbf{H}_{k,\mathcal{T}}^{H}
\mathbf{B}_{k,Q}
\mathbf{H}_{k,\mathcal{T}}+\mathbf{K}^{-1}_{\mathcal{T},Q}|
}{
|\mathbf{K}_{\mathcal{T},Q}^{-1}|,
}\right]\\
&= \sum_{q\in \mathcal{Q}}p(q)
\log
\left|\mathbf{K}_{\mathcal{T},q}\sum_{k\in\mathcal{S}^{c}}\mathbf{H}_{k,\mathcal{T}}^{H}\dv\Sigma_k^{-1/2}
\tilde{\mathbf{B}}_{k,q}\dv\Sigma_k^{-1/2}
\mathbf{H}_{k,\mathcal{T}}+\mathbf{I}\right|\label{eq:OBound_highSNR_def}\\
&= \sum_{q\in \mathcal{Q}}p(q)
\log
\left|\frac{1}{\epsilon}\mathbf{A}_{\mc T, \mc S,q} +\mathbf{I}\right|\label{eq:OBound_highSNR_def2}\\
%
%&= \sum_{q\in \mathcal{Q}}p(q) \sum_{m= 1}^M
%\log
%\left(\frac{1}{\epsilon}\lambda_{m}(\mathbf{A}_{\mc T, \mc S,q} +\mathbf{I})\right)\label{eq:OBound_highSNR_def2}\\
%
&= \sum_{q\in \mathcal{Q}}p(q) \sum_{m= 1}^M
\log
\left(\frac{1}{\epsilon}\lambda_{m}(\mathbf{A}_{\mc T, \mc S,q}) + 1\right)\label{eq:OBound_highSNR_eigen}\\
&\leq
 \sum_{q\in \mathcal{Q}}p(q) \sum_{m= 1}^M\left(
\log
\left(\frac{1}{\epsilon}\lambda_{m}(\mathbf{A}_{\mc T, \mc S,q})\right)+
\left(\frac{\epsilon}{\lambda_{m}(\mathbf{A}_{\mc T, \mc S,q})} \right)
\right))\label{eq:OBound_highSNR_UpBound}\\
&= \sum_{q\in \mathcal{Q}}p(q) \sum_{m= 1}^M
\log
\left(\frac{1}{\epsilon}\lambda_{m}(\mathbf{A}_{\mc T, \mc S,q})\right)+
\epsilon \sum_{q\in \mathcal{Q}}p(q)\mathrm{Tr}\{\mathbf{A}_{\mc T, \mc S,q}^{-1}\}\label{eq:OBound_highSNR_trace}\\
&= \sum_{q\in \mathcal{Q}}p(q)
\log
\left|\frac{1}{\epsilon}\mathbf{K}_{\mathcal{T},q}^{1/2}\sum_{k\in\mathcal{S}^{c}}\mathbf{H}_{k,\mathcal{T}}^{H}\tilde{\dv\Sigma}_k^{-1/2}
\tilde{\mathbf{B}}_{k,q}\tilde{\dv\Sigma}_k^{-1/2}
\mathbf{H}_{k,\mathcal{T}}\mathbf{K}_{\mathcal{T},q}^{1/2}\right|+
\epsilon \sum_{q\in \mathcal{Q}}p(q)\mathrm{Tr}\{\mathbf{A}_{\mc T, \mc S,q}^{-1}\}\label{eq:OBound_highSNR_Jensen}\\
&\leq
\log
\left|\sum_{q\in \mathcal{Q}}p(q) \mathbf{K}_{\mathcal{T},q}\right| + \log\left|\frac{1}{\epsilon}\sum_{k\in\mathcal{S}^{c}}\mathbf{H}_{k,\mathcal{T}}^{H}\tilde{\dv\Sigma}_k^{-1/2}
\bar{\mathbf{B}}_{k}\tilde{\dv\Sigma}_k^{-1/2}
\mathbf{H}_{k,\mathcal{T}}\right|+
\epsilon \sum_{q\in \mathcal{Q}}p(q)\mathrm{Tr}\{\mathbf{A}_{\mc T, \mc S,q}^{-1}\}\label{eq:OBound_highSNR_Jensen}\\
&\leq
\log
\left| \mathbf{K}_{\mathcal{T}}\right| + \log\left|\frac{1}{\epsilon}\sum_{k\in\mathcal{S}^{c}}\mathbf{H}_{k,\mathcal{T}}^{H}\tilde{\dv\Sigma}_k^{-1/2}
\bar{\mathbf{B}}_{k}\tilde{\dv\Sigma}_k^{-1/2}
\mathbf{H}_{k,\mathcal{T}}\right|+
\epsilon \sum_{q\in \mathcal{Q}}p(q)\mathrm{Tr}\{\mathbf{A}_{\mc T, \mc S,q}^{-1}\}\label{eq:OBound_highSNR_PowerConst}
\end{align}
where \eqref{eq:OBound_highSNR_def} follows from the definition of $\tilde{\dv B}_{k,q}$ and since $\dv K_{\mc T,q}$ is definite positive; \eqref{eq:OBound_highSNR_def2} follows from the definition in \eqref{eq:defAq} and
$\dv \Sigma_k = \epsilon \tilde{\dv\Sigma}_k$; \eqref{eq:OBound_highSNR_eigen} is due to $\lambda_m(\dv A_{\mc T,q}+ \dv I) = \lambda_m(\dv A_{\mc T,q})+1$, $m=[1\!:\!M]$; \eqref{eq:OBound_highSNR_UpBound} is due to \eqref{eq:logEquality} and \eqref{eq:logEquality2}; \eqref{eq:OBound_highSNR_trace} is due to $\mathrm{Tr}\{\dv A^{-1}\}= \sum_{i=1}^{M}\lambda^{-1}(\dv A)$ for a $M\times M$ matrix $\dv A$; \eqref{eq:OBound_highSNR_Jensen} is due to Jensen's inequality; and \eqref{eq:OBound_highSNR_PowerConst} is due to the power constraint \eqref{eq:powConst} and  Weyl's inequality \cite{Horn1985}.

Combining \eqref{eq:logDetProp22} and \eqref{eq:OBound_highSNR_PowerConst} with \eqref{eq:GaussSumCap}, we obtain \eqref{eq:HighSNR_OutRegion}, and thus $\mc C_{\text{G}}(C_{\mc K})\subseteq \mc C_{\text{G}}^{\mathrm{out}}(C_{\mc K})$.

Next, let us define the inner bound  $\mc C_{\text{G}}^{\mathrm{in}}(C_{\mc K})$ as the set of rate tuples $(R_1,\ldots, R_L)$ satisfying that for all $ \mathcal{T} \subseteq \mathcal{L}$ and all $\mathcal{S} \subseteq \mathcal{K}$,
{\small
\begin{align}
\sum_{t\in\mathcal{T}}R_{t} &\leq
 \sum_{k\in \mathcal{S}}\left[C_k + \log|\mathbf{I}-\bar{\mathbf{B}}_k| \right]
+ \log
\left|\mathbf{K}_{\mathcal{T}}\sum_{k\in\mathcal{S}^{c}}\mathbf{H}_{k,\mathcal{T}}^{H}
\dv\Sigma_k^{-1/2}\bar{\mathbf{B}}_{k}\dv\Sigma_k^{-1/2}
\mathbf{H}_{k,\mathcal{T}}\right|+
\epsilon\mathrm{Tr}\{(\mathbf{A}_{\mc T,\mc S}+ \epsilon \dv I)^{-1}\}\\
&:= f_{\mathrm{in}}(\mc{T}, \mc S)\label{eq:HighSNR_InRegion},
\end{align}}
for some $\mathbf{0}\preceq \bar{\mathbf{B}}_{k}\preceq \mathbf{I}$ and
where we define, for $q\in \mc Q$ and all $\mc T\subseteq \mc L$ and $\mc S\subseteq \mc K$, the $M\times M$ matrix given by
\begin{align}
\bar{\mathbf{A}}_{\mc T, \mc S} := {\mathbf{K}}_{\mathcal{T}}^{1/2}\sum_{k\in\mathcal{S}^{c}}\mathbf{H}_{k,\mathcal{T}}^{H}\tilde{\dv\Sigma}_k^{-1/2}
\bar{\mathbf{B}}_{k}\tilde{\dv\Sigma}_k^{-1/2}
\mathbf{H}_{k,\mathcal{T}}{\mathbf{K}}_{\mathcal{T}}^{1/2}\label{eq:defAq_Inner}.
\end{align}

Next, we show that $\mc C_{\text{G}}^{\text{in,no-ts}}(C_{\mc K})\subseteq \mc C_{\text{G}}^{\text{no-ts}}(C_{\mc K})$, where $\mc C_{\text{G}}^{\text{no-ts}}(C_{\mc K})$ is given in \eqref{eq:GaussSumCap_EqualK}.  Let us define $\tilde{\dv B}_{k,q} := \dv\Sigma_k^{1/2}\dv B_{k,q}\dv\Sigma_k^{1/2}$. We have from \eqref{eq:GaussSumCap_EqualK},
\begin{align}
\log\frac{|\mathbf{\Sigma}_k^{-1}|}{|\mathbf{\Sigma}^{-1}_k-\mathbf{B}_k|} = -\log|\dv I-\bar{\mathbf{B}}_k|
\end{align}
and
\begin{align}
\log \frac{
|\sum_{k\in\mathcal{S}^{c}}\mathbf{H}_{k,\mathcal{T}}^{H}
\mathbf{B}_{k}
\mathbf{H}_{k,\mathcal{T}}+\mathbf{K}^{-1}_{\mathcal{T}}|
}{
|\mathbf{K}_{\mathcal{T}}^{-1}|,
}
&=
\log
\left|\mathbf{K}_{\mathcal{T}}\sum_{k\in\mathcal{S}^{c}}\mathbf{H}_{k,\mathcal{T}}^{H}\dv\Sigma_k^{-1/2}
\bar{\mathbf{B}}_{k}\dv\Sigma_k^{-1/2}
\mathbf{H}_{k,\mathcal{T}}+\mathbf{I}\right|\label{eq:InBound_highSNR_def}\\
&= \log
\left|\frac{1}{\epsilon}\bar{\mathbf{A}}_{\mc T, \mc S} +\mathbf{I}\right|\label{eq:InBound_highSNR_def2}\\
&=\sum_{m= 1}^M
\log
\left(\frac{1}{\epsilon}\lambda_{m}(\bar{\mathbf{A}}_{\mc T, \mc S}) + 1\right)\label{eq:InBound_highSNR_eigen}\\
&\geq
 \sum_{m= 1}^M\left(
\log
\left(\frac{1}{\epsilon}\lambda_{m}(\bar{\mathbf{A}}_{\mc T, \mc S})\right)+
\left(\frac{\frac{\epsilon}{\lambda_{m}(\bar{\mathbf{A}}_{\mc T, \mc S})}}{1+\frac{\epsilon}{\lambda_{m}(\bar{\mathbf{A}}_{\mc T, \mc S})}} \right)
\right)\label{eq:InBound_highSNR_lowBound}\\
&=
\sum_{m= 1}^M
\log
\left(\frac{1}{\epsilon}\lambda_{m}(\bar{\mathbf{A}}_{\mc T, \mc S})\right)+
\epsilon \sum_{m= 1}^M
\frac{ 1 }{ \lambda_{m}(\bar{\mathbf{A}}_{\mc T, \mc S})+\epsilon }
\label{eq:InBound_highSNR_UpBound}\\
&= \sum_{m= 1}^M
\log
\left(\frac{1}{\epsilon}\lambda_{m}(\bar{\mathbf{A}}_{\mc T, \mc S})\right)+
\epsilon \mathrm{Tr}\{(\bar{\mathbf{A}}_{\mc T, \mc S}+\epsilon \dv I)^{-1}\}\label{eq:InBound_highSNR_trace}
\end{align}
where \eqref{eq:InBound_highSNR_eigen} follows since $\lambda_m(\bar{\mathbf{A}}_{\mc T, \mc S})+ \dv I = \lambda_m(\bar{\mathbf{A}}_{\mc T, \mc S}) + 1$, $m\in[1\!:\!M]$,  \eqref{eq:InBound_highSNR_lowBound} is due to inequalities \eqref{eq:logEquality} and \eqref{eq:logEquality2}. This shows that $\mc C_{\text{G}}^{\text{in,no-ts}}(C_{\mc K})\subseteq \mc C_{\text{G}}^{\text{no-ts}}(C_{\mc K})$.

Now, we show that if $(R_1,\ldots, R_L)\in \mc C_{\text{G}}^{\mathrm{out}}(C_{\mc K})$, then $(R_1-\Delta_{\epsilon},\ldots, R_L-\Delta_{\epsilon})\in \mc C_{\text{G}}^{\text{in,no-ts}}(C_{\mc K})$, where  we define
\begin{align}\label{eq:DeltaOptim}
\Delta_{\epsilon} := \max_{\mc S\subseteq \mc L, \mc T\subseteq \mc K} \Delta_{\epsilon}(\mc T,\mc S),
\end{align}
and
\begin{align}
\Delta_{\epsilon}(\mc T,\mc S):= \frac{\epsilon  \sum_{q\in \mathcal{Q}}p(q)\mathrm{Tr}\{\mathbf{A}_{\mc T, \mc S,q}^{-1}\} - \epsilon \mathrm{Tr}\{(\bar{\mathbf{A}}_{\mc T, \mc S}+\epsilon \dv I)^{-1}\}}{|\mc T|}.
\end{align}

Then, for any rate tuple $(R_1,\ldots, R_L)\in \mc C_{\text{G}}^{\mathrm{out}}(C_{\mc K})$ we have
\begin{align}
\sum_{t\in \mc T}(R_t-\Delta_{\epsilon})&= \sum_{t\in \mc T}R_t-|\mc T|\Delta_{\epsilon}\\
&\leq f_{\mathrm{out}}(\mc S,\mc T)-|\mc T|\Delta_{\epsilon}\label{eq:InOuterBouns}\\
&\leq f_{\mathrm{out}}(\mc S,\mc T)-|\mc T|\Delta_{\epsilon}(\mc T,\mc S)\label{eq:LowBoundDelta}\\
&= f_{\mathrm{out}}(\mc S,\mc T) - \left(\epsilon \sum_{q\in \mathcal{Q}}p(q)\mathrm{Tr}\{\mathbf{A}_{\mc T, \mc S,q}^{-1}\} - \epsilon \mathrm{Tr}\{(\bar{\mathbf{A}}_{\mc T, \mc S}+\epsilon \dv I)^{-1}\}\right)\\
&= f_{\mathrm{in}}(\mc S,\mc T),
\end{align}
where \eqref{eq:InOuterBouns} follows since $(R_1,\ldots, R_L)\in \mc C_{\text{G}}^{\mathrm{out}}(C_{\mc K})$; \eqref{eq:LowBoundDelta} follows since $\Delta_{\epsilon}\geq \Delta_{\epsilon}(\mc T, \mc S)$ for all $\mc T\subseteq L$ and $\mc S\subseteq K$ due to its definition in \eqref{eq:DeltaOptim}.
This shows that  $(R_1-\Delta_{\epsilon},\ldots, R_L-\Delta_{\epsilon})\in \mc C_{\text{G}}^{\text{in,no-ts}}(C_{\mc K})$.

Next, we show that in the high SNR regime, i.e., $\epsilon \rightarrow 0$,  we have $\Delta_{\epsilon}\rightarrow 0$. We have
\begin{align}
\lim_{\epsilon \rightarrow 0}\Delta_{\epsilon} &= \lim_{\epsilon \rightarrow 0}\max_{\mc S\subseteq \mc L, \mc T\subseteq \mc K} \Delta_{\epsilon}(\mc T,\mc S)\\
&\leq  \lim_{\epsilon \rightarrow 0}\max_{\mc S\subseteq \mc L, \mc T\subseteq \mc K}[\epsilon \cdot \sum_{q\in \mathcal{Q}}p(q)\mathrm{Tr}\{\mathbf{A}_{\mc T, \mc S,q}^{-1}\}] \label{eq:boundDeltastar}\\
&= \lim_{\epsilon \rightarrow 0}~\epsilon \cdot \max_{\mc S\subseteq \mc L, \mc T\subseteq \mc K}[ \sum_{q\in \mathcal{Q}}p(q)\mathrm{Tr}\{\mathbf{A}_{\mc T, \mc S,q}^{-1}\}] \\
&=0,\label{eq:boundDeltastar2}
\end{align}
where \eqref{eq:boundDeltastar} follows since $\mathrm{Tr}\{(\bar{\mathbf{A}}_{\mc T, \mc S}+\epsilon \dv I)^{-1}\}\geq 0$ since $\bar{\mathbf{A}}_{\mc T, \mc S}+\epsilon \dv I\succ \dv 0$; and \eqref{eq:boundDeltastar2} follows since $0 \leq \mathrm{Tr}\{\mathbf{A}_{\mc T, \mc S,q}^{-1}\}]< \infty$ since $\mathbf{A}_{\mc T, \mc S,q}\succ \dv 0$ and it is independent of $\epsilon$.

This completes the proof of Proposition \ref{th:GaussSumCap_highSNR}.\qed

\end{appendices}

\bibliographystyle{ieeetran}
\bibliography{ref}

% APPENDICES TO BE RELEGATED TO A SEPARATE FILE

\end{document}